\documentclass[aps,notitlepage,nofootinbib,
twocolumn,
longbibliography,superscriptaddress,10pt]{revtex4-1}

\usepackage{graphicx}  
\usepackage{dcolumn}   
\usepackage{bm}        
\usepackage{verbatim}   
\usepackage{bbold}
\usepackage{color}
\usepackage{xcolor}
\usepackage{amsthm}
\usepackage{amsmath}
\usepackage{ulem}

\usepackage{amsmath}
\usepackage{amsthm, amssymb}
\usepackage{slashed} 
\usepackage{graphicx}
\usepackage{xcolor}
\usepackage{tikz}
\usetikzlibrary{calc,fadings,decorations.pathreplacing,shapes,shapes.multipart,arrows,shapes.misc,intersections,positioning,patterns}
\usepackage{bm}
\usepackage[colorlinks=true,citecolor=blue,linkcolor=magenta]{hyperref}
\usepackage{dsfont}
\usepackage{changepage}
\usepackage{array}

\newcommand{\bit}{\begin{itemize}}
\newcommand{\eit}{\end{itemize}}

\newcommand{\f}{\frac}
\renewcommand{\>}{\right\rangle}
\newcommand{\<}{\left\langle}
\newcommand{\ba}{\begin{align}}
\newcommand{\ea}{\end{align}}
\newcommand{\be}{\begin{equation}}
\newcommand{\ee}{\end{equation}}
\newcommand{\bi}{\begin{itemize}}
\newcommand{\ei}{\end{itemize}}
\newcommand{\lf}{\left(}
\newcommand{\ri}{\right)}
\newcommand{\dd}{\mathrm{d}}

\newcommand{\Tr}{\operatorname{Tr}}

\def\a{\alpha}
\def\b{\beta}
\def\d{\delta}
\def\D{\Delta}
\def\g{\gamma}

\def\s{\sigma}

\def\be{\begin{equation}}
\def\ee{\end{equation}}
\def\bea{\begin{eqnarray}}
\def\eea{\end{eqnarray}}
\def\>{\rangle}
\def\<{\langle}
\newcommand{\bra}[1]{\langle#1|}
\newcommand{\ket}[1]{|#1\rangle}

\newtheorem{theorem}{Theorem}
\newtheorem{corollary}{Corollary}[theorem]

\theoremstyle{definition}
\newtheorem{definition}{Definition}

\theoremstyle{remark}
\newtheorem*{remark}{Remark}

\begin{document}

\title{Quantum criticality of loops with topologically constrained dynamics}

\author{Zhehao Dai}
\affiliation{
   Department of Physics,
   Massachusetts Institute of Technology,
   Cambridge, MA 02139, USA}
\author{Adam Nahum}
\affiliation{
Theoretical Physics, University of Oxford, Parks Road, Oxford OX1 3PU, United Kingdom}
\affiliation{
   Department of Physics,
   Massachusetts Institute of Technology,
   Cambridge, MA 02139, USA}

\date{\today}

\begin{abstract}
\noindent
Quantum fluctuating loops in 2+1 dimensions give gapless many-body states that are beyond current field theory techniques. Microscopically, these loops can be domain walls between up and down spins, or chains of flipped spins similar to those in the toric code. The key feature of their dynamics is that reconnection of a pair of strands is forbidden. This happens 
at previously-studied multi-critical points between topologically nontrivial phases. 
We show that this topologically constrained dynamics leads to universality classes with unusual scaling properties.
For example, scaling operators at these fixed points are classified by topology, and not only by symmetry.
We introduce the concept of the topological operator classification, 
provide universal scaling forms for correlation functions, and analytical and numerical results for critical exponents. We use an exact correspondence between the imaginary-time dynamics of the 2+1D quantum models 
and a classical Markovian dynamics for 2D classical loop models with a nonlocal Boltzmann weight (for which we also provide scaling results).
We comment on open questions and 
generalizations of the models discussed for both quantum criticality and classical Markov processes.
\end{abstract}

\maketitle

\section{Introduction} \label{Sec: introduction}

Much of quantum criticality can be understood in terms of long-wavelength fluctuations of quantum fields.
Lagrangian field theory is an organizing scheme for a wide range of critical states \cite{sachdev2011quantum,
fisher1974renormalization},
and often gives either exact results or powerful approximation schemes for universal quantities.
But there are some critical states, 
realizable in relatively simple lattice Hamiltonians, 
for which we so far lack any useful continuum description.

One set of examples is a family of 2+1--dimensional quantum-critical models for fluctuating loops introduced in Refs.~\cite{freedman2004class,freedman2005line}.
These models describe multicritical points, reachable by tuning several parameters, in local spin systems.
Microscopically the loops may arise as chains of flipped spin--1/2s (like those in the toric code~\cite{kitaev2003fault}) or alternately as domain walls in an Ising-like order parameter, and the ground state is  simply a superposition of  loop configurations on the 2D plane. The feature of the quantum loop models of Refs.~\cite{freedman2004class,freedman2005line} that makes them gapless is a `topological constraint' on the dynamics of loops. This constraint forbids events in which loops change their connectivity.
The constraint leads to the slow dynamics characteristic of a quantum critical point~\cite{freedman2005line,freedman2008lieb,troyer2008local}: see the cartoon in Fig.~\ref{Fig:localmove}. 

There is no obvious way to encode this dynamical constraint, which is preserved under the renormalization group flow (as we clarify below), in a Lagrangian description. 
These models therefore lie outside the class that we currently know how to describe using field theory, and it is instructive to map out their scaling structure and low-energy excitations.
This is what we do here.

The quantum loop models yield a two-dimensional space of renormalization group (RG) fixed points parameterized by a loop amplitude $d$.
This is a complex number with $|d|\leq \sqrt 2$: the wavefunction amplitude $\Psi(C)$ for a loop configuration $C$ is proportional to $d^\text{(no. loops in $C$)}$.
In general, scaling dimensions change continuously as a function of $d$. 
When $d=1$ the ground state wavefunction coincides with that of the toric code, or the Ising paramagnet, depending on whether we realize the loops as strings of flipped spins or as domain walls (for now we neglect ground state degeneracies). However the {dynamics} and low energy excitations remain nontrivial because of the constraint: even for $d=1$ the model is gapless.

Surprisingly, we find that many universal properties follow from the topology and geometry of the fluctuating loops.
For example we argue that the topologically constrained dynamics implies a topological classification of local operators, generalizing the classification of operators by symmetry that we have at more familiar fixed points.
We provide scaling forms for correlation functions (resolving an apparent paradox at $d=1$, where all equal-time correlators are trivial) as well as analytical results for correlation functions and critical exponents for arbitrary $d$. We  check these using Monte Carlo simulations.

The quantum loop models have a second special property, distinct from the dynamical constraint,
which is a type of quantum-classical correspondence. Thanks to this correspondence we can obtain analytical results despite the lack of a 2+1D field theory.
 
Ref.~\cite{freedman2005line} showed that ground-state expectation values of operators that are diagonal in the $\sigma^z$ basis map to expectation values in a 2D classical lattice model for fluctuating loops. 
The feature of a mapping to a classical model in the same number of spatial dimensions \cite{henley1997relaxation,henley2004classical,ardonne2004topological,castelnovo2005quantum,isakov2011dynamics} is shared with other models, most famously the Rokhsar-Kivelson model \cite{rokhsar1988superconductivity,moessner2001short,henley2004classical,fradkin2004bipartite} 
for quantum dimers. 
Fendley has also constructed quantum loop models that relate to classical models in a different way, imposing non-orthogonal inner products between configurations in order to obtain nonabelian topological states \cite{fendley2008topological,fendley2013fibonacci,fendley2007quantum,fendley2005realizing}. 

There are two main differences between standard Rokhsar-Kivelson-like models and the models we study here. One is that the 2D classical model has a \textit{nonlocal} Boltzmann weight. But the key difference is the constrained nature of the dynamics, which leads to a correspondence with a topologically constrained Markovian dynamics for classical loops.

We construct a dictionary between quantum and classical observables, and find that the correlation functions of \textit{non}-diagonal operators are nontrivial as a result of the nonlocality of the classical model. 
We map correlation functions of local quantum operators to \textit{non-local}, `geometrical' correlation functions in the classical model. The latter are well understood  
\cite{nienhuis1987coulomb,duplantier1987exact,jacobsen2009conformal}, so we obtain numerous exact critical exponents in the quantum models.

We also generalize the quantum-classical correspondence to dynamical correlation functions in the loop model, using the ideas of Ref.~\cite{henley1997relaxation,henley2004classical,castelnovo2005quantum} on Rokhsar-Kivelson-like models (in the process we also clarify the relation between `frustration free' and `Rokhsar Kivelson' Hamiltonians). This dynamical correspondence allows us to investigate dynamical correlation functions.

The quantum loop models are not Lorentz-invariant, and their dynamical exponent $z$ is not 1. Previously it was believed that $z$ was equal to 2 for these models, but we rule this out by strengthening a previous lower bound $z\geq 2$, obtained via a variational argument in Ref.~\cite{freedman2005line,freedman2008lieb}, to $z\geq 4-d_f$ where $d_f<2$ is the fractal dimension of the loops in the classical ensemble. This bound depends on the loop weight $d$ through $d_f$.

However, we point out that the dynamical exponent $z$ is constant along any line of RG fixed points on general grounds (under mild assumptions), i.e. `superuniversal'.
In conjunction with the previous bound, this result gives $z\geq 2.6\dot{6}$.
We have not succeeded in calculating the exact value of $z$ analytically. However our numerical estimates are consistent with $z=3$ for all of the models we simulate, in accord with the `superuniversality' of this exponent for all the models in this two-dimensional space of universality classes.

The loop models are examples of frustration-free Hamiltonians.
There are numerous critical frustration-free models in the literature (many of them, like the Rokhsar-Kivelson model, studied before the name `frustration-free' was common). 
We provide some general results for the scaling structure of these frustration-free critical points, defining what we call `hidden' scaling operators and giving relations between scaling dimensions.

We have already mentioned the point in the quantum loop model with loop amplitude $d=1$, which has some special properties. 
Another interesting special point is $d=\sqrt 2$ \cite{freedman2004class,freedman2005line}. Here  a perturbation exists (the `Jones-Wenzl projector') that leads to a new universality class, but with the same ground state as the original model.
We show that whenever a frustration-free model possesses such a `ground-state--preserving' relevant perturbation, the dynamical exponent at the new fixed point can be bounded in terms of exponents of the original model: this gives $z\geq 2.5$ for the model with the Jones-Wenzl projector. We also characterize its operator spectrum.

Models with topologically constrained dynamics appear to be critical for reasons very different to the ones that we are used to. Therefore they may have useful lessons to teach us. 
We give some other constructions of such models (which may be interesting to examine further) and we briefly discuss obstacles to finding field theory descriptions of them.

\tableofcontents

\section{Overview}
\label{sec:overview}

\subsection{General features}
\label{sec:generalfeatures}

\begin{figure}[b]
\begin{center}
\includegraphics[width=0.5\textwidth]{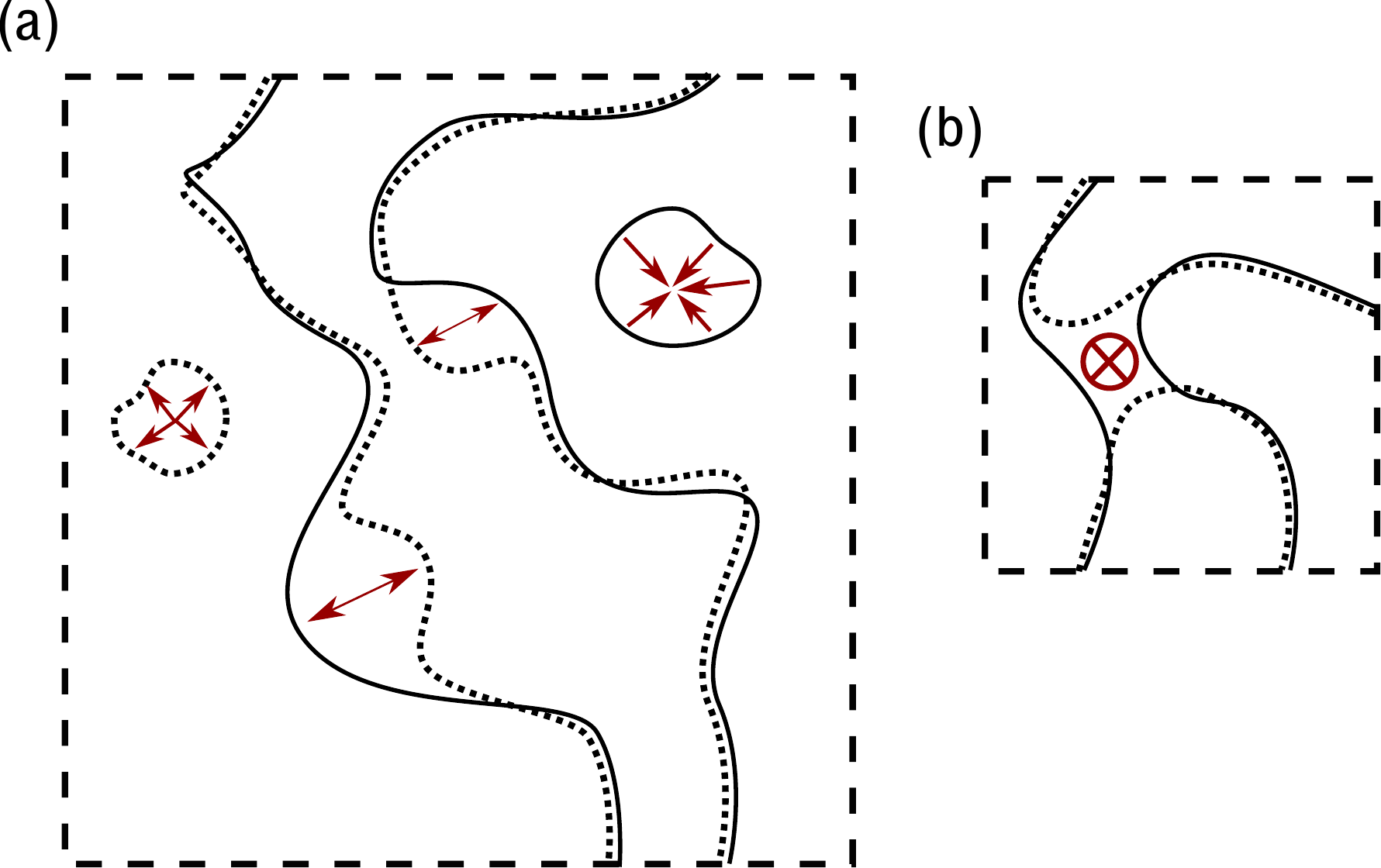}
\caption{(a) Dynamics of loops in continuum space (the figure shows a subregion of the 2D space). Solid/Dotted lines represent possible loop configurations before/after  evolution for short time. (b) A disallowed move. Reconnections of all kinds are forbidden in the time evolution.}
\label{Fig:localmove}
\end{center}
\end{figure}

Refs.~\cite{freedman2004class,freedman2005line} introduced a family of gapless models whose low-energy Hilbert space is spanned by configurations of loops in the plane (Fig.~\ref{Fig:localmove}). 
In the lattice Hamiltonian of Refs.~\cite{freedman2005line} these loops are realized as chains of flipped spins, 
much like the $\mathbb{Z}_2$ flux lines in the toric code,
but by a standard duality we could also think of them as domain walls in a 2D Ising model.
This would change the ground state degeneracy which is a feature of the models (see Sec.~\ref{sec:gsd}), for example, but would not change the key features of the critical scaling.

The ground states of these models are superpositions of loop configurations: schematically, 
\be\label{eq:introducegroundstate}
\ket{\Psi} \propto \sum_{C} d^{|C|} \ket{C},
\ee
where $C$ labels a loop configuration (a set of loops in the plane), $\ket{C}$ is the corresponding basis state, $|C|$ is the number of distinct loops in $C$, and the numerical constant $d \in \mathbb{C}$ is a parameter in the model. When $|d|^2 < 2$ the ground state is scale-invariant in a sense we describe below, and contains large loops with an appreciable amplitude. We will review the models' microscopic Hamiltonians in the next subsection, but first we describe their key features in a continuum language.

It is useful to think of these critical states in terms of their dynamics in the loop basis, i.e. in terms of the  Feynman histories that contribute to the path integral in (for simplicity) imaginary time: see Fig.~\ref{Fig:worldsurfaces1}.
The key feature  distinguishing these dynamics from the gapped dynamics of flux lines in $\mathbb{Z}_2$ gauge theory, or of domain walls in the paramagnetic phase of the Ising model, is that `reconnection' events are not allowed.

\begin{figure}[t]
\begin{center}
\includegraphics[width=0.49\textwidth]{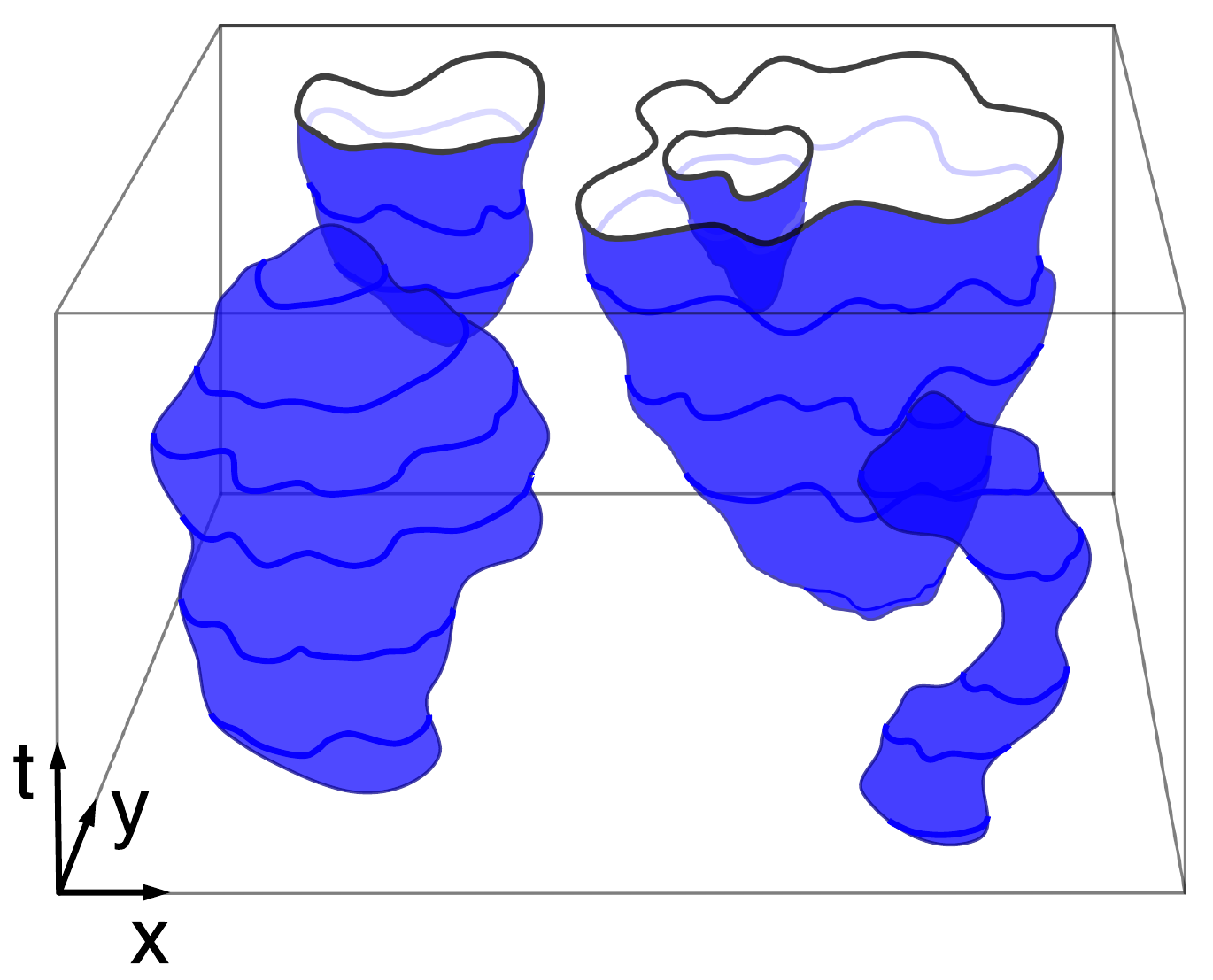}
\caption{Feynman histories of loops.}
\label{Fig:worldsurfaces1}
\end{center}
\end{figure}

Fig.\ref{Fig:localmove} illustrates this at the cartoon level.  
Note firstly that loops do not intersect in these models, so that there is no ambiguity in declaring whether two points are connected by a loop or not.
The quantum dynamics allows two types of events. 
Firstly, loops can fluctuate locally without intersecting or changing their topology. Secondly, loops of microscopic size can be `born' (appear from the vacuum) and `die'. However, processes like the one in Fig.~\ref{Fig:localmove}(b), which changes the connectivity of strands, are forbidden.
The Hamiltonian is unable to break and reconnect loops.

We will refer to this key feature as the dynamical `topological constraint'. In the Hamiltonian of Ref.~\cite{freedman2005line,freedman2008lieb} it is imposed exactly in the ultraviolet: the off-diagonal Hamiltonian matrix elements for reconnections are zero. We will argue in Sec.~\ref{sec:topologicaloperatorclassification} and Sec.~\ref{Sec: operators and correlation functions} that in other models the dynamical topological constraint can emerge in the infrared even without being imposed in the ultraviolet, so long as a sufficient number of relevant couplings are tuned to zero.

Whenever loops of size much greater than the lattice spacing appear in the ground state (as in Eq.~\ref{eq:introducegroundstate} with $|d|^2\leq 2$)  the dynamical constraint implies  quantum critical dynamics:
i.e. a characteristic timescale that diverges with system size, $\tau\sim L^z$, where $z$ is the dynamical exponent of the quantum critical state,  
and corresponding gapless excitations \cite{freedman2008lieb}. 

For a heuristic picture of these slow dynamics we may imagine following the lifespan of a particular loop throughout a Feynman history of the system (we do this numerically in Sec.~\ref{Sec: Excitations and the dynamical exponent}). 
It makes sense to talk about a particular loop persisting through time because loops do not reconnect. 
(By contrast,  a typical Feynman history of the toric code involves loops continually reconnecting, so there is no unique way to identify a loop at a given time with a particular loop at an earlier time). 
Our loop is born at some time in the history, grows to some maximal linear size of order $R$, and dies at a later time. If $R$ is large, the loop's lifespan is necessarily much longer than microscopic timescales, because the loop can only grow by local fluctuations. This lifetime scales as 
\be
\tau\sim R^z
\ee
where $z$ is the dynamical exponent. This exponent also controls the 
 energy scale for the lowest bulk excitations on scale $R$,\footnote{In a system of linear size $L$ there can be anomalously low-lying states with energy $1/L^{w}$ with $w>z$, depending on boundary conditions; we show this in Sec.~\ref{sec:gsd}.} 
\be
\Delta E \sim R^{-z}.
\ee
We address the value of $z$ in Sec~\ref{Sec: Excitations and the dynamical exponent}, both numerically and analytically, by establishing a mapping to the motion of classical loop under a Markov process.

The dynamical constraint is therefore an unusual mechanism for quantum criticality, which has an intuitive explanation in terms of the path integral but which, we emphasize, has no obvious connection to the `usual' mechanism for quantum criticality, namely long-wavelength fluctuations of quantum fields. 

In this paper we will not answer the question of whether a useful field theory  can be found for these states.\footnote{In previous work it was suggested that these models may be described by certain nonabelian gauge theories with dynamical exponent $z=2$~\cite{freedman2005line}. However we prove here that $z > 2$.}  However we will show that they have some very unusual properties that make them quite unlike the critical points that we know how to describe using field theory.

In particular, in Sec.~\ref{sec:topologicaloperatorclassification} and Sec.~\ref{Sec: operators and correlation functions} we introduce the idea of a topological classification of scaling operators at dynamically constrained fixed points. Again this can be understood in terms of the path integral. The Feynman histories contributing to the path integral have no reconnection events. But this can be changed by an insertion of a local operator $\mathcal{O}(x,t)$ that performs a reconnection at a particular spacetime point. The possible reconnection operators have a rich structure, because there are an infinite number of topologically distinct reconnection events that can take place (see Fig.~\ref{Fig: topological type partial order} for examples). 
Some of this topological information about the operator is preserved under the renormalization group (RG), giving us a topological operator classification.
The crucial point is that  the Hamiltonian itself does not reconnect loops. Therefore, loosely speaking, coarse-graining a local operator \textit{cannot}  transform it into an operator that performs more complex reconnections.
We describe this in Sec.~\ref{sec:topologicaloperatorclassification}.

\subsection{Review of lattice Hamiltonians}
\label{sec:latticeH}

Let us now move from the continuum to concrete lattice models for spins on the links of the honeycomb lattice. This subsection reviews the lattice constructions of Refs.~\cite{freedman2004class,freedman2005line,troyer2008local}. We can think of these models as modifications of the toric code so as to forbid reconnection of loops. This modification also allows the nontrivial loop amplitude $d$ to be realized.
Like the toric code, these models involve 3-spin interactions around a vertex of the honeycomb lattice and 6-spin interactions around a hexagon.

We first recall the toric code \cite{kitaev2003fault}, with spins $\sigma_l$ located on the \textit{links} $l$ of the honeycomb lattice. This has the Hamiltonian
\be
\label{eq:toric}
H= H^\text{vertex} + H^\text{flip},
\ee
with ($v$ is a vertex and $p$ is a hexagonal plaquette; 
$l\in  v$ if  the link $l$ includes the vertex $v$, and $l\in p$ if $l$ is in the plaquette $p$):
\ba
H^\text{vertex} & =  -U
\sum_v
\prod_{l \in v}{\sigma_{l}^{z}},
 &
 H^\text{flip} & =   -K 
\sum_p
 \prod_{l\in p}{\sigma_{l}^x}.
\end{align}
`Strings' are made up of links where $\sigma^z = -1$.
The first term is an energy penalty for ends of strings; any vertex $v$ with odd number of neighboring edges occupied by down spins costs energy $2U$. In the ground state there are only closed strings, i.e. loops.

The second term is the kinetic part of the toric code Hamiltonian, which flips all the spins around a plaquette.
Some of the possible `moves' implemented by this term are shown in Fig.~\ref{Fig: local move on lattice}.
The toric code Hamiltonian is solvable because all the terms commute, and the ground states are equal amplitude superpositions of loop configurations.

\begin{figure}[t]
\begin{center}
\includegraphics[width=3in]{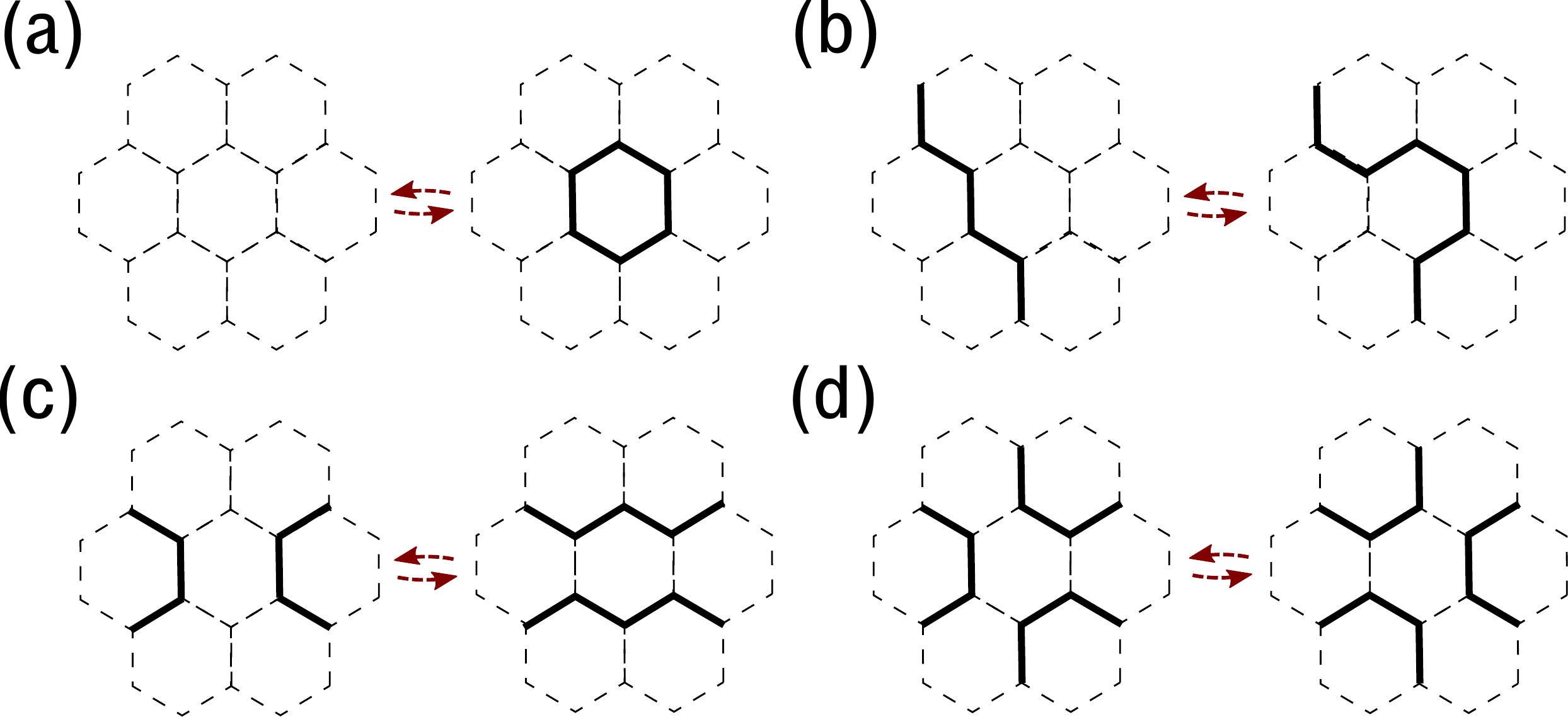}
\caption{Effects of $\prod_{i\in \partial p}{\sigma_{i}^x}$ on loop configurations (this figure follows Ref.~\cite{troyer2008local}). The underlying lattice, shown dashed, has a spin-1/2 degree of freedom on each edge. Down spins on the edges (occupied edges) are shown as solid lines. We consider only closed-loop configurations, but in the illustration, we only show a part of the loop near the plaquette where $\prod_{i\in \partial p}{\sigma_{i}^x}$ acts. (a) $\prod_{i\in \partial p}{\sigma_{i}^x}$ on plaquette $p$ creates a small loop of down spins when the six spins around plaquette $p$ are all up (unoccupied). (b) $\prod_{i\in \partial p}{\sigma_{i}^x}$ moves a loop across the plaquette. (c) $\prod_{i\in \partial p}{\sigma_{i}^x}$ reconnects two loops entering the plaquette. (d) $\prod_{i\in \partial p}{\sigma_{i}^x}$ reconnects three loops entering the plaquette. In the exactly solvable Hamiltonian (Eq.\ref{eq:flippable} and Eq.~\ref{Eq: H gapless loop}), the Hamiltonian is modified to keep only processes (a) and (b) \cite{freedman2004class,freedman2005line,troyer2008local}.}
\label{Fig: local move on lattice}
\end{center}
\end{figure}

The kinetic term $H^\text{flip}$ can be written as a sum over local \textit{moves}: for each move there is a projector \cite{freedman2005line}. A gapless model is obtained simply by dropping the projectors corresponding to the disallowed moves (Fig.~\ref{Fig: local move on lattice}(c-d)) \cite{freedman2005line}.

Let $\ket{a}_p$ be a state of the six spins on the plaquette, and $\ket{\bar a}_p$ the flipped configuration.  Let $\mathcal{P}_p^{(a)}$ be the projector onto the state ${\ket{a}_p - \ket{\bar a}_p}$. Then, dropping a constant, the second term in Eq.~\ref{eq:toric} for a given plaquette is 
\be\label{eq:flip}
H^\text{flip}_p =  K \sum_{a}  \mathcal{P}_p^{(a)}.
\ee
An equal-amplitude superposition of loop configurations is orthogonal to all of these projectors because of the minus sign in the definition of the state ${\ket{a}_p - \ket{\bar a}_p}$. Therefore such a superposition is a ground state.

The gapless models retain from the sum in Eq.~\ref{eq:flip} only the moves that respect the dynamical constraint. (A little thought shows that it is possible to tell whether a move is allowed by looking only at the 6 spins on the hexagon: it is not necessary to examine other spins \cite{freedman2005line}.) This gives a restricted sum, which we denote with a prime, over `flippable' configurations of the hexagon,
\be\label{eq:flippable}
H_p^{\text{flip}'}=  K {\sum_{a}}'  \mathcal{P}_p^{(a)}.
\ee
The complete Hamiltonian is ${H=
H^{\text{vertex}}+
H^{\text{flip}'}}
$.
Ground states of the toric code are still ground states of this model. On the sphere,\footnote{To put  these models on a system of spherical topology we must allow non-hexagonal plaquettes.} both models have the same, unique, groundstate (we neglect ground state degeneracy for other boundary conditions until Sec.~\ref{Sec: Excitations and the dynamical exponent}). However the dynamics generated by Eq.~\ref{eq:flippable} are nontrivial because different plaquette terms no longer commute.

Finally, the amplitude for loop creation and annihilation can be adjusted by modifying the projector which implements the loop birth/death move. Let us write the state with a loop on the hexagon (i.e. all spins down) as $\ket{\circ}$, and the state where the hexagon is empty (all spins up) as $\ket{\phantom{\circ}}$. Then instead of projecting onto ${\ket{\phantom{\circ}}- \ket{\circ}}$ in Eq.~\ref{eq:flippable}, we project onto
\be
\bar{d}\, \ket{\phantom{\circ}}-\ket{\circ}.
\ee
Note that this is orthogonal to ${\ket{\phantom{\circ}}+d\ket{\circ}}$, in which the state with an extra loop has the desired extra factor of $d$.
One can check that the ground states are then of the form given above,
\be\label{eq:psiamplitude}
\ket{\Psi} = \f{1}{Z^{1/2}} \sum_{C} d^{|C|} \ket{C},
\ee
where $C$ must run over a `complete' set of loop configurations: if $C$ is in this set, and $C'$ can be reached from $C$ by allowed moves, so is $C'$. On a manifold of trivial topology (sphere, or disc with appropriate boundary conditions) the ground state is unique, as all loop configurations can be reached from all others by allowed moves. $Z$ above is a normalization constant, which we discuss below.

To be more explicit, let $A_{p}\equiv \prod_{l\in p}{\sigma_{l}^{-}}$ annihilate a small loop around plaquette $p$, and $B_p\equiv \s_1^{+}\s_2^{-}\s_3^{-}\s_4^{-}\s_5^{-}\s_6^{-} + \s_1^{+}\s_2^{+}\s_3^{-}\s_4^{-}\s_5^{-}\s_6^{-} + \cdots$
and its hermitian conjugate $B_p^{\dagger}$ move a loop adjacent to plaquette $p$ across the plaquette:\footnote{We label the spins around plaquette $p$ anticlockwise. 
The terms which move a loop across plaquette $p$  and which contain
 $\s_1^{+}$ are included in $B_p$, and those with  $\s_1^{-}$ in $B_p^{\dagger}$ (this is one way to split up the sum over terms, there are other equivalent rewritings).}
\bea \label{Eq: H gapless loop}
H_p^{\text{flip}'} &=& K_1\sum_{p}\frac12(B_{p}^{\dagger}B_{p}- B_{p} - B_{p}^{\dagger} + B_{p}B_{p}^{\dagger}) \\ \nonumber+
&K_2&\sum_{p}\frac{1}{1+|d|^{2}}(|d|^{2}A_{p}^{\dagger}A_{p}- \bar{d}A_{p} - dA_{p}^{\dagger} + A_{p}A_{p}^{\dagger}).
\eea
Quadratic terms in $A_p$ and $B_p$ are included to make projectors \cite{freedman2005line}.
Note that the coupling $K$ of the individual projectors in Eq.~\ref{eq:flippable} can be separately varied without changing the ground state. In our simulations (Appendix~\ref{Appendix: numerical methods}) we take $K_1/K_2=|d|^2/(1+|d|^2)$, where $K_1$ is the coupling for the projectors that move loops, and $K_2$ is the coupling for the projector that creates/annihilates them.

It is important to note that the  Hamiltonian specified by Eq.~\ref{eq:flippable} is fine-tuned, and that is why it is possible to write ground states explicitly. While it is not a sum of commuting projectors, it is `frustration free', meaning that it is a sum of projectors that can all be simultaneously minimized. Perturbations will generically spoil this property (even if they retain the dynamical constraint exactly). However, only a few such perturbations are expected to be RG--relevant, as we argue in Sec.~\ref{Sec: operators and correlation functions}, so other lattice models, that are not frustration free, can flow to the same RG flxed point. 
(An open question, which we will not discuss here, is whether there are other universality classes that obey the dynamical topological constraint but which cannot be realized in frustration-free models.)

\subsection{Classical mapping}
\label{subsection:classicalmapping}

The quantum loop model is scale-invariant in the IR for the critical range $|d|^2\le 2$. This is most easily seen by a `plasma analogy' in which the modulus square of the quantum wavefunction yields a classical statistical ensemble for fluctuating loops \cite{freedman2005line}. The probability in this ensemble of a loop configuration $C$  is ${P(C)=|d|^{2|C|}/Z}$, and the classical partition function $Z$ (the square of the normalization constant in Eq.~\ref{eq:psiamplitude}) is
\be\label{eq:classicalpartitionfunction}
Z = \sum_C \, |d|^{\,2|C|}.
\ee
This loop ensemble is a well-studied topic in statistical mechanics, sometimes known as the `dense $O(n)$ loop model'  \cite{nienhuis1982exact} (here $n=|d|^2$).
Numerous exact scaling results exist in the literature, which we will review as needed.
In this ensemble, large loops are fractal objects: the length $\ell$ of a  loop, defined as the number of links on it, scales with its linear size $R$ (defined e.g. as the largest distance between two points on the loop) as $\ell\sim R^{d_{f}}$, where the fractal dimension  $d_f$ is a decreasing function of the weight $|d|^2$ \cite{saleur1987exact}. For example, $d_f = 3/2$ for $d=\sqrt{2}$, $d_f = 7/4$ for $d=1$, and $d_f \rightarrow 2$ as $d$ tends to zero (for a review, see Ref. ~\cite{jacobsen2009conformal}).

Equal-time correlation functions of operators that are diagonal in the loop basis $\{ \ket{C} \}$ map straightforwardly to local correlators in the classical loop model. 
In Secs.~\ref{sec:topologicaloperatorclassification} and \ref{Sec: operators and correlation functions}, we shall see that more general correlators allow for richer structures. 

\begin{figure}[t]
\begin{center}
\includegraphics[width=0.3\textwidth]{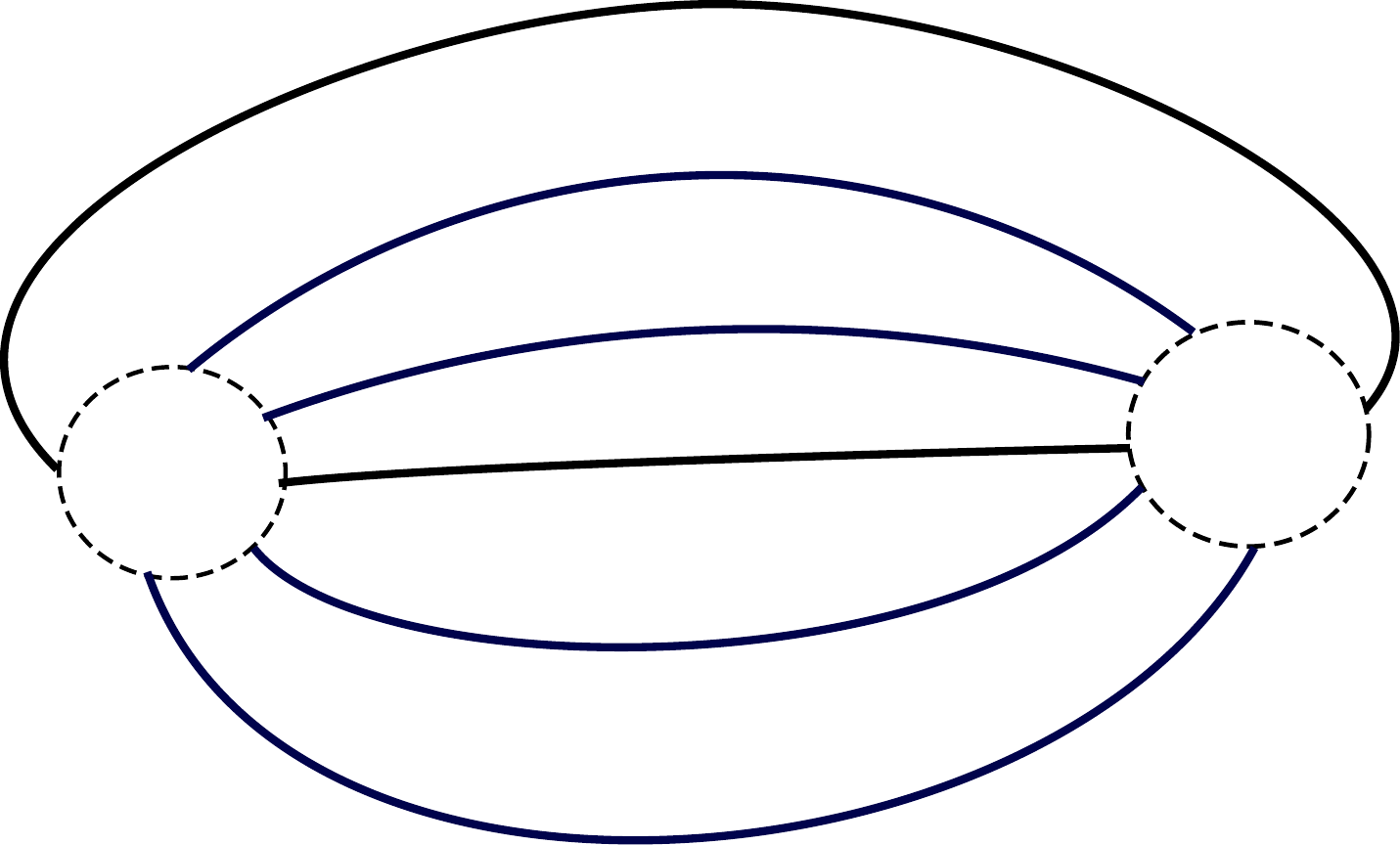}
\caption{Illustration of a 6-strand watermelon correlator. The watermelon correlator is the probability that two microscopic discs are connected by $2k$ strands ($k=3$ in the  figure). In the scaling limit this probability decays as $C_{2k}(\mathbf{r})\sim |\mathbf{r}|^{-2x_{2k}}$ where $x_{2k}$ is given in Eq.~\ref{Eq: x2k}.}
\label{Fig:watermelon}
\end{center}
\end{figure}

First, we show that correlation functions of local but non-diagonal quantum operators map to interesting \textit{non-local} correlation functions in the classical loop model (Sec.~\ref{Sec: operators and correlation functions}). These non-local correlators (known as `watermelon correlators' in the literature on the classical loop model \cite{jacobsen2009conformal}) measure the probabilities that distant points are connected in various ways by large loops (Fig.~\ref{Fig:watermelon}).

Second, we show that temporal correlation functions map to correlation functions in a  classical Markov process for fluctuating loops (Sec.~\ref{Sec: Excitations and the dynamical exponent} and Appendix~\ref{Appendix: Quantum Markovian correspondence}). 
This Markov process is itself an interesting extension of the classical loop models  that could be studied further. Here we give analytical and numerical results for correlation functions and the dynamical exponent.

Our analysis also leads to some general features of `frustration-free' critical points, which we discuss in Sec.~\ref{sec:ffscalingoperators}.

\section{Scaling forms}
\label{sec:scalingforms}

Let us begin by recalling the conventional quantum-critical scaling forms for the equal and non-equal time two-point functions of a local scaling operator $\mathcal{O}$, respectively
\be\label{eq:equaltimescalingform}
\<\mathcal{O}(\mathbf{r},0)\mathcal{O}(0,0)\> = \f{A_\mathcal{O}}{r^{2 x_\mathcal{O}}}
\ee
(we consider an infinite system, and assume that ${\< \mathcal{O}\> =0}$) and
\be\label{eq:generalscalingform}
\<\mathcal{O}(\mathbf{r},t)\mathcal{O}(0,0)\> = \f{1}{r^{2 x_\mathcal{O}}} F_\mathcal{O}\lf \frac{t}{|\mathbf{r}|^z} \ri.
\ee
We would expect similar scaling forms in both real and imaginary time, with different scaling functions $F$, but for simplicity we consider imaginary time throughout this paper. We have suppressed nonuniversal dimensionful constants built from the lattice spacing and the microscopic energy scale (and for simplicity we focus on spin-zero operators for now, so that the dependence on $\mathbf{r}$ is only through $|\mathbf{r}|$).
Two exponents appear here: $z$, the dynamical exponent discussed in the previous section, and the scaling dimension $x_\mathcal{O}$. To have a finite nonzero limit as ${|\mathbf{r}|\rightarrow 0}$ we need $F(u)\sim u^{-2x_\mathcal{O}/z}$, and then
\be\label{eq:twotime}
\<\mathcal{O}(0,t)\mathcal{O}(0,0)\> = \f{B_\mathcal{O}}{t^{2x_{\mathcal{O}}/z}}.
\ee
Finally, if the operator $\mathcal{O}$ is added to the Hamiltonian, then the fact that ${\int \dd^2 \mathbf{r} \, \dd t \, \mathcal{O}(\mathbf{r}, t)}$ scales like ${[\text{length}^{2+z- x_\mathcal{O}}]}$ implies that the RG eigenvalue of this perturbation is:
\be\label{eq:ydef}
y_\mathcal{O} = 2 + z - x_\mathcal{O}.
\ee

At first sight there is a paradox in applying these forms to the quantum loop models. Consider the loop model at the special value $d=1$, where the ground state coincides with that of the toric code (or the Ising paramagnet, depending on the choice of Hilbert space). All equal-time two point functions are strictly zero in this model (this vanishing was emphasised in Ref.~\cite{troyer2008local}).
But at the same time, various operators, most notably the two-loop reconnection operator, are RG relevant in this model (the relevance of two-loop reconnection, which takes us to the toric code phase,  has been checked numerically~\cite{troyer2008local}; we will give exact exponent values below). These operators should therefore have positive RG eigenvalues  $y_\mathcal{O}$ and scaling dimensions ${x_\mathcal{O} < 2+z}$. Why do we not see the corresponding power-law decay in the equal-time correlation function (\ref{eq:equaltimescalingform})? Does this signal a breakdown of renormalization group reasoning?

The resolution of this `paradox' is in fact simple, and there is no need to abandon the scaling forms above.
It is just that, in the model with $d=1$, all of the the scaling functions $F_\mathcal{O}(u)$ vanish as $u\rightarrow 0$.
In Sec.~\ref{Sec: operators and correlation functions} we show that, as a function of $d$,  the amplitudes $A_\mathcal{O}(d)$ vanish at $d=1$, while the scaling dimensions $x_\mathcal{O}(d)$ are finite and continuous there. While these scaling dimensions cannot be extracted from the equal-time two-point function, they can be probed using the temporal correlator (\ref{eq:twotime}) or other more complex correlators.

While $d=1$ is the most extreme case, where \textit{all} equal-time correlation functions are trivial, it turns out that for any $d$ there is a subset of scaling operators whose equal-time correlation functions vanish. We show in Sec.~\ref{sec:ffscalingoperators} below that special `hidden' operators whose equal-time correlators vanish are  a generic feature of frustration-free RG fixed points.

The vanishing of amplitudes $A_\mathcal{O}$ is possible because of the lack of rotational symmetry (let alone conformal invariance) in Euclidean spacetime. In CFT it is common to use the normalization convention $A_\mathcal{O}=1$, but that is not possible here. Another familiar feature of conformal field theory is the orthogonality of two-point functions: covariance under special conformal transformations implies that the two-point function of $\mathcal{O}$ and $\mathcal{O'}$ vanishes if $x_\mathcal{O}\neq x_{\mathcal{O}'}$~\cite{francesco2012conformal}.
Here we cannot assume that in general.
However, for \textit{equal time} correlation functions, we can use the conformal invariance of the classical ensemble described in Sec.~\ref{subsection:classicalmapping} to obtain the same result.

The next section (Sec.~\ref{sec:topologicaloperatorclassification}) describes properties of the operator spectrum that arise from the dynamical topological constraint.
Then Sec.~\ref{sec:ffscalingoperators} discusses features that arise from the frustration-free property and which therefore apply to any critical frustration-free Hamiltonian. 
A key question for the stability of the loop model critical points is the number of relevant (or marginal) scaling operators; we discuss this in Sec.~\ref{subsec: lowlying operator}.

\section{Topological operator classification}
\label{sec:topologicaloperatorclassification}

It is convenient to think of the topological classification of operators in terms of Feynman histories, as mentioned in 
Sec.~\ref{sec:generalfeatures}. 
The dynamical constraint means that these histories contain no reconnection events except those that are put there by operator insertions.
What does this mean for the classification of scaling operators?

To answer this, it is useful to have in mind a renormalization group (coarse-graining) transformation which acts on the spacetime Feynman histories.
Roughly speaking, this transformation eliminates information on lengthscales shorter than $b$ and on timescales shorter than $b^z$ for some dimensionless rescaling factor $b$ (we set microscopic dimensionful constants to 1).
We can think of it as smoothing out the world-surfaces of the loops (Fig.~\ref{Fig:worldsurfaces1}) on these scales.
Small world-surfaces can also disappear below the new UV cutoff, i.e. be eliminated. 
However, our coarse-graining transformation must faithfully preserve the connectivity of the loops of sizes greater than $b$,  because this must remain consistent with their past and future dynamics.\footnote{For example, it is important to distinguish between a configuration in which  two large loops of size $\gg b$ closely approach each other near a point ${\mathbf r}$, and a configuration in which a single large loop closely approaches itself near point ${\mathbf r}$. In the latter scenario, the large loop must have been created in the far past as a single small loop and grown to the current configuration, while in the first scenario, the two loops were created and expanded separately before meeting at the point ${\mathbf r}$.}

\begin{figure}[t]
\begin{center}
\includegraphics[width=0.35\textwidth]{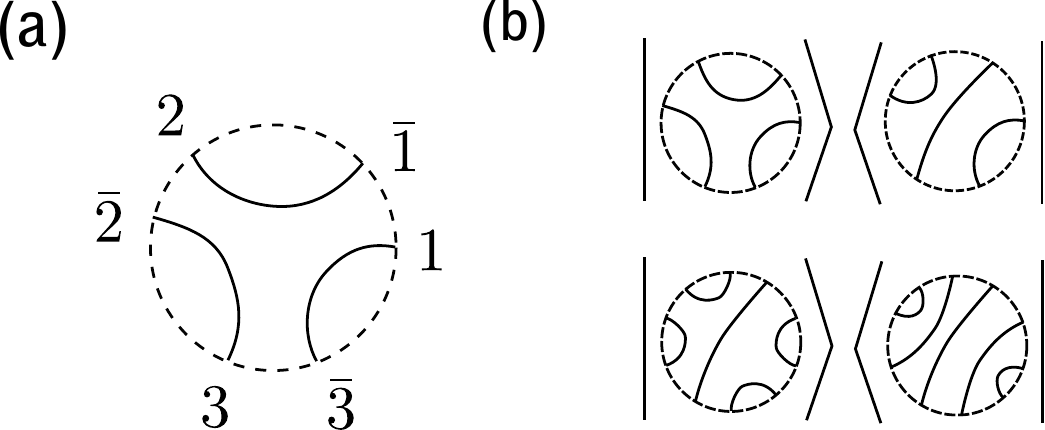}
\caption{Topological types. Left: The ends of the strands are labeled by $1,\bar{1},2,\bar{2},3,\bar{3}$ anticlockwise. This 3-strand configuration may be labeled as $(1\bar{3})(2\bar{1})(3\bar{2})$, or equivalently as the permutation $(123)\rightarrow (312)$. Right, Top: this is a \textit{two}-loop reconnection operator: the lower-right strand is not reconnected, and can be pushed out of the disc. Right, Bottom: a 5-loop reconnection operator. The spectator strand in the middle \textit{cannot} be pushed out.}
\label{Fig: label spectator}
\end{center}
\end{figure}

Let us now consider the action of operators.
Microscopically, local operators can be classified by the way in which they reconnect strands. 

Let the operator act within a spatial disc $D$. Inside the disc $D$, there can be small loops completely contained in $D$, and other loops passing through $D$, with segments inside which connect to the boundary of $D$.
We may first classify the possible \textit{states} within the disc by the number of points where loops cross the boundary of the disc --- which we denote $2k$ since it is always even --- and by the way these points are connected by segments inside. We will call these points on the boundary `endpoints' (of course all loops are closed when we consider the full configuration).
In order to classify states on $D$, we neglect closed loops contained entirely within $D$.
We label the endpoints of internal segments by $1, \bar{1}, 2, \bar{2}, 
\ldots k, \bar{k}$, starting at an arbitrary point and proceeding anticlockwise; then the connection corresponds to a pairing of the points $1, 2, \ldots, k$ with the strands $\bar{1}, \bar 2, \ldots, \bar k$.\footnote{A little thought shows that under the current labeling rule, strands $1,2,\dots,k$ must connect to $\bar{1},\bar{2},\dots,\bar{k}$ instead of themselves, in order to avoid strands crossing. The non-crossing rule means that there are $C_k=(2k)!/(k!(k+1)!)$ possible choices of pairing (a Catalan number). A pairing corresponds to a `non-crossing permutation' in the permutation group $S_k$.}
See  Fig.~\ref{Fig: label spectator}(a) for an example.

Let us label connections (pairings) by $\sigma$, $\sigma'$.
We will use the shorthand $\mathcal{O}=\ket{\sigma}\bra{\sigma'}$ to denote an operator whose only nonzero matrix elements are between `in' states of type $\sigma'$ and `out' states of type $\sigma$; 
see
Fig.~\ref{Fig: label spectator}(b) and
Fig.~\ref{Fig: topological type partial order} for examples.
Such an operator inserts  a reconnection event of a particular type in spacetime.

For a topological classification, we may assume that the pair $(\sigma, \sigma')$ does not contain any removable strands: these are strands that are \textit{both}   paired the same way in $\sigma$ and $\sigma'$, i.e. are not reconnected by the operator; \textit{and} can be pushed out of the disc without being blocked by other strands that \textit{are} reconnected (Fig.~\ref{Fig: label spectator}(b)). 
For example, operators of the form $\ket{\sigma}\bra{\sigma}$ are of the trivial topological type, independent of $\sigma$.
We must also identify pairs $(\sigma, \sigma')$ and $(\tau, \tau')$ that are related simply by a `rotation' of the labels $1, \bar{1}, 2, \bar{2}, 
\ldots k, \bar{k}$, i.e. by changing the starting point for the labelling (note also that rotating an operator does not change its topological type).
We refer to an operator that nontrivially reconnects $2k$ strands as a `$k$--loop reconnection operator'.

\begin{figure}[b]
\begin{center}
\includegraphics[width=0.45\textwidth]{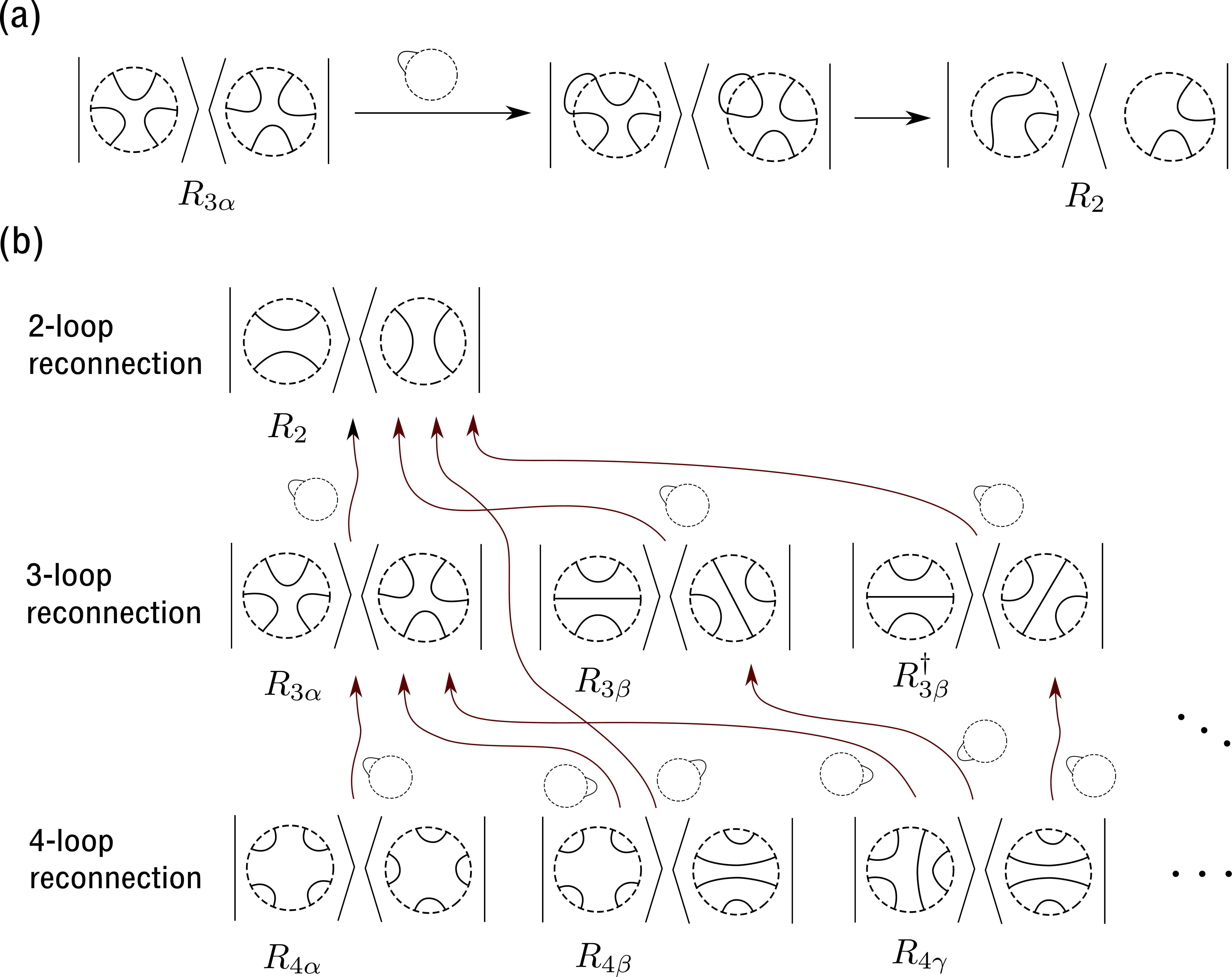}
\caption{Fig.(a): By capping off 2 strands from outside, the 3-loop reconnection $R_{3\alpha}$ becomes $R_2$ (but the reverse \textit{cannot} happen). Thus $R_2$ is an offspring type of $R_{3\alpha}$, and can be generated from $R_{3\alpha}$ under RG. Fig. (b): simple topological types and their partial order. We list all possible 2-loop and 3-loop reconnection types and three of the 4-loop reconnection types.}
\label{Fig: topological type partial order}
\end{center}
\end{figure}

What can happen to a reconnection operator $\ket{\sigma}\bra{\sigma'}$ under coarse-graining?
At first sight we might think that the topological type is simply preserved. 
However a key point is that some of the strands inside the disc that are involved in the reconnection can lie on microscopic loops whose world-surfaces disappear under coarse-graining. Similarly, different strands may be connected by short sections of the same loop, so that after coarse-graining they become a single strand. 
This means that an operator of the form $\ket{\sigma}\bra{\sigma'}$ in general transforms into a sum of operators,  including not only the same topological type but also simpler, `offspring' topological types. These simpler types are obtained by capping off strands in the way shown in Fig.~\ref{Fig: topological type partial order}.
 
There is therefore a partial order on this basis of quantum operators. Since coarse-graining cannot generate nontrivial reconnection events that were not there microscopically, a scale transformation can transform an operator only into operators of the same or lower type. The lowest type is the topologically trivial type which does not reconnect any loop.

As a result we do not expect an operator of a fixed type $\ket{\sigma}\bra{\sigma'}$ to be a scaling operator in general: a scaling operator transforms into itself (times $b^{-x_\mathcal{O}}$) under an RG transformation, whereas $\ket{\sigma}\bra{\sigma'}$ will also generate lower reconnections.
Instead, the RG transformation has a  block-upper-triangular form when restricted to $\ket{\sigma}\bra{\sigma'}$ and its `offspring'. We must diagonalize this matrix in order to construct scaling operators. We assume this can be done and will do it explicitly in several cases.\footnote{ In principle the scale transformation might contain nontrivial Jordan blocks which would imply logarithms in correlation functions \cite{gurarie1993logarithmic,cardy1999logarithmic,vasseur2012logarithmic,cardy2013logarithmic}. We will not address this here in general. The correlation functions we compute explicitly do not contain logarithms.}

At the end of the day, we can still label every scaling operator by its leading topological type $(\sigma, \sigma')$,
but we should think of such an operator as a linear combination of microscopic operators not only of type $(\sigma, \sigma')$ but also of the lower types.

In Sec.~\ref{Sec: operators and correlation functions} and Appendix~\ref{Appendix:topologicaloperatorclassification}, we verify the physical arguments presented above by directly calculating equal-time correlators. 
The loop model has the advantage that equal-time correlators can be mapped to (non-local) 2D classical correlation functions. In Sec.~\ref{Sec: operators and correlation functions} we use this mapping to classify equal-time correlators in the quantum models. 
We will see that in order to construct
a $k$--loop reconnection scaling operator, we have to subtract operators of simpler (`offspring') topological types that connect fewer than $k$ loops.

We find that the equal-time two-point functions of a large class of quantum reconnection operators can be mapped exactly to watermelon correlators in the classical loop ensemble.
The basic watermelon correlators measure the probability that two distant points are connected by a certain number, $2k$, of strands: see the cartoon in Fig.~\ref{Fig:watermelon}. This correlation function has scaling dimension (for a review, see ~\cite{jacobsen2009conformal})
\be\label{Eq: x2k}
x_{2k} = 
\frac{g^2k^2-(1-g)^2}{2g},
\ee
with 
\ba\label{eq:gdef}
g & = \text{arccos}\lf-\f{|d|^2}{2}\ri / \pi,&  \frac{1}{2}& <g<1.
\end{align}
The quantum model has $k$-loop reconnection operators with dimension $x_{2k}$ for every $k>1$. 
(The lowest of the above dimensions, $x_2$, determines the fractal dimension of the loops in the classical ensemble, via $d_f=2-x_2$~\cite{jacobsen2009conformal}. But as far as we are aware, there is no local quantum operator in the present models with dimension $x_2$.\footnote{ A modified quantum model, allowing multiple colours of loops, also has a quantum operator with dimension $x_2$, which simply measures the local colour.}) 

For large $k$, there are many distinct operators with the same scaling dimension $x_{2k}$, which effect topologically distinct reconnections $\ket{\sigma}\bra{\sigma'}$ involving the same number of strands. In addition to the topological label, we will also discuss more conventional symmetry indices of reconnection operators, for example the spin under spatial rotations (Sec.~\ref{Sec: operators and correlation functions}).

What is the relation between the topological operator classification above and conventional classification of operators by symmetry? The two ideas share the logic that constraints in a Hamiltonian that are preserved under RG, either symmetry constraints or topological constraints (both can be exact or emergent) dictate the scaling structure in the IR. 
In this sense, the dynamical topological constraint we discuss here may be viewed as a generalization of symmetry, and the topological operator classification as the analog of classification by symmetry representations. It is also possible that in some models, these two points of view overlap.

\section{Scaling operators for frustration-free Hamiltonians}
\label{sec:ffscalingoperators}

A frustration-free Hamiltonian is at first glance just an infinitely fine-tuned choice among all possible Hamiltonians in a universality class. However, a closer examination shows that the existence of one frustration-free Hamiltonian imposes strong constraints on the scaling structure of the universality class.

Local frustration--free Hamiltonians have the schematic form 
\be\label{eq:ffH}
\mathcal{H} = \sum_i  J_i \sum_{\mathbf r} \mathcal{P}_{i, {\mathbf r}},
\ee
where the index $i$ runs over different types of Hamiltonian term, and $\mathcal{P}_{i, {\mathbf r}}$ is a projection operator.\footnote{We have assumed translational symmetry here, but similar considerations apply without it.} 
Let us write
\be
\mathcal{P}_{i, {\mathbf r}} =
\ket{i}_{D_{\mathbf r}} \bra{i}_{D_{\mathbf r}}
\otimes
\mathbb{1}_{\overline{D_{\mathbf r}}}
\label{Eq: local projector}
\ee
where $\ket{i}_{D_{\mathbf r}}$ is a state defined in some small patch $D_{\mathbf r}$ around ${\mathbf r}$, and  $\mathcal{P}_{i, {\mathbf r}}$ acts as the identity, $\mathbb{1}_{\overline{D_{\mathbf r}}}$, outside this patch. Frustration-freeness means that all the terms can be simultaneously satisfied, i.e. there is at least one ground state $\ket{\Psi}$ that is annihilated by all the projectors,
\be\label{eq:projectorzero}
\mathcal{P}_{i, \mathbf{r}}\ket{\Psi} = 0.
\ee
`Rokhsar-Kivelson'--type Hamiltonians are a subclass of frustration-free Hamiltonians for which there is also a correspondence to a local classical Markov process; this requires an additional condition on the projectors that we provide in Appendix~\ref{Appendix: Quantum Markovian correspondence}.

The structure in Eqs.~\ref{eq:ffH},~\ref{eq:projectorzero} implies that critical frustration-free models will generally have scaling operators whose correlators vanish in the ground state(s). For simplicity consider equal time correlators, and to avoid clutter let us drop the subscripts $D_{\mathbf r}$. 
The frustration-free condition implies that
the reduced density matrix  (of a ground state) for a patch $D_{\mathbf r}$, which is ${\rho=\Tr_{\overline{D_{\mathbf r}}} \ket{\Psi}\bra{\Psi}}$, is orthogonal to the `forbidden' states $\ket{i}$  in the patch for all $i$: 
\ba\label{eq:hiddenstate}
\rho \ket{i} &=0,
&
\bra{i} \rho &=0.
\end{align}
Therefore the equal time correlation function 
\be
\bra{\Psi} (\ldots) \,  \mathcal{O}_{\mathbf r} \ket{\Psi}
\ee
(where `$\ldots$' represents operators
outside the patch $D_{\mathbf{r}}$, hence
commuting with $\mathcal{O}_{\mathbf{r}}$)
vanishes if $\mathcal{O}_{\mathbf r}$ is of the form $\ket{i}\bra{\phi}$  or
$\ket{\phi}\bra{i}$, for any state $\ket{\phi}$ on the patch ${D_{\mathbf r}}$.\footnote{For operators of the former (latter) type, nonequal time correlation functions also vanish if $\mathcal{O}$ is the latest (resp. earliest) time operator in the correlator.} 
We refer to linear combinations of these operators as `hidden' operators. 
More generally, we define a hidden operator as an operator of this form in which $\ket{i}$ is orthogonal to the reduced density matrix (for the appropriate local patch) as in Eq.~\ref{eq:hiddenstate}, regardless of whether the projector onto $\ket{i}$ appears in the Hamiltonian.

A subset of hidden operators correspond to infinitesimal deformations of the Hamiltonian which  transform it into a different frustration-free Hamiltonian. A  deformation in which the state $\ket{i}$ appearing in one of the Hamiltonian's projectors is replaced with ${\ket{i} + \epsilon_i \ket{\phi}_i}$ yields a perturbation of the form 
$\ket{i}\bra{\phi}_i +  \ket{\phi}_i\bra{i}$.
(Whether the total perturbation preserves the frustration free property is a global property, not a property of the local operator.)

Another subset of hidden operators, which we might call `doubly hidden', is made up of those of the form $\ket{f}\bra{f'}$ that act entirely within the subspace that is forbidden in the ground state (i.e. $\rho\ket{f}=\rho\ket{f'}=0$ for the appropriate local patch). This includes the perturbations $\ket{i}\bra{i}$ that are induced by varying the couplings $J_i$ in Eq.~\ref{eq:ffH}. In Sec.~\ref{Sec: Jones Wenzl} we will discuss an example of a nontrivial doubly-hidden operator in the loop model which appears only at $d = \pm\sqrt{2}$ (and will describe a general feature of RG flows induced by doubly-hidden operators, when they are RG-relevant).

Not all frustration-free Hamiltonians admit a classical correspondence: we clarify the conditions under which a classical correspondence exists in Appendix~\ref{Appendix: Quantum Markovian correspondence}. 
But in cases where there is a classical mapping, some of the hidden operators can be related to perturbations of the classical problem. (Recall that $|\Psi(C)|^2$ is the Boltzmann weight in the classical ensemble.)

Take a local observable 
$\mathcal{O}_\text{cl}=\mathcal{O}_{\text{cl},\mathbf{r}}$ at position $\mathbf{r}$
in the classical ensemble, with scaling dimension $x_\text{cl}$.
In general there are at least three separate quantum operators whose scaling dimensions are related to $x_\text{cl}$.
First, there is a diagonal quantum operator $\mathcal{O}_1=\mathcal{O}_{1,{\mathbf r}}$ whose matrix elements are given by $\mathcal{O}_\text{cl}$, and whose equal-time correlators are equal to those  of $\mathcal{O}_\text{cl}$ in the classical ensemble. Evidently,
\be
x_1 = x_\text{cl}.
\ee

Next, imagine perturbing the phase or the amplitude of the ground state wavefunction with $\mathcal{O}_\text{cl}$:
\ba\label{eq:wavefunctiondeformation}
|\tilde{\Psi}\> & = e^{i\lambda \mathcal{O}_{1}}|\Psi\>,
&&\text{or}
&
|\tilde{\Psi}\> &= e^{\lambda \mathcal{O}_{1}}|\Psi\>,
\end{align}
where $\lambda$ is a small real number. What perturbations of the \textit{Hamiltonian} induce these perturbations of the ground state?
The first one is nothing but a local unitary transformation of the ground state, corresponding to the infinitesimal change
\ba
\mathcal{H} &\rightarrow \mathcal{H} + \lambda \mathcal{O}_2,&
&\text{where} 
&
\mathcal{O}_2 &= i[\mathcal{O}_{1}, \mathcal{H}] =  - \dot{\mathcal{O}}_{1}.
\end{align} 
The operator $\mathcal{O}_2$ has scaling dimension
\bea
x_2 = z + x_\text{cl}.
\eea

For a generic system, the second perturbation in Eq.~\ref{eq:wavefunctiondeformation} 
does not correspond to any local perturbation of the Hamiltonian. However for the present class of models, with a classical correspondence, it does.
The energy $E_C$ in the classical partition function (given for a configuration $C$ by $e^{-E_C} = |\Psi(C)|^2$)  is perturbed by $-2 \lambda \mathcal{O}_\text{cl}$.
The RG eigenvalue of this perturbation in the classical problem is $y = d - x_\text{cl}$,
where $d$ is the spatial dimension.
We expect that there is a corresponding local operator $\mathcal{O}_3$ (preserving frustration-freeness)
with which we can perturb the quantum model
so as to yield a correspondence with the perturbed classical model. 
This operator must have the same RG eigenvalue $y$ as in the classical problem. This means that it also has the scaling dimension
\be\label{eq:x3}
x_3 = z + x_\text{cl}.
\ee

What is the algebraic expression for $\mathcal{O}_3$? Consider a spatial patch around $\mathbf{r}$ with the local Hamiltonian
\bea
\mathcal{H}_\mathbf{r} = \sum_i  J_i \mathcal{P}_{i, {\mathbf r}}.
\eea 
We choose the patch large enough so that $\mathcal{O}_\text{cl}$ depends only on the configuration within the patch (i.e. so that $\mathcal{O}_1$ is supported on the patch), and also large enough so that all the terms in $H$ that do not commute with $\mathcal{O}_1$ are included. 
It is convenient to take the `forbidden' states $\ket{i}$ that are penalized by the projectors $\mathcal{P}_{i,{\mathbf r}}$ to be states on the full patch rather than on subregions of it: this can be done by a simple rewriting (preserving the form above but increasing the number of terms; see Appendix ~\ref{Appendix: O3 patches}).
 
The ground state is orthogonal to all the forbidden states $\ket{i}$.
In order to implement the desired perturbation to it (hence to the projector in Eq.~\ref{eq:ffH}), we make the invertible transformation ${|i\>\rightarrow e^{-\lambda \mathcal{O}_\mathbf{r}}|i\>}$: this preserves the frustration-freeness of the Hamiltonian and induces the change
\ba
\mathcal{H}_\mathbf{r}&\rightarrow \mathcal{H}_\mathbf{r} -\lambda\mathcal{O}_{3},
&
&\text{with}
&
\mathcal{O}_{3} &= \{\mathcal{O}_{1}, \mathcal{H}_\mathbf{r}\}\label{Eq: O3 expression}.
\end{align}
Note that $\mathcal{O}_3=\mathcal{O}_{3,{\mathbf r}}$ is a local operator, since $O_1$ and $\mathcal{H}_{\mathbf r}$ are both supported in the local patch.

In fact, the relation between the scaling dimensions of $O_1$ and $O_3$ above is an example of a more general relationship for  frustration-free Hamiltionians, including those without a classical correspondence, which constrains scaling dimensions of certain operators.
In general, multiplying two local operators with dimensions $x$ and $x'$ does not simply give an operator with dimension $x+x'$.
However in a frustration-free model, multiplying a local operator with the  local Hamiltonian density (summed over an appropriate local patch) 
simply increases the scaling dimension by $z$, which is the dimension of the Hamiltonian, as in Eq.~\ref{eq:x3}. We show this by directly computing the 2-point function: 
\ba \notag
\<(\mathcal{O}_{\mathbf{r}}\mathcal{H}_\mathbf{r})(t)  (\mathcal{O}_{\mathbf{r}'}\mathcal{H}_{\mathbf{r}'})^\dag(0) \> 
& = \<\Psi|\mathcal{O}_{\mathbf{r}}\mathcal{H}_{\mathbf{r}}e^{-\mathcal{H}_{\mathbf{r}}t}\mathcal{H}_{\mathbf{r}'}\mathcal{O}_{\mathbf{r}'}^\dag|\Psi\>\\ \notag
& = \<\Psi|\mathcal{O}_{\mathbf{r}}\mathcal{H}e^{-\mathcal{H}t}\mathcal{H}\mathcal{O}_{\mathbf{r}'}^\dag|\Psi\>\\
&=\f{d^2}{dt^2}\<\mathcal{O}_{\mathbf{r}}(t)\mathcal{O}_{\mathbf{r}'}^\dag(0)\>.
\end{align}
We used frustration-freeness in going from the first line to the second: 
the difference between $\mathcal{H}_{\mathbf r}$ and $\mathcal{H}$ consists of terms that commute with $\mathcal{O}_{\mathbf r}$ and individually annihilate the ground state.
Comparing with Eq.~\ref{eq:twotime} gives the result for the scaling dimension.

Operators $\mathcal{O}_2$ and $\mathcal{O}_3$ are examples of hidden operators. One of these, $\mathcal{O}_2$, is however a redundant perturbation in the RG sense \cite{Wegner:1976bn}, as the first part of Eq.~\ref{eq:wavefunctiondeformation} shows that it can be absorbed into a quasilocal unitary basis change.
In the loop models, we cannot obtain the scaling dimensions of hidden operators by directly calculating their \textit{equal-time} two-point functions, since these vanish; however, we can obtain the scaling of at least some hidden operators by the logic above.

If the classical model has an exactly marginal \textit{local} perturbation $\mathcal{O}_\text{cl}$ with dimension $x_\text{cl}=d$, this will give rise to a quantum operator $\mathcal{O}_1$ with dimension $x_1=d$ and two exactly marginal quantum operators of dimension $x_2 = x_3 =z+d$, as above.

It is worth noting that in the quantum loop model there are marginal operators for any $d$ that, unlike the above, do not have a counterpart \textit{local} operator of lower dimension. These are the operators that change the real and imaginary parts of the loop amplitude $d$. These have the scaling dimension $x=z+d$, like the operators $\mathcal{O}_2$ and $\mathcal{O}_3$ in the above situation. However, in this case  the analogues of the operators $\sum_{\mathbf r}\mathcal{O}_{\text{cl},{\mathbf r}}$ and $\sum_{\mathbf r} \mathcal{O}_{1,{\mathbf r}}$ (in the classical and quantum models respectively) are nonlocal operators that count the number of loops.
Note that changing the loop amplitude is a nonlocal perturbation in the classical model, but can be achieved by a local perturbation of the quantum Hamiltonian.\footnote{
For this reason, the  loop weight $|d|^2$ in the classical model does not flow under RG. There is no such protection in the quantum model.}

\section{Excitations and dynamics}\label{Sec: Excitations and the dynamical exponent}

The existence of gapless excitations is equivalent to the presence of slow dynamics in either real or imaginary time.
In the present model, thinking about time evolution of large loops with the no-reconnection constraint provides useful insights into the excited states. 
With the topological constraint, what could otherwise be done by a single reconnection of large loops now requires gradually shrinking loops, annihilating loops, creating new loops, and growing them to the new position. 
When we rely on such extended paths through state space to connect two configurations, then low-lying excitations can be made by introducing smooth phase twists into the wavefunction \cite{freedman2008lieb}, as we discuss below.

We note that the dynamical exponent, and in fact the entire spectrum, is necessarily independent of the argument of the complex loop amplitude $d$ and can depend at most on its modulus $|d|$. 
This follows from the fact that the phase of $d$ can be rotated by $\theta$ by conjugating the Hamiltonian with a unitary transformation $U^{\theta}$.
This transformation is diagonal in the loop basis and acts by $U^{\theta}\ket{C}=e^{i\theta |C|} \ket{C}$, where again $|C|$ is the number of loops in configuration $C$. This transformation is non-local, but preserves the locality of the Hamiltonian. Since it is unitary, it preserves the spectrum and therefore the value of $z$, so this cannot depend on $\arg d$. 

In this section, we first introduce a correspondence between the frustration-free quantum Hamiltonian and a classical Markov process for loops in two spatial dimensions. 
This correspondence allows us think about of the motion of loops in a classical language, and to compute the dynamical exponents and various temporal correlation functions numerically with a classical Monte Carlo simulation. 

Next, we prove a new analytical bound on the dynamical exponent,  strengthening the previous bound $z\geq 2$ \cite{freedman2008lieb}.
We show that $z\ge 4-d_f$, where
\be
d_f = 1 + \frac{\pi}{2\text{arccos}(-|d|^2/2)}
\ee
is the fractal dimension of loops discussed in Sec.~\ref{subsection:classicalmapping}.
Our result rules out the value $z=2$ that was previously believed to be exact \cite{freedman2005line,freedman2008lieb}.

We then report numerical results which show that the actual value of $z$ is close to 3, and apparently independent of the weight $|d|$ (Sec.~\ref{sec:znumerics}). 
In the following section (Sec.~\ref{sec:superuniversality}) we will explain this `superuniversality' of $z$ as a generic feature of RG fixed lines/surfaces. 

Taking the fact that $z$ must be $d$-independent  into account allows the bound  ${z\geq 4-d_f(d)}$ to be strengthened, by using the smallest  value of $d_f$ that is attained anywhere on the critical surface. 
At first glance this is $d_f=3/2$ which is attained at $n=2$, giving $z\geq 2.5$ for all $d$.
But in fact, as we explain in Sec.~\ref{subsec:dilutecriticalpoint}, the critical surface folds over at $|d|^2=2$ onto another `sheet' for the same range $0<|d|^2<2$, with smaller fractal dimension for a given $d$. 
This gives the analytical bound $z\geq 2.6\dot 6$.

At the end of the present section we describe anomalously low-lying states that appear for the loop model on topologically nontrivial manifolds. (Sec.~\ref{sec:gsd}).

\subsection{Mapping to a Markov process}
\label{sec:markovprocess}

A classical Markov process is described by a master equation
\be
\label{Eq: Master Equation}
\frac{dp_{C}(t)}{d t}=\sum_{C'} W_{CC'}p_{C'}(t)
\ee
where $C,C'$ label classical configurations and 
$p_{C}(t)$ is the evolving probability distribution.
The off-diagonal elements $W_{CC'}$ are transition rates from $C'$ to $C$, and in order to conserve probability the diagonal elements are $W_{CC}= -\sum_{C'\neq C}W_{C'C}$. 

The transition matrix is in general non-Hermitian, unlike the Hamiltonian appearing in a Schrodinger equation.
However  a large class of frustration-free quantum models \cite{henley2004classical}, including the loop model, are related to a classical system by a diagonal similarity transformation built from the wavefunction:
\bea 
S_{CC'} 
= \Psi^*(C) \delta_{C,C'} 
=
\frac{
\bar{d}^{|C|}\delta_{C,C'}
}
{
\sqrt{Z}
}.
\eea 
Multiplying by $S$ maps the ground state wavefunction ${\Psi(C) = d^{|C|}/\sqrt{Z}}$ to the equilibrium probability distribution of the classical loop model, ${p(C) = |\Psi(C)|^2 = |d|^{2|C|}/Z}$. The nontrivial observation is that the frustration-free Hamiltonian is mapped to (minus) the transition matrix of a local Markov process for the loops:
\bea 
{H}\ \longrightarrow \ 
-W\equiv S {H}S^{-1}.
\eea 
This Markov process allows the same transitions as the quantum Hamiltonian (motion of a loop and creation/annihilation of a small loop) and obeys the same no-reconnection constraint.
It is worth noting that while the equilibrium state of the Markov process contains the nonlocal factor $|d|^{2|C|}$, its dynamics is entirely local.

For a simple example of the mapping, take the two configurations for a single hexagon shown in Fig.~\ref{Fig: local move on lattice}(a): one empty (`${\ket{\phantom{\circ}}}$')and one occupied with a small loop (`${\ket{\circ}}$'). Restricted to these two states, the quantum Hamiltonian and the classical transition matrix have the forms
\ba
{H} & \propto \left( \begin{array}{cc}
|d|^2 & -\bar{d}\\
-d & 1
\end{array} \right),
&
W & \propto \left( \begin{array}{cc}
-|d|^2 & 1\\
|d|^2 & -1
\end{array} \right).
\label{Eq:2by2projector}
\end{align}
The classical transition matrix $W$ has positive real transition rates, and the steady state (with zero-eigenvalue) is ${\ket{\phantom{\circ}} + |d|^2\ket{\circ}}$, as expected. Each column of $W$ adds up to zero, ensuring conservation of probability in the classical dynamics.

For the present loop model, the similarity transformation gives a local Markov process.  Refs.~\cite{henley1997relaxation,henley2004classical,castelnovo2005quantum} point out that such a mapping always exists for a Hamiltonian decomposable in terms of ${2\times2}$ projectors. However, a classical correspondence is not guaranteed for a generic frustration-free Hamiltonian.
Ref.~\cite{velenich2010string} discussed possible obstacles to constructing such a mapping.
For example, for the topological models constructed in Ref.~\cite{levin2005string}, those exhibiting abelian topological order map to local Markov processes, but the non-abelian models do not.
We explicitly write down  sufficient conditions under which a quantum frustration-free Hamiltonian maps to a local Markov process, and discuss general aspects of the map, in  Appendix~\ref{Appendix: Quantum Markovian correspondence}.

The classical correspondence greatly simplifies the computation of temporal correlators.  Refs.~\cite{henley1997relaxation,henley2004classical,castelnovo2005quantum} show that the imaginary time quantum correlator of \textit{diagonal} operators maps to the two-point correlation function in the classical Markov process,
\ba
\<\mathcal{O}_{1}(t)\mathcal{O}_2(0)\> & =
\< \mathcal{O}_1(t)\mathcal{O}_2(0)\>_\text{cl}.
\end{align}
Explicitly, 
\ba
\< \mathcal{O}_1(t)\mathcal{O}_2(0)\>_\text{cl}
&= \sum_{C,C'}\mathcal{O}_{1}(C')p_{C',C}(t)
\mathcal{O}_{2}(C)
p_C,
\ \ \label{Eq: correlator Markov}
\end{align}
where $p_{C',C}(t)$ is the classical  probability to go from state $C$ to $C'$ in time $t$. 
In Appendix ~\ref{Appendix: Quantum Markovian correspondence}, we generalize this correspondence to off-diagonal operators. 

These relationships mean that quantum correlation functions can be obtained from a Monte Carlo simulation of the dynamics of the classical 2D system. By contrast, conventional quantum Monte Carlo would act on the Feynman histories in 3D spacetime (endowed with an additional fictitious `dynamics').

Simulations are described in Appendix~\ref{Appendix: numerical methods}.
The mapping relates the quantum Hamiltonian in Eqs.~\ref{eq:flippable},~\ref{Eq: H gapless loop} to a Markov process in continuous time; for numerical convenience we instead use discrete time, with  $L^2$ random  hexagons updated in each time step. This difference is unimportant when $L$ and $t$ are large.  Our choice of rates for the Markov process corresponds to the choice of couplings $K_1/K_2 = |d|^2/(1+|d|^2)$ described in Sec.~\ref{sec:latticeH}. We  expect universal properties to be independent of this choice.

\subsection{Improved analytical bound on $z$}\label{sec:zbound}

In order to motivate the bound, let us first make a crude estimate of the timescale for annihilating a loop of linear size $R \gg 1$ and of total length $\ell\sim R^{d_f}$ (see Sec.~\ref{subsection:classicalmapping} for a discussion of the fractal dimension $d_f$). We use the classical Markov process described in the previous subsection.

In a given $O(1)$ time interval, the area $A\sim R^2$ of a large loop changes by the addition of $O(\ell)$ random positive or negative increments, at various locations around the loop.
In reality the increments at different times and in different locations are correlated, but for simplicity let us first imagine that they are completely independent. 
Then the change $\Delta A$ in the area grows diffusively with time: ${|\Delta A|\sim (t\cdot \ell)^{1/2}}$. 
Setting $\Delta A\sim R^2$ gives the timescale ${t\sim R^4/\ell\sim R^{4-d_f}}$ to change the area by an $O(1)$ fraction (where we have used the fractal scaling of $\ell$ with $R$).
Identifying this with the characteristic dynamical timescale $t\sim R^z$ on scale $R$ gives $z\overset{?}{=} 4-d_f$.

In reality this value is not correct, because of the correlations mentioned above, which slow the dynamics further. However one can show analytically that the above value is a lower bound:
\be{
z\geq  4-d_f.
}\ee
In Appendix~\ref{Appendix: analytical bound} we construct a series of variational wavefunctions orthogonal to the ground states, labeled by an integer $n\neq 0$, and show that their energy $E$ satisfies $E_n \le cn^2L^{d_f-4}$ ($c$ is an $O(1)$ constant.) 
We use the idea of Ref.~\cite{freedman2008lieb}, which is to create an excited state of the loop model by twisting the phase of the wavefunction $\Psi(C)$ as a function of a slow `mode'. Our observation here is that using the area of a large loop for the slow mode gives a stronger bound than using the length, as was done previously.
We construct a series of variational states, labeled by an integer $n$, by continuously changing the phase of the wavefunction according to the area of the largest loop:
\be
\ket{n}=\frac{1}{\sqrt{Z}}\sum_{C}d^{|C|}e^{2\pi n i \, p(A_C)}\ket{C}.
\ee
Here $A_C$ is the area of the largest loop in configuration $C$, and $p(A_C)$ is the cumulative probability distribution of $A_C$ in the classical ensemble. Making the phase of the wavefunction proportional to this probability ensures the variational states are orthogonal to the ground state and to each other (see Appendix ~\ref{Appendix: analytical bound}).

\begin{figure}[t]
\begin{center}
\includegraphics[width=0.99\linewidth]{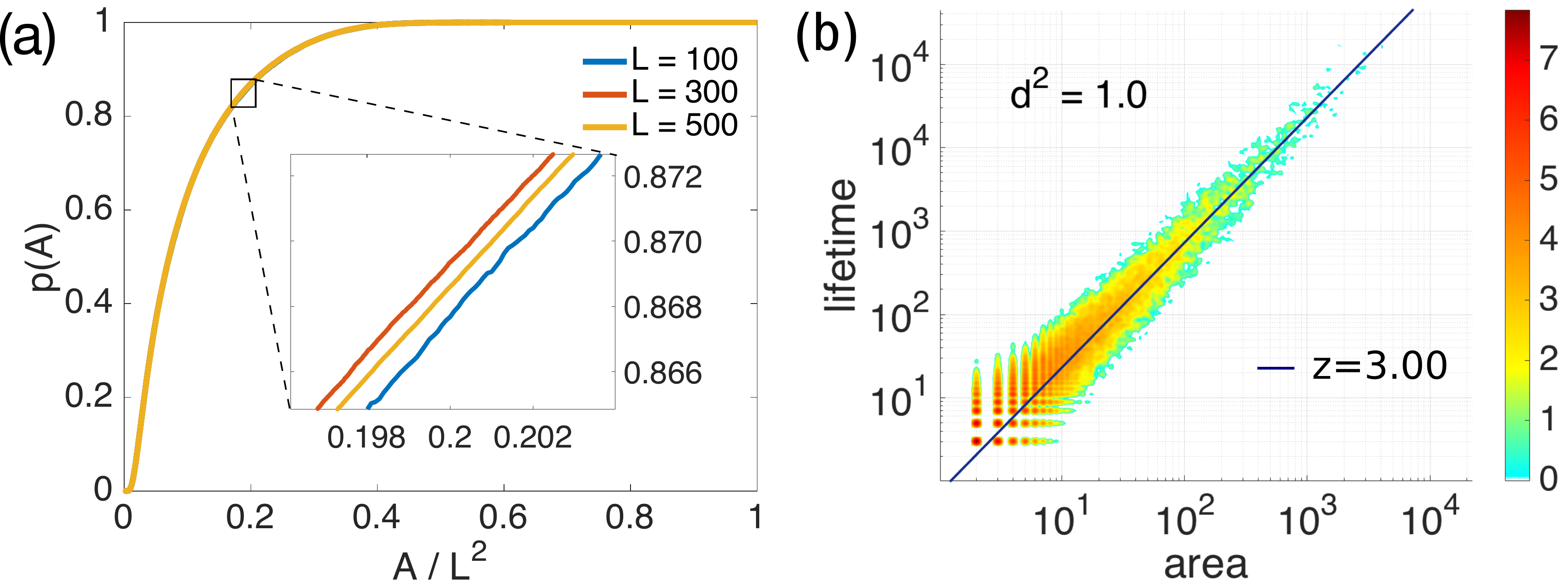}
\caption{Area, length, and lifetime distribution of loops. Fig. (a) shows the cumulative area distribution of the largest non-winding loop in the classical ensemble on an $L\times L$ torus ($d=1$). $p(A)$ is the probability that the largest loop has area smaller than $A$. For system sizes $L = 100, 300, 500$, $p(A)$ collapses to the same curve. Fig. (b) shows the joint distribution of lifetime and maximum area of loops during the Markov process corresponding to the quantum dynamics. The intensity in Fig. (b) at each point is proportional to the log of number of loops with the corresponding lifetime and maximum area. The high intensity region follows the line: lifetime $\sim\ \text{area}^{1.5}$, consistent with the scaling $t\sim R^3$.}
\label{Fig: area length life}
\end{center}
\end{figure}

In the physical proof sketched above, we implicitly assumed that the size distribution of the largest loop, $p(A_C)$, converges to a smooth function of $A_C/L^2$ in the thermodynamic limit. It is well known that the classical loop ensemble has loops of radius comparable with the system size for $|d|^2\le 2$ (this is a consequence of scale-invariance).
In Fig.~\ref{Fig: area length life}(a), we verify it for $d=1$, and numerically compute the probability distribution of the area of the largest non-winding loop. On the torus with size $L=100$, 300 and 500,\footnote{We take periodic boundary conditions in two directions: the horizontal bond direction and the direction at 120 degrees to it.}
we found $p(A_C)$ converges to a smooth function of $A_C/L^2$. The residual differences are mainly statistical.

\subsection{Numerical estimates of $z$}
\label{sec:znumerics}

We extract the dynamical exponent numerically in two ways, finding that the lower bound above is not saturated.
We use the mapping to a classical Markov process described in Sec.~\ref{sec:markovprocess} 
and Appendix~\ref{Appendix: Quantum Markovian correspondence}.
We trace the evolution of large loops in this dynamics, and we also compute dynamical correlation functions.

First, Fig.~\ref{Fig: area length life}(b) shows the distribution of \textit{lifetimes} of loops in the classical Markov process, as a function of loop size: specifically, as a function of the maximum area  attained by the loop during its lifetime (we follow the loops that are present in an initial equilibrated state). 
Scale invariance of the dynamics implies that ${(\text{lifetime})\sim(\text{area})^{z/2}}$.
We find that the data agrees well with this scaling relation if we take $z=3$.\footnote{For this result, we take a torus of size $L=500$, time $T = 10^5$ (see App.~\ref{Appendix: numerical methods} for the definition of the time unit). The results shown are averaged over 6 such runs.}

\begin{figure}[t]
\begin{center}
\includegraphics[width=0.8\linewidth]{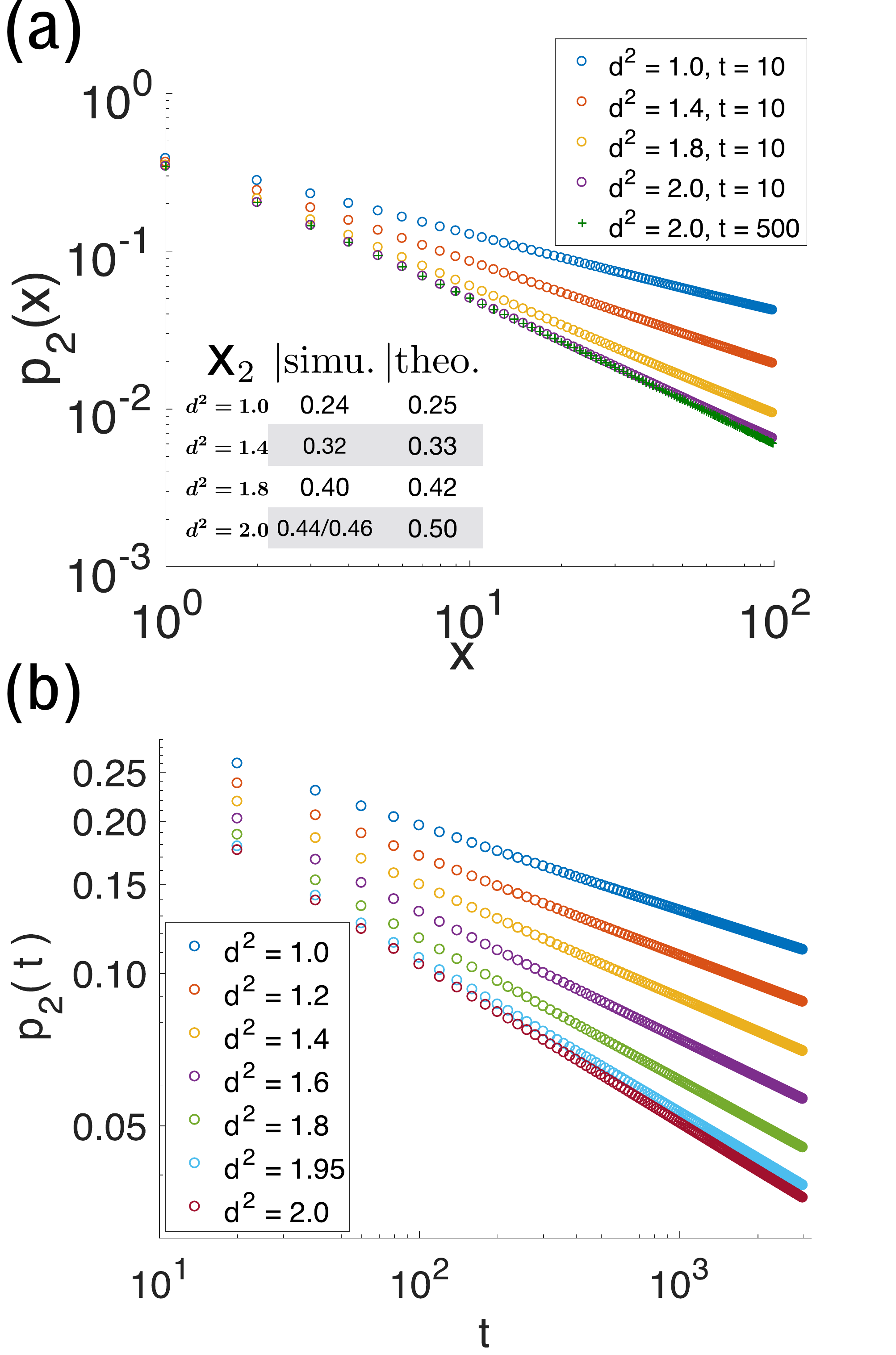}
\caption{Fig.~(a): $p_2(x)$, the probability that two points separated by distance $x$ are connected by one loop (on a torus with $1000\times 1000$ plaquettes). Different colors label different $d$, and/or different warm-up times $t$ (App.~\ref{Appendix: numerical methods}). Theoretically, $p_2(x)\sim x^{-2x_2}$. The theoretical and measured $x_2$ are shown in the inset chart. For $d^2=1.0,1.4,1.8$, after 10 warm up steps, the measured exponents fit the theory very well. For $d^2=2.0$, there is a significant deviation because of a known logarithmic correction to the power-law decay at this critical $|d|^2$. Fig.~(b): $p_2(t)$, the probability that a fixed point is occupied by the same loop at times $0$ and $t$, for various $d$, on a torus with $400\times400$ plaquettes (note that $t\ll L^z$ for these values, so finite size effects are expected to be small).}
\label{Fig: CorrX2xt}
\end{center}
\end{figure}

\begin{figure}[t]
\begin{center}
\includegraphics[width=0.7\linewidth]{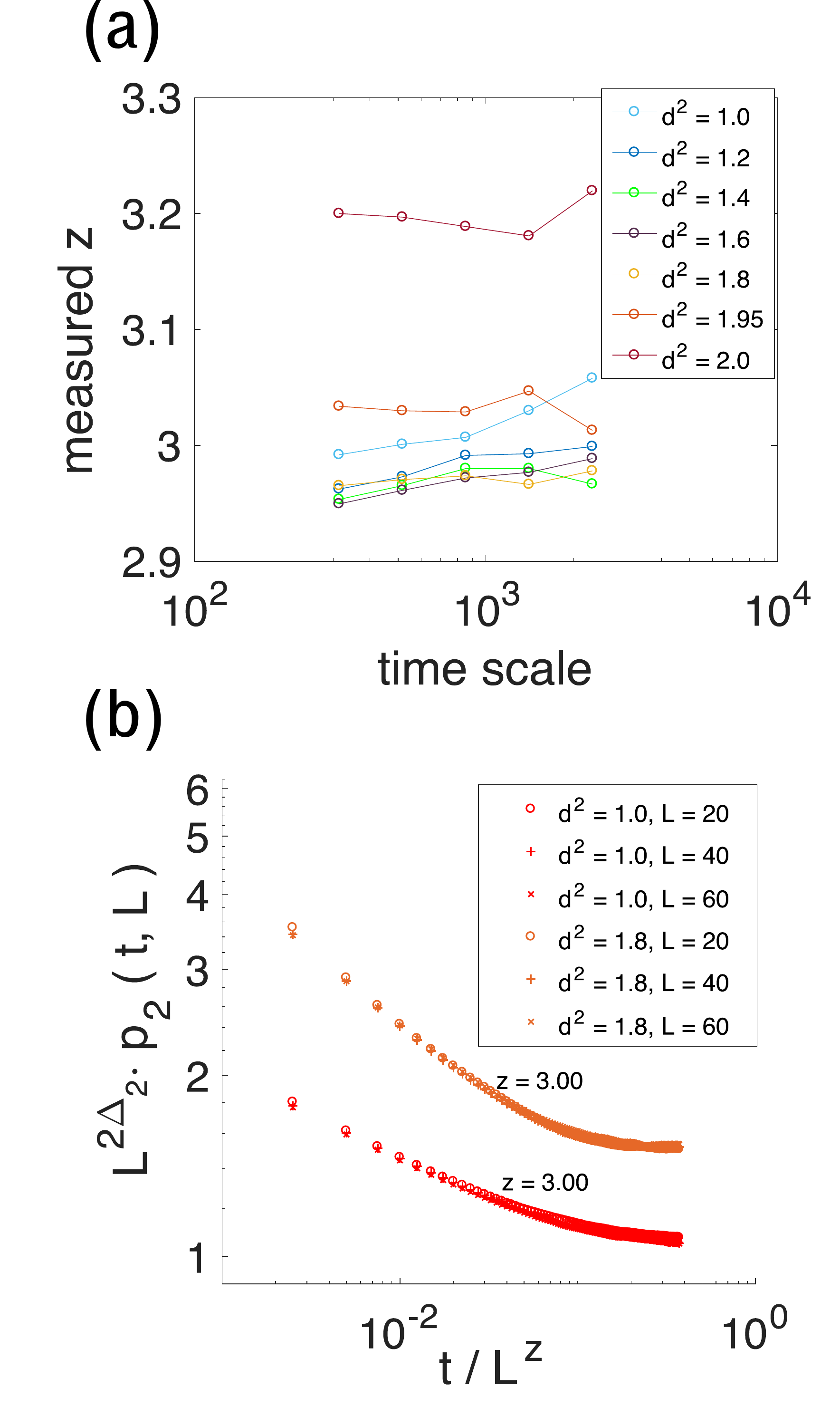}
\caption{Fig.~(a): estimates of $z$ from the data in Fig.~\ref{Fig: CorrX2xt}~(a). We divide the time $t$ into 5 ranges: $5.5\le\ln t<6.0$, $6.0\le\ln t<6.5$,\dots,$7.5\le\ln t<8.0$, and perform separate linear fitting to the log-log plot, to get dynamical exponent estimates on each scale. Results are consistent with $z=3.00(6)$ for all ${|d|^2<2}$. We expect that the larger deviation from 3 at $|d|^2 = 2$ is due to logarithmic finite-time corrections there (see text). 
In Fig.~(b), we plot the correlator $p_2(t)$ for smaller systems, with linear size $L = 20, 40, 60$. The scaling hypothesis requires ${p_2(t,L) = c'L^{-2\D_1}F(t/L^z)}$. Correlators for different $L$ collapse very well assuming $z=3$ for both $|d|^2= 1$ and $|d|^2 = 1.8$.}
\label{Fig:X2zfit}
\end{center}
\end{figure}

Next we use time-dependent correlation functions to give a more accurate measurement of the dynamical exponent.
Specifically, we consider here the probability, which we denote $p_2$, that two space-time points lie on the same worldsurface of a loop. 
This probability does not map to a local correlation function in the quantum model, but scale-invariance implies that it still obeys the scaling form in Eq.~\ref{eq:generalscalingform}, so it can be used to extract $z$.
In Sec.~\ref{Sec: operators and correlation functions} we perform similar analyses of conventional local  correlators in the quantum model.

First we check the expected equal-time scaling (which does not involve the dynamical exponent $z$). 
When the points are at equal time, $p_2(x)$ is the probability in the classical 2D ensemble that two points separated by $x$ (along the horizontal lattice direction) lie on the same loop. This scales as ${p_2(x)\sim x^{-2 x_2}}$, where ${x_2 = 2-d_f}$ depends on $|d|$ (see the discussion around Eq.~\ref{Eq: x2k}).
Our numerics for various values of $d$ in Fig.~\ref{Fig: CorrX2xt}(a), where the fitted exponent values are shown, are in good agreement with the expected values of $x_2$, except at ${|d|^2 = 2}$ where there is a larger error.
We attribute this to logarithmic finite-size effects which are present at the special value ${|d|^2=2}$ because of a marginally irrelevant scaling variable (see  Sec.~\ref{subsec: lowlying operator}).

Fig.~\ref{Fig: CorrX2xt}(b) shows the probability that a given spatial point lies on the same world surface at time $0$ and time $t$, which we denote $p_2(t)$. 
By the scaling hypothesis, we expect $p_2(t)\sim t^{-2x_2/z}$.
Results in Fig.~\ref{Fig: CorrX2xt}(b) are consistent with power-law decay in time. 
To test the scaling hypothesis and to put an error bar on $z$, we divide the time $t$ into 5 ranges: $5.5\le\ln t<6.0$, $6.0\le\ln t<6.5$,\dots,$7.5\le\ln t<8.0$, and perform separate linear fits to the log-log plot. 
The result in Fig.~\ref{Fig:X2zfit}(a) is consistent with $z=3$ for all values of $d$ (with an error bar $\simeq \pm 0.06$), except for the case $|d|^2 = 2$, where we expect logarithmic finite-time effects (Sec.~\ref{subsec: lowlying operator}). 
There is some indication of a systematic drift in $z$ as the timescale is increased, but nevertheless it seems $z$ is close to three.

Finally, in Fig.~\ref{Fig:X2zfit}(b), we
check the scaling ansatz for the full range of $t/L^z$, using smaller systems with $L\leq 60$ so that times of order $L^z$ can be accessed numerically. We expect the following scaling form (the correlation function is at equal position; the second argument denotes the system size)
\be
p_2(t,L) = L^{-2x_2}F(t/L^z), 
\ee
where $F$ is an unknown scaling function. Therefore we plot $L^{2x_2}p_2$ against $t/L^z$ for $L=20,40,60$. 
Correlators for the same $d$ and different $L$ are seen to collapse when we take ${z=3}$, consistent with this value of the dynamical exponent.
Different scaling curves are obtained for different $d$.

See Appendix~\ref{Appendix: numerical methods} for details of numerical methods.

\subsection{Ground state degeneracy}
\label{sec:gsd}

Does the energy of the first excited state scale as $E\sim L^{-z}$, with $z=3$? 
This is the naive expectation from scale invariance, and it is true on the sphere and on the disc with appropriate boundary conditions. 
However, on topologically nontrivial manifolds the loop model Hamiltonian (Eq.~\ref{eq:flippable}) has a nontrivial ground state degeneracy.
For a more general Hamiltonian in the same universality class,
this ground state degeneracy is lifted by \textit{dangerously irrelevant} operators, giving a gap that scales with an exponent larger than $z$.

On the torus, there are ${ \mathcal{O}(L^2)}$ different topological sectors, labeled by winding numbers of loops in the $x$ and $y$ directions (see Appendix~\ref{Appendix: GSD} for details).
Without reconnection, winding numbers are conserved, and there is at least one ground state in each sector. 

In fact, the frustration-free Hamiltonian on the torus has a more severe (though still subextensive) ground state degeneracy: there are $\sim \exp \lf {\text{const.}\times L} \ri$ `frozen' ground states which are fully packed by large loops all winding in the same direction (every site is visited by a winding loop). These frozen configurations do not have any allowed moves, and are ground states of the solvable Hamiltonian.

However, the degeneracy of these states with the  ground states of interest is an artifact of fine-tuning.
A \textit{generic} Hamiltonian in the same universality class will include irrelevant perturbations in the ultra-violet that are absent in the ideal Hamiltonian displayed in Sec.~\ref{sec:latticeH}. 
The irrelevant four-loop reconnection operator discussed in the following section is one  candidate for such an irrelevant operator.
These perturbations can push the frozen states up to an $O(1)$ energy.
This energy is $O(1)$ (rather than a negative power of $L$) because the frozen states already differ from the nontrivial ground state at the lattice scale, and because the coefficient of the irrelevant operators in the UV is generically $O(1)$.

The states with $O(1)$ winding 
(which are locally close to the ground state, since the winding number is $\ll L$)
also have their degeneracy lifted by `dangerously' irrelevant\footnote{They are dangerously irrelevant in the sense that despite being irrelevant they give the leading contribution to this gap.} operators. But instead of being pushed to $O(1)$ energies they acquire an anomalously small gap (smaller than $L^{-z}$).
In general, states with $O(1)$ winding numbers recombine into the new ground state and the anomalously low-lying states: these have a gap of order $L^{-w}$, where $w>z$ is determined by the RG eigenvalue $y_\text{irr}<0$ of the irrelevant operator and by its matrix elements in the space of low-lying states (the simplest possibility being $w = z+|y_\text{irr}|$).

Interestingly, we have these extra low-lying states only on topologically nontrivial manifolds, and they are related to the total winding number.
The quantum dimer model also has low-lying states (that can be split by dangerously irrelevant operators) associated with a winding number with a slightly different definition \cite{fradkin2004bipartite, henley2004classical}.
These states are reminiscent of the Anderson tower of states in the ordered phase of a quantum antiferromagnet, where states with spin $S$ have energy ${\sim S^2/\text{volume}}$, well below the gap to the Goldstone mode. However, the reason for low-lying states in the present model is purely topological.

At $d=\pm\sqrt{2}$, we can add a \textit{relevant} 3-loop reconnection operator which preserves the ground state on the sphere (and hence the equal-time scaling functions) but completely changes the dynamics of loops: see Sec.~\ref{Sec: Jones Wenzl} and Appendix~\ref{Appendix: GSD}. The low-lying states on the torus are also lifted, with only 9 of them remaining, corresponding to states in the doubled $SU(2)_2$ topological field theory~\cite{freedman2004class,freedman2005line} (Appendix~\ref{Appendix: GSD}).

\section{Superuniversality of dynamical exponents}
\label{sec:superuniversality}

A striking feature of our numerical results in Sec.~\ref{sec:znumerics} is that, while the scaling dimensions $x$ depend on $|d|$, the dynamical exponent $z$ seems to be independent of $d$.  
Here we show that this is in fact the generic expectation for any line of RG fixed points.

Assume we have a Hamiltonian $H_u$ that depends on a parameter $u$, and that in the IR this theory flows to an RG fixed line, with $u$ controlling the position on the fixed line. Let $z(u)$ be the dynamical exponent at the corresponding position on the fixed line. At large $L$, the energy gap on (say) the sphere\footnote{For the present class of models the energy gap scales with $L^{-z}$ on the sphere, but scales with a larger exponent on the torus because of almost-degenerate states that are split by dangerously irrelevant operators.} scales as 
\be
\Delta(u,L) = \f{A(u)}{L^{z(u)}}  + \ldots
\ee
where $A(u)$ is a nonuniversal constant. If we differentiate with respect to $L$, assuming that $A(u)$ and $z(u)$ and the subleading terms are well-behaved as a function of $u$,
\be\label{eq:gapderivative}
\f{\dd \Delta(u,L)}{\dd u}
\simeq
\f{ A'(u) - A(u) z'(u) \ln L }{ L^{z(u)} }.
\ee
Alternatively, we may compute the same quantity from the derivative of the Hamiltonian:\footnote{By the Feynman-Hellmann theorem, the derivative of the ground state or the excited state with respect to $u$ does not contribute.}
\be
\f{\dd \Delta(u,L)}{\dd u}
= 
\bra{\text{ex}} 
\f{\dd H_u}{\dd u}  
\ket{\text{ex}}
-
\bra{\text{GS}} 
\f{\dd H_u}{\dd u}  
\ket{\text{GS}}
\ee
where $\ket{\text{ex}}$ is the first excited state.
The right-hand side involves expectation values of a local perturbation summed  over space. From a standard coarse-graining argument,  we would expect the right-hand side to scale as  $L^{d-x}$, where the factor of $L^d$ comes from the spatial sum, and $x$ is the scaling dimension of the perturbation (which may be expressed in terms of scaling operators of the continuum theory). Since this perturbation is  marginal by assumption, $x=d+z$, this gives
\be\label{eq:nolog}
\f{\dd \Delta(u, L)}{\dd u} \sim 
L^{-z(u)}.
\ee
Note the absence of a logarithmic term in $L$. 

Therefore, assuming conventional scaling for the expectation value of  $\dd H_u / \dd u$, comparing with Eq.~\ref{eq:gapderivative} implies that $z'(u) = 0$. This is the fact stated above.

Could the conventional scaling expectation for $\dd H_u / \dd u$ break down?
We can certainly obtain logarithms in expectation values if the theory has a marginally \textit{irrelevant} perturbation in addition to the exactly marginal one. However generically we do not expect such a perturbation, and even if one is present for a given microscopic Hamiltonian, in most cases we can simply tune the Hamiltonian so that the coefficient of this perturbation vanishes (and  the logarithms go away). We can then repeat the argument to obtain the desired property for the RG fixed line.\footnote{There may be exceptional cases where the marginally irrelevant perturbation is dangerous and cannot be set to zero.}
In the very different context of nonunitary conformal field theories logarithms appear via a different mechanism, but we do not expect that to be relevant here.\footnote{In some classical statistical mechanics models described by nonunitary conformal field theories the action of the renormalization group transformation on the set of operators with a given scaling dimension is not diagonalizable, and this leads to logarithms in correlation functions \cite{cardy2013logarithmic}. 
We do not know whether something similar is possible in unitary but non-conformal quantum theories, but in any case we argue that it could not change the above result (assuming the RG flow terminates on a fixed line without dangerously irrelevant couplings). The simplest case would be a pair of operators $\mathcal{O}$, $\mathcal{O}'$ which under coarse-graining by a factor of $b$ transform as $\mathcal{O}\rightarrow b^{-x} \lf \mathcal{O}+ \mathcal{O}' \ri$, $\mathcal{O}'\rightarrow b^{-x} \mathcal{O}'$. In this case the one-point function $\< \mathcal{O} \>$ contains a logarithm. However if $\dd H_u /\dd u$ transformed like $\mathcal{O}$ does here, then  $\dd H_u /\dd u$ would not be an exactly marginal perturbation in the usual sense, simply because under RG it would generate the additional perturbation $\mathcal{O}'$ to the Hamiltonian.}

\section{Correlation functions and
scaling operators in the loop model} 
\label{Sec: operators and correlation functions}

In this section we discuss the operator content of the loop model, provide analytical scaling dimensions for various quantum operators and check them numerically.

We develop a systemic treatment that maps every equal-time quantum correlator to a sum of classical probabilities, from which we can read off the scaling dimensions of many operators. This formalism allows us to understand (at least in principle) every quantum operator with nonzero equal-time correlator, revealing the topological structure the in operator spectrum argued for on general RG grounds in Sec.~\ref{sec:topologicaloperatorclassification}.

This understanding is not quite the end of the story since (as noted in Secs.~\ref{sec:scalingforms}, ~\ref{sec:ffscalingoperators}) some nontrivial scaling operators have vanishing equal-time correlators. In the special case $d=1$, all operators have this property, but  we can obtain their scaling dimensions simply by taking the limit $d\rightarrow 1$ in formulas for generic $d$. 
For $d\neq 1$ a subset of operators are `hidden' operators whose  equal-time correlation functions vanish.
However we can determine the scaling dimensions of at least some of these by the logic in Sec.~\ref{sec:ffscalingoperators}, i.e. by relating them to perturbations of the associated classical Boltzmann weight.

One natural question is how many relevant or marginal perturbations the quantum loop models have. 
We will exhibit a set of low-lying (relevant or marginal) scalar operators that is plausibly complete.
We cannot prove that there are no other low lying scalar operators, because we cannot rule out the possibility of additional hidden operators that neither show up in equal-time correlation functions nor correspond to deformations of the classical Boltzmann weight.  However, numerical results discussed below are consistent with the hidden operators discussed above being the most relevant ones.

\subsection{Low-lying local operators}\label{subsec: lowlying operator}

Before giving a formal classification of operators, let us summarize some of the most important ones. 

Since we have a preferred basis for the Hilbert space, the loop occupation basis, we can distinguish diagonal and non-diagonal operators in this basis. They will have different interpretations in terms of fluctuating loops.\footnote{{ We can also neglect operators that create open strings: these have exponentially decaying correlators, as the Hamiltonian imposes a gap for strings with dangling ends.}}

When we consider equal time correlators, diagonal operators built from $\sigma^z$ map to local operators in the classical loop ensemble,
which can be expanded in terms of scaling operators in the conformal field theory for the loop model. This allows us to determine the scaling dimensions of the quantum operators.

At first glance, the simplest microscopic operator is the spin $\sigma^z$ on a link of the lattice.
Surprisingly, the leading continuum operator contributing to this lattice operator is not a scalar operator, 
but instead an operator of spin 2 under spatial rotations.
When inserted into equal-time correlation functions, this diagonal quantum operator maps to the stress tensor $T_{\mu \nu}$ of the two-dimensional classical theory. 
The point is that  $\sigma^z$ on a particular edge is invariant under reflection about this edge, but is not invariant under spatial rotations; it has the same symmetry as $T_{xx}$, where the $x$ axis is aligned with the bond direction. 
The 2D stress-energy tensor has scaling dimension 2 and is the most relevant operator with the required symmetry, and hence the leading contribution to the correlator of $\sigma^z$. 
Note that, since the spin-2 scaling operator in the quantum theory is not symmetric under rotations, it does not appear in a perturbed Hamiltonian if the perturbation respects the symmetry.

Fig.~\ref{Fig:PhysicalcorrelatorSz} shows numerical data for the  correlator of $\s^z$ at \textit{distinct} times. From the above result for the scaling dimension, and the scaling forms in Sec.~\ref{sec:scalingforms}, we expect this to decay as $t^{-4/z}$ with $z\simeq 3$, for all values of $d$. The data is consistent with this expectation.
In Ref.~\cite{troyer2008local} the loop models were  simulated using a different approach, and numerical scaling $\sim t^{-1}$ was reported for the $\s^z$ two-point function at $d=1$; this exponent is not too far from the $4/3$ we find.

To obtain a scalar operator (i.e. to isolate the subleading scalar contribution to $\sigma^z$) we can symmetrize the lattice operator under rotations, for example by taking the sum of $\sigma^z$ for the three spins surrounding a given vertex. 
The resulting operator simply measures whether or not a given site is visited by a loop. We denote it by $\sigma^z_\text{symm}$.\footnote{Note that there is no symmetry that changes the sign of $\sigma^z$, so it is not significant that the operator is odd in $\sigma^z$. Any generic lattice operator that is invariant under rotations and diagonal in the $\sigma^z$ basis will yield the same continuum operator as $\sigma^z_\text{symm}$.}
In the 2D theory there are two low-lying scalar operators, with dimensions:
\ba
x_{\text{symm}}&=4
&
x'_{\text{symm}} & = { \f{4-2g}{g} . }
\end{align}
The first of these is $\bar{T}T$, the product of the holomorphic and antiholomorphic components of the 2D stress tensor;
see Eq.~\ref{eq:gdef} for the definition of $g$ as a function of $|d|$.
Correspondingly we expect two scaling operators in the quantum theory, $\mathcal{O}_\text{symm}$ and $\mathcal{O}_\text{symm}'$,
both trivial under all symmetries, with these dimensions. 
Both will appear when we expand a lattice operator like $\sigma^z_\text{symm}$ in terms of continuum operators, and the one with the smaller dimension will dominate in correlation functions.
For $|d|^2<1$, this is $x_{\text{symm}}$, but for $|d|^2 >1$ it is $x'_{\text{symm}}$. 
For $|d|^2 \gtrsim 0.4$ both operators are relevant in the RG sense if used to perturb the Hamiltonian, i.e. $y=2+z-x$ is positive (Eq.~\ref{eq:ydef}) for both.

\begin{figure}[t]
\begin{center}
\includegraphics[width=0.4\textwidth]{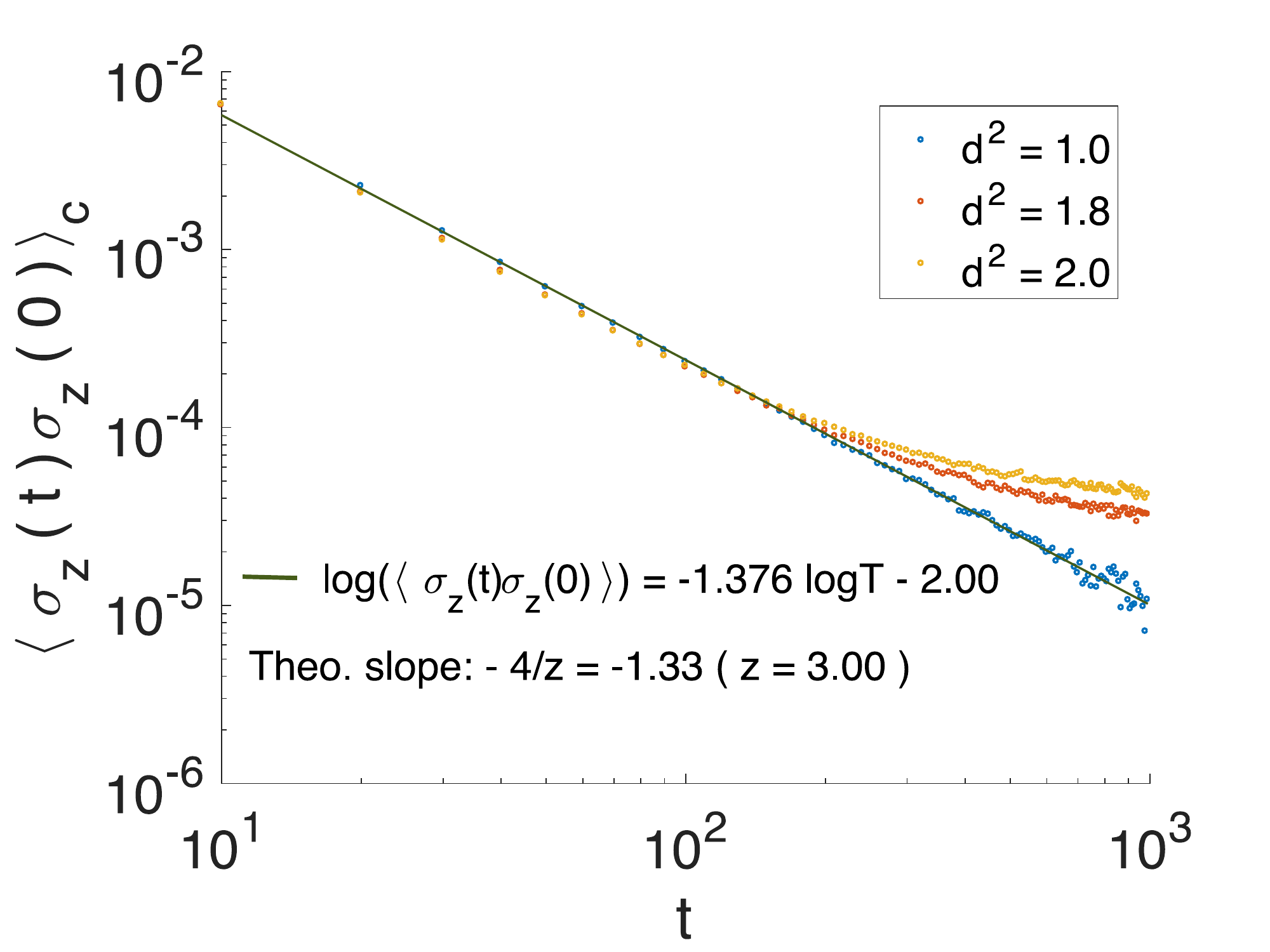}
\caption{The connected 2-point function of $\sigma^z$ on an edge (unsymmetrized), on a torus with $400\times 400$ plaquettes. Theoretically we expect a leading scaling dimension 2. The numerical result fits the expected power-law decay very well at smaller $t$. For the larger $t$, the numerical error for $\<\sigma^z\>$ becomes significant and $\<\sigma^z(t)\sigma^z(0)\>_c$ inherits this error. At $d=1$, we know $\<\sigma^z\>$ vanishes exactly, so this problem does not exist.
}
\label{Fig:PhysicalcorrelatorSz}
\end{center}
\end{figure}

\begin{figure}[t]
\begin{center}
\includegraphics[width=0.35\textwidth]{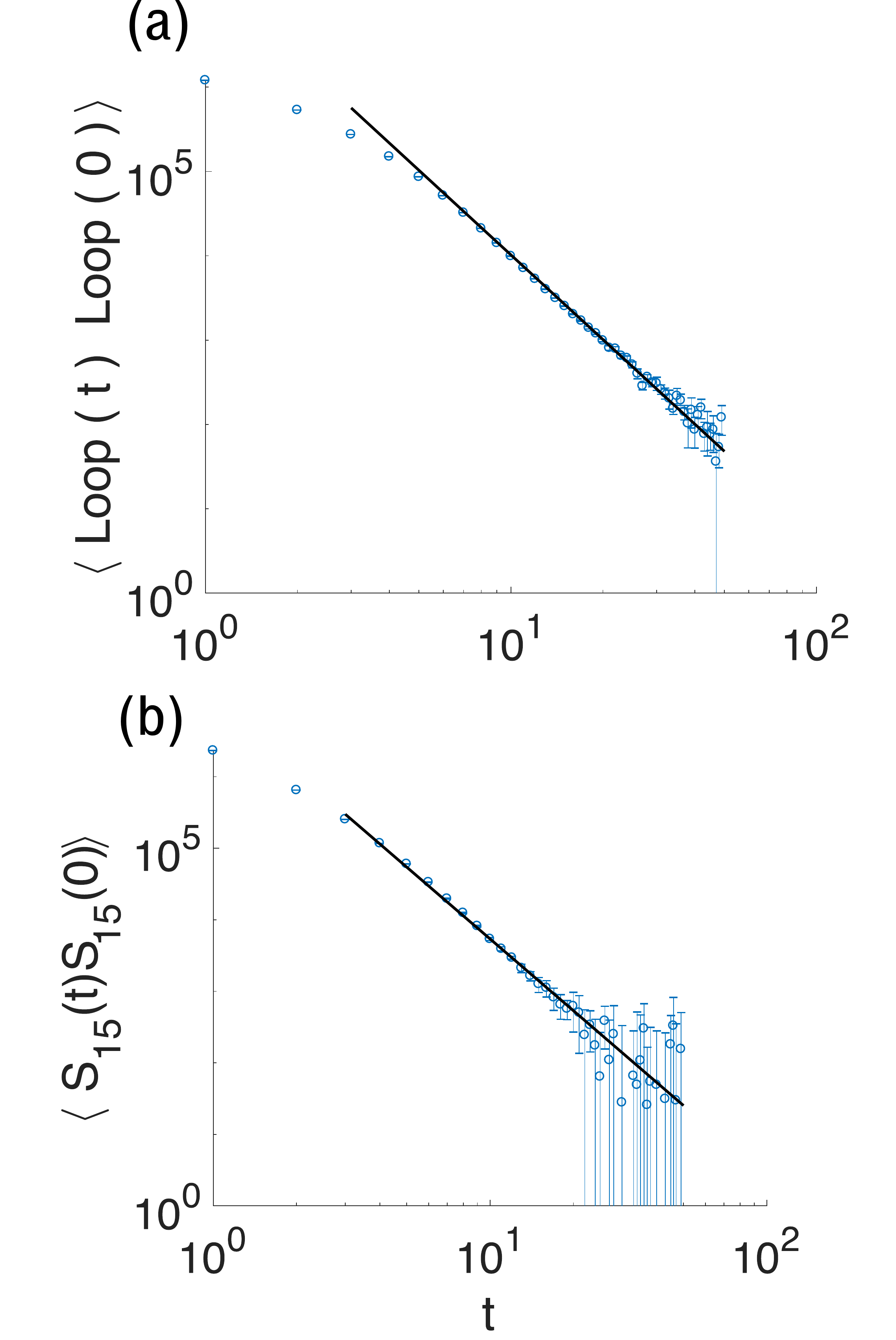}
\caption{Correlators of hidden operators measured at ${|d|^2=1.8}$. Fig.~(a): Correlator of the loop creation/annihilation operator, the hidden operator that preserves the frustration-free condition and induces the change of $|d|$ in the ground state. The black line indicates a power-law decay $t^{-2(z+2)/z}$, where we have set $z=3$ (ditto in Fig.~(b)).  Fig.~(b): Correlator of the lattice hidden operator $S_{15}$ defined in the text, which has no useful classical correspondence. Numerical results are consistent with a single exactly marginal continuum scaling operator dominating both of the lattice operators shown in this figure.}
\label{Fig:hiddenOperator}
\end{center}
\end{figure}

By the argument in Sec.~\ref{sec:ffscalingoperators}
there are also less-relevant quantum operators with dimensions ${x_\text{symm}+z}$ and ${x_\text{symm}'+z}$ obtained using the anticommutator with the Hamiltonian density; these are irrelevant, except at $|d|^2=2$ when the latter becomes marginal (though not exactly marginal), reflecting the existence of a marginal perturbation in the classical ensemble at $|d|^2=2$.
(This is the source of the logarithmic finite-size corrections mentioned in Sec.~\ref{sec:znumerics}.)

The definition of a local operator also depends on whether we interpret the loops as chains of flipped spins or as Ising domain walls. In the latter case the Ising spin, which we will denote $\tau^z$, gives an additional low-lying local operator (which is odd under the Ising symmetry that exists in that representation). This maps to a simple kind of twist operator in the loop model: the correlator $\<\tau^z(r) \tau^z(0)\>$ measures the parity of the number of loops crossing a line between $r$ and $0$. This operator has scaling dimension
\be
x_{\tau}= { \f{3-2g}{2g}}.
\ee

We turn now to non-diagonal operators. We have already discussed (Sec.~\ref{sec:ffscalingoperators}) two exactly marginal operators, with $x=2+z\simeq 5$, which, when added to the Hamiltonian, change the magnitude and phase of $d$ by modifying the loop creation/annihilation term in the Hamiltonian. These operators have nonzero matrix elements between states with and without a small loop. They do not reconnect large strands, however, so they are topologically trivial according to the classification in Sec.~\ref{sec:topologicaloperatorclassification}. 

For a numerical calculation of the  temporal correlator for one of these operators, see Fig.~\ref{Fig:hiddenOperator}(a). This is the operator which we perturb the the Hamiltonian by when we change $|d|$.\footnote{For real $d$, we can define the operator as $\operatorname{Loop} \propto \left( \begin{array}{cc}
2d & -1\\
-1 & 0
\end{array} \right)$ in the basis of an empty plaquette state and a small loop state.}
Results are consistent with the expected power-law decay in time.
By the argument in Sec.~\ref{sec:ffscalingoperators}, this operator is a hidden operator with vanishing equal-time correlator.

\begin{figure}[t]
\begin{center}
\includegraphics[width=0.4\textwidth]{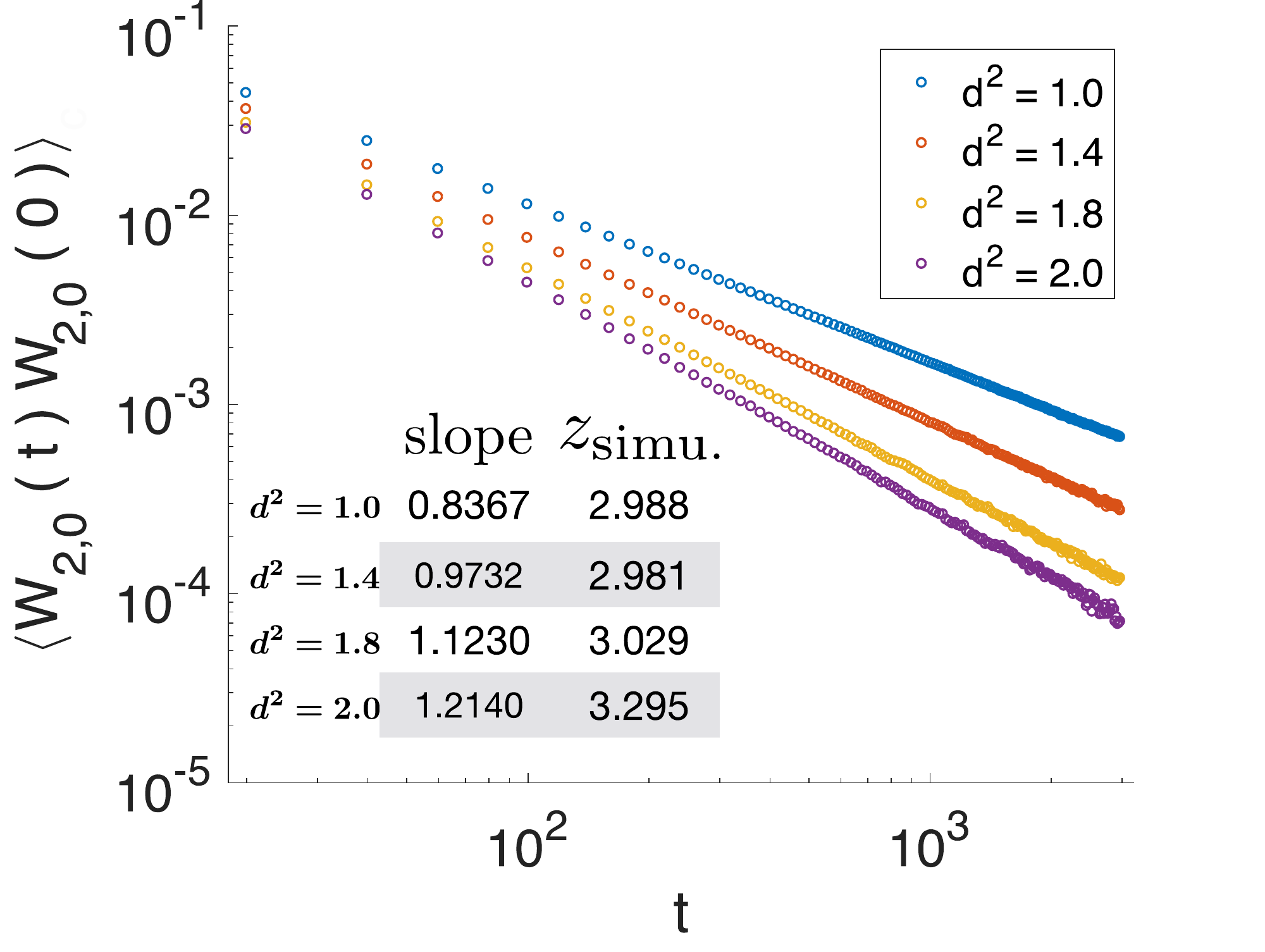}
\caption{Temporal correlation function of the 2-loop reconnection watermelon operator $W_{2,0}$ on a torus with $400\times 400$ plaquettes, for $d=1.4,~1.8,~2.0$. For $d=1$, where the correlator is zero, we instead plot its derivative with respect to $d$. These results are consistent with the analytical result for the scaling dimension of the 2-loop reconnection operator and the previous result $z=3.00(6)$ from Fig.~\ref{Fig:X2zfit}(a) (noting that for $|d|^2 = 2$ we expect a log correction).
}
\label{Fig: PhysicalcorrelatorX4}
\end{center}
\end{figure}

The most relevant topologically \textit{nontrivial} operator (reconnection operator) is the two-loop reconnection operator with spin zero, which we call $W_{2,0}$,\footnote{See the definition and detailed discussions in Sec.~\ref{sec:constructreconnectionoperators}. Reconnection operators with nonzero spin are also discussed in in Sec.~\ref{sec:constructreconnectionoperators}.} whose scaling dimension is
\be
x_{4,0} =\f{4g^2-(1-g)^2}{2g} 
\ee
as we show in the next subsection. 
This operator is strongly RG-relevant if added to the Hamiltonian. 
This operator is the leading contribution to the operator $\prod \sigma^x$ which flips all the spins around a plaquette.
The temporal correlator of the two-loop reconnection operator (whose lattice definition is given in the next subsection) is shown in Fig.~\ref{Fig: PhysicalcorrelatorX4}.
At $d=\sqrt{2}$ our result implies a spatial two-point function that decays as $r^{-4}$, which differs from the value close to three quoted from numerical simulations in Ref.~\cite{troyer2008local}: this may be a result of finite size effects in that study.

The  three-loop reconnection operators (with spin zero) have the larger scaling dimension
\be
x_{6,0} = \f{9 g^2-(1-g)^2}{2g}.
\ee 
Three-loop reconnection operators act within a disc with six incoming strands (recall Fig.~\ref{Fig: label spectator}). 
The number of topologically distinct reconnection events increases with the number of strands. 
As a result, while there is only a single spin-zero two-loop reconnection operator,\footnote{Barring possible hidden operators.} 
there are three distinct three-loop reconnection operators with spin zero.
Of these, two are even, and one is odd, under time reversal and parity (reflection in a spatial axis). Therefore if we retain symmetry under spatial rotations and reflections, there are two distinct three-loop reconnection operators that can be used to perturb the Hamiltonian. They are RG relevant, since the scaling dimension $x_{6,0}$ above is smaller than $2+z$.

At the special point $d=\sqrt 2$, both of these three-loop reconnection operators become `hidden' operators (whose equal-time two-point function vanishes). A linear combination of them, called the `Jones Wenzl projector' \cite{freedman2004class,freedman2005line}, can be added to the Hamiltonian \textit{without changing the ground state}. Nevertheless it is a relevant perturbation that leads to a new universality class for the quantum dynamics. We discuss this new universality class in Sec.~\ref{Sec: Jones Wenzl}.

In the original family of universality classes, higher reconnection operators, starting with the four-loop reconnections, are RG-irrelevant (with the exception of the regime $|d|^2\lesssim 0.8$ when four-loop reconnection is relevant). 
The scaling dimension of a spin-zero $k$-loop reconnection operator is
\be
x_{2k, 0}=\f{g^2 k^2-(1-g)^2}{2g}.
\ee

In Table~\ref{optable} we list all the perturbations that we have found which preserve all the spatial symmetries of the lattice model (though not necessarily the dynamical topological constraint), and which are relevant or marginal for $0.8\lesssim |d|^2 <2$. We ignore redundant operators which (if used to perturb the Hamiltonian) can be absorbed in a re-scaling of space and time coordinates. 

\begin{table}[t]
\centering
\begin{tabular}{ |p{1cm}||p{5.2cm}|p{1.5cm}| }
 \hline
  & Description &
  Scaling dimension\\
 \hline
 \hline
 $\mathcal{O}_\text{symm}$   &  diagonal    & 4\\
 $\mathcal{O}_\text{symm}'$&   diagonal   & $\frac{4-2g}{g}$\\
 $W_{2,0}$ & 2-loop reconnection & $\frac{4g^2-(1-g)^2}{2g}$\\
 $W_{3\a,0}$ & 3-loop reconnection & $\frac{9g^2-(1-g)^2}{2g}$\\
 $W_{3\beta,0}$& 3-loop reconnection & $\frac{9g^2-(1-g)^2}{2g}$\\
 $\operatorname{Loop}$ &  `hidden' operator for changing $|d|^2$ & $2+z$\\
$\operatorname{Loop}'$ &  `hidden' operator for changing $\operatorname{arg} d$ & $2+z$ \\
  \hline
\end{tabular}
\caption{Relevant or marginal operators that preserve all symmetries, excluding redundant operators. (For $|d|^2\lesssim 0.8$, 4-loop reconnection also becomes relevant; at $|d|^2=2$ there is an additional marginally irrelevant operator.)}.
\label{optable}
\end{table}

As we will discuss in Sec.~\ref{subsec:dilutecriticalpoint}, there is another $d$-dependent family of fixed points for the quantum loop models that is related to the so-called `dilute' phase of the classical loop ensemble; there the scaling dimensions of the reconnection operators are larger.

In the following subsections we present a classification of local operators up to hidden operators.
We find no other (symmetric, non-redundant) operators that are relevant or marginal.

To test whether there could be any other relevant \textit{hidden} operators we must resort to numerics. Here we perform a partial test at $|d|^2 = 1.8$, considering only diagonal operators. 
We calculate the temporal correlator of a lattice operator which we denote $S_{15}$.
This is a diagonal operator defined on the six links of a hexagon: it is equal to $+1$ if a total of 5 of the 6 links are occupied, to $-1$ if 1 of the links is occupied, and $0$ otherwise.
This operator is a sum of hidden operators (according to the definition in Sec.~\ref{sec:ffscalingoperators}) and its equal-time correlator vanishes.
When it is expanded in continuum operators only hidden scaling operators will appear, and its temporal correlator will be dominated by the leading hidden operator that appears (assuming this has a nonvanishing temporal correlator).

Fig.~\ref{Fig:hiddenOperator} shows that the scaling dimension of this operator is consistent with the marginal value $x=2+z$. This is consistent with the leading hidden operators being the marginal ones $\operatorname{Loop}$ and $\operatorname{Loop}'$ discussed above.

\subsection{Operator equivalence relation}
\label{sec:opequivalence}

In this subsection, we introduce an equivalence relation among local operators. This equivalence relation allows us to write operators in an intuitive standard form that simplifies the computation of correlation functions.

\begin{figure}[t]
\begin{center}
\includegraphics[width=0.45\textwidth]{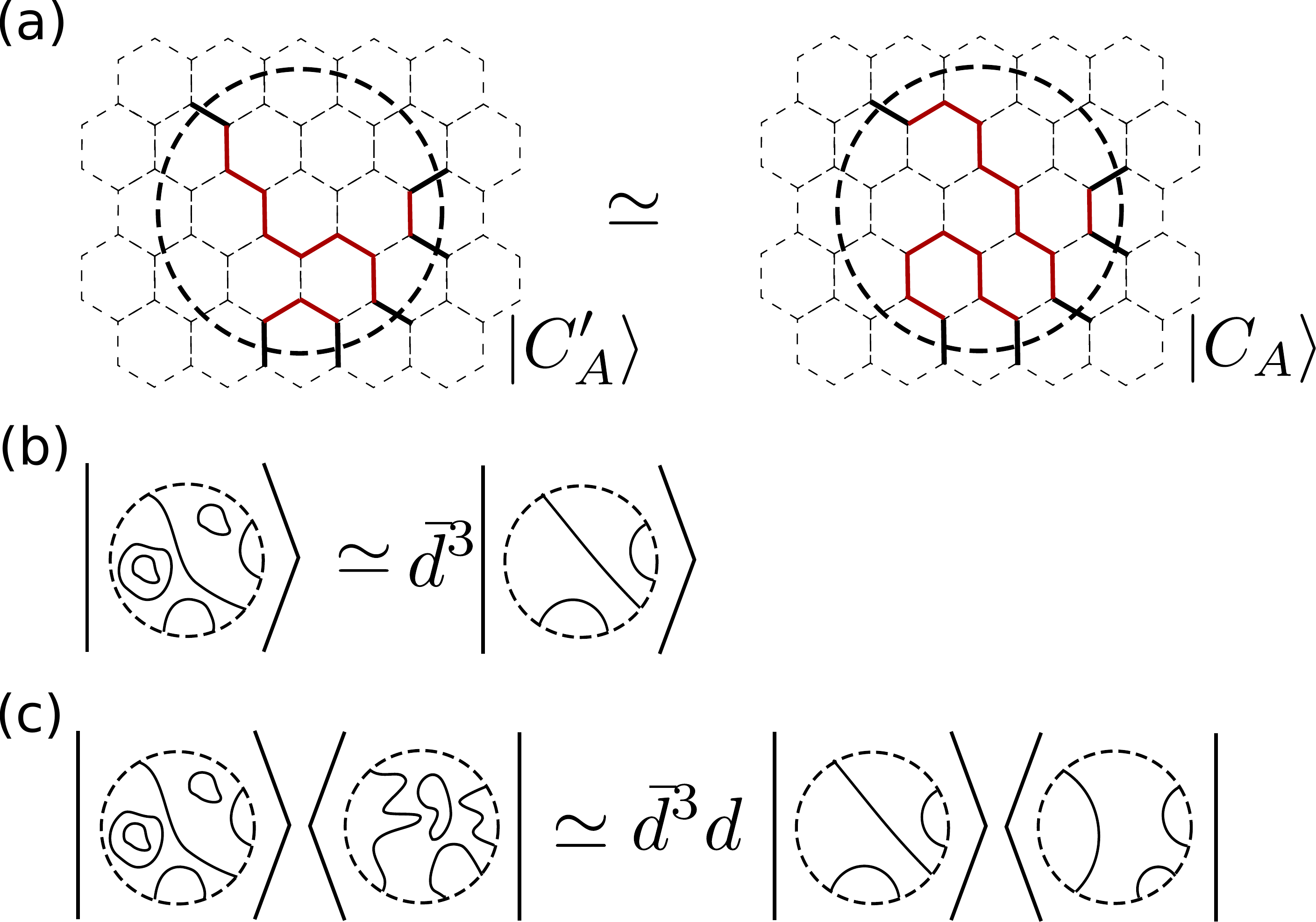}
\caption{Fig.~(a-b) Equivalence relation for states. Two states are defined to be equivalent if their difference is a `forbidden state'. Fig.~(c) Equivalence relation between operators. Two operators are defined to be equivalent if their difference is a hidden operator.}
\label{Fig: state operator equivalence}
\end{center}
\end{figure}

First we introduce an equivalence relation among states defined on a local patch $A$, e.g. the dashed circle  in Fig.~\ref{Fig: state operator equivalence}(a). We define two states to be equivalent if and only if their difference is a hidden state, i.e. if their difference is a state which is  orthogonal to the ground state reduced density matrix $\rho_A$ for patch $A$. This equivalence has a geometrical interpretation as will be clear below.

For example, the two states in Fig.~\ref{Fig: state operator equivalence}(a) are equivalent to each other, which we write as
\bea \label{eq:equivrel}
|C'_A\> \sim |C_A\>.
\eea
The two configurations are related by a deformation of the strands in the interior of $A$ without changing their connections. For any configuration \textit{outside} disc $A$, $C_{\bar{A}}$, the number of loops in the full configuration, which we denote $C_AC_{\bar{A}}$, is always the same as that in $C'_AC_{\bar{A}}$. Therefore these two configurations have the same amplitude in the ground state,
\bea 
\<\text{GS}|C'_AC_{\bar{A}}\> - \<\text{GS}|C_AC_{\bar{A}}\> = 0
\eea 
(for any $C_{\bar{A}}$), and so
\bea 
\rho_A (|C'_A\> - |C_A\>) = 0,
\eea 
which is the meaning of Eq.~\ref{eq:equivrel}.

More generally, if $C'_A$ has $n$ extra small loops compared to $C_A$, as in the example in Fig.~\ref{Fig: state operator equivalence}(b), then using the ground state wavefunction we have
\bea 
\rho_{A}\lf |C'_A\> - \bar{d}^{n}|C_A\> \ri = 0,
\eea 
hence
\bea 
|C'_A\> \sim \bar{d}^{n}|C_A\>.
\eea 

Let $\a_A$ denote an equivalence class of configurations inside $A$. The equivalence class is specified by the positions of strand endpoints on the boundary of the disc $A$, and the topological information about their connections by strands within $A$.
It is useful to pick a reference configuration, without any small loops, for each equivalence class: we call the corresponding state $|\a_A\>$. Any configuration $C_A$ in the equivalence class can be reduced to the reference configuration by removing the small loops and deforming the strands.

\begin{figure}[t]
\begin{center}
\includegraphics[width=0.35\textwidth]{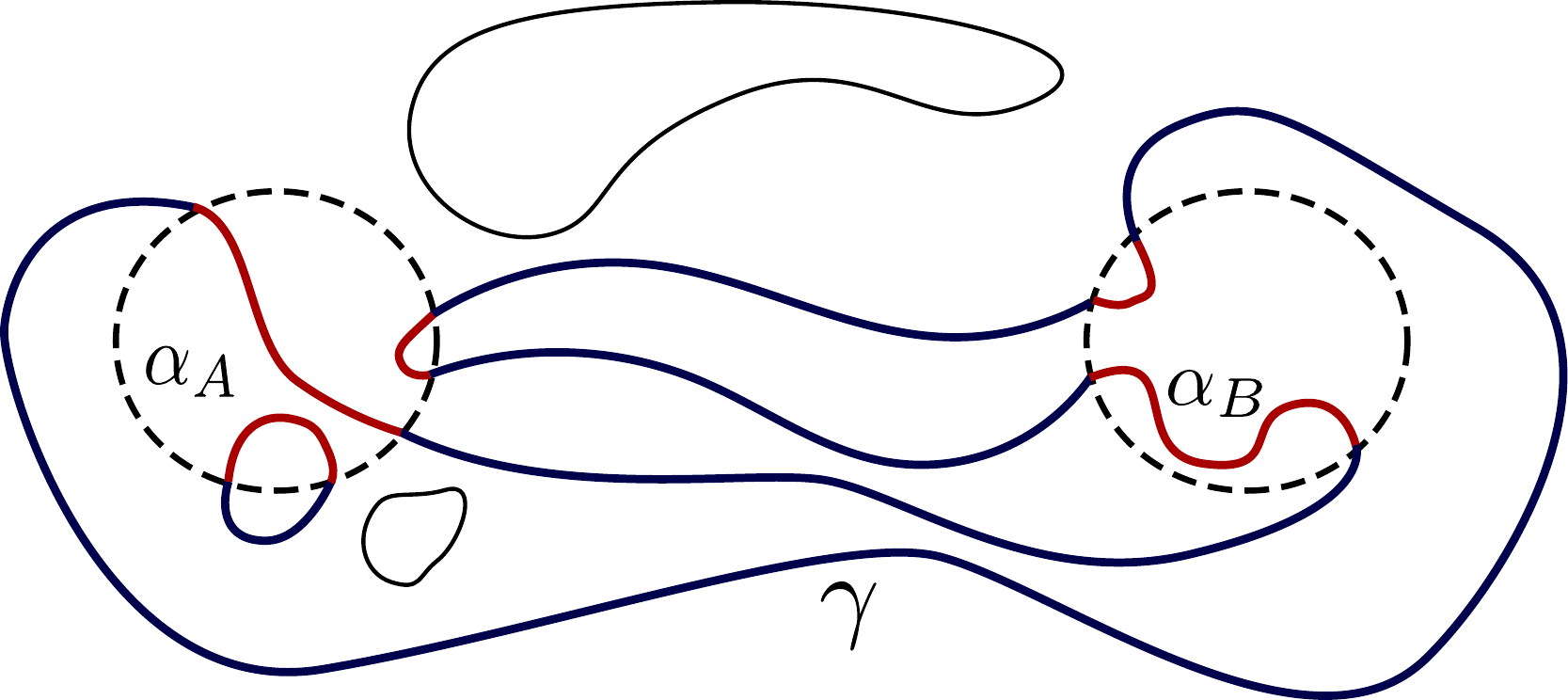}
\caption{Local operators and their 2-point functions. We label the connections of end points and the position of end points, by $\a_A$, $\a_B$ (red lines, inside disc A/B), and $\g$ (blue lines, outside disc A and B). In this configuration, the number of loops passing through the two discs, which we call $[\a_A\g\a_B]$, is 2. The number of loops completely outside the two discs, which we call $[C_{\g}]$, is also 2.}
\label{Fig:alphaGammaConnection}
\end{center}
\end{figure}

Next, we define two local operators on disc $A$ to be equivalent {($\mathcal{O}\sim \mathcal{O'}$)} if their difference is a hidden operator as defined in Sec.~\ref{sec:ffscalingoperators}. Since hidden operators have vanishing equal-time correlators with all operators outside $A$ (Sec.~\ref{sec:ffscalingoperators}), equivalent operators have exactly the same equal-time correlators with every operator outside disc $A$.

Given an arbitrary operator $O_A$ on disc $A$, written in terms of states in $A$ as
\bea
O_A = \sum_{C'_A,C_A}O_{C_A,C'_A}|C'_A\>\<C_A|,
\eea
we can always reduce it to an equivalent operator acting on the smaller space spanned by the reference states $|\a_A\>$ [this is illustrated in Fig.~\ref{Fig: state operator equivalence}(c)]
\bea 
O_A \sim \sum_{\a_A,\a'_A} \widetilde{O}_{\a'_A,\a_A}|\a'_A\>\<\a_A|,
\label{Eq:operatorEquivalence}
\eea
by subtracting hidden operators of the form
\begin{align*}
\lf |C'_A\> - \bar{d}^{n}|{\a_A'}\>\ri &\<C_A|, &
|C'_A\>&\lf \<C_A| - d^{m}\<\a_A| \ri,
\end{align*}
where $m$ and $n$ are the numbers of small loops in configurations $C_A$ and $C'_A$.

Thus, as far as equal-time correlators are concerned, we only need to study the operators on the right hand side of Eq.~\ref{Eq:operatorEquivalence}.

\subsection{Map between equal-time quantum correlators and classical probabilities}

Using the operator equivalence relation introduced in the last subsection, we now prove an important property of equal-time correlators of quantum operators: every equal-time quantum correlator is equal to a sum of classical `geometric' correlators.
We focus on 2-point functions, but the formalism applies to arbitrary $n$-point functions, in particular to 1-point functions on the disc (with specified boundary conditions) which we discuss in Sec.~\ref{subsec:1ptfunction}.

We calculate the 2-point function of two arbitrary operators $O_A$ and $O_B$ supported on disjoint regions $A$ and $B$.
Making use of  Eq.~\ref{Eq:operatorEquivalence}, we can take these operators to have nonzero matrix elements only between reference states $\ket{\alpha_A}$ (and similarly for $B$) of the type described above.

Just as we label the positions of endpoints on the boundary of disc $A$ and their topological connectivity inside $A$ by $\a_A$,
and similarly for $\a_B$ in disc $B$,
let us label the connectivity of these points outside both discs by $\g$: see Fig.~\ref{Fig:alphaGammaConnection}.\footnote{{$\g$ contains the information about the positions of the endpoints on the boundaries of $A$ and $B$, plus the purely topological information about how these endpoints are connected in the exterior.}}
We also use $C_{\g}$ to denote a  configuration outside the two discs that is in equivalence class $\g$  {(so that the sum $\sum_{C_{\g}}$ is implicitly restricted to configurations with connectivity $\g$).} 

Note that the number of loops \textit{that pass through $A$ and/or $B$} does not depend on the specific configuration $C_\g$, but only on its equivalence class $\gamma$; we denote this number by $[\a_A\g\a_B]$. We denote the number of loops completely \textit{outside} both $A$ and $B$ by $|C_\g|$. The total number of loops is then
\bea 
|C| = [\a_A\g\a_B] + |C_\g|.
\eea 
With this notation, we can organize the 2-point function,
\begin{align}
\<O_AO_B\>& = \f1Z\sum_{C,C'}\bar{d}^{C'}d^{C}\<C'|O_AO_B|C\>,
\end{align}
as follows (we use the fact that the operators only act on the reference states within the discs):
\begin{align}\notag
\<O_AO_B\>& =\sum_{\g} \Bigg[   \left(\frac{1}{Z} \sum_{C_{\g}}|d|^{2|C_\g|}\right) \times \\  
&\sum_{\substack{\a_A,\a_B,\\ \a'_A,\a'_B}} \bar{d}^{[\a'_A\g\a'_B]}d^{[\a_A\g\a_B]}\<\a'_A\a'_B|O_AO_B|\a_A\a_B\>  \Bigg].
\end{align}
We now make definitions for the two factors appearing in the square bracket above:
\begin{align}\label{Eq:2ptfunctionfinalexpression}
\<O_AO_B\>&=\sum_{\g} \, \tilde{p}_{\g}(\mathbf{r})\cdot\<O_AO_B\>_\g,
\end{align}
with ($\mathbf{r}$ is the separation of the two discs)
\ba\label{eq:tildepdef}
 \widetilde{p}_{\g}(\mathbf{r}) \equiv\frac1Z\sum_{C_{\g}}|d|^{2|C_\g|}
\end{align}
and
\ba \label{Eq:OAOBgammadefinition}
 \<O_AO_B\>_\g  \equiv 
 \hspace{-1.6mm}
 \sum_{\substack{\a_A,\a_B,\\ \a'_A,\a'_B}}
 \hspace{-1.6mm}
 \bar{d}^{[\a'_A\g\a'_B]}d^{[\a_A\g\a_B]}
 \<\a'_A\a'_B|O_AO_B|\a_A\a_B\>.
\end{align}
This rewriting separates out two conceptually different types of contribution to the correlator. 

First, the quantity $\widetilde{p}_{\g}(\mathbf{r})$ is a `geometrical correlator' in the classical ensemble which is independent of the operators.
More precisely, it is given (up to a positive and $\mathbf{r}$-independent proportionality constant which will not concern us in the following\footnote{{The  probability of $\gamma$ in the classical ensemble is $p_\gamma(\mathbf{r})=\sum_{C_\gamma}\sum_{C_A,C_B} |d|^{2|C_\gamma|+ 2[C_A\gamma C_B]}$, which is $\widetilde{p}_\gamma(\mathbf{r}) \times  \sum_{C_A, C_B} |d|^{2[C_A\gamma C_B]} $. The proportionality constant depends on $\gamma$ but not on $\mathbf{r}$.}})
by the classical probability $p_{\g}(\mathbf{r})$ that the boundaries of $A$ and $B$ are connected as $\g$ indicates. 

Second, $\<O_AO_B\>_\gamma$ is the `topological' part of the correlation function, conditioned on the connectivity outside being $\gamma$. This term depends on the matrix elements of the operators. The quantity $[\a_A\g\a_B]$ is well-defined only when the positions of end points match: for convenience, we have formally extended the definition to arbitrary $\a_A$, $\a_B$ and $\g$ by defining $d^{[\a_A\g\a_B]}=0$ when the positions do not match.

This seemingly tautological rewriting of correlation functions is the foundation for the rest of this section. It reduces an arbitrary quantum correlator to a sum of classical probabilities which know nothing about quantum operators. On the other hand, $\<O_AO_B\>_\g$ is purely topological, without any dependence on $\mathbf{r}$ or the shape of loops: it depends only on the number of loops passing through the discs, before and after the action of $O_A$ and $O_B$. The important features are that local operators $O_A$ and $O_B$ together can detect the nonlocal connectivity $\g$ of the underlying loops, 
and that the topologically distinct reconnection moves performed by operator $O_A$ can be `detected' by another operator $O_B$ far away. 
However, it remains to construct operators for which the sum in Eq.~\ref{Eq:2ptfunctionfinalexpression} simplifies, for example to a single $\gamma$. We do this next.

\subsection{Watermelon operators and topological types}
\label{sec:constructreconnectionoperators}

In this subsection, we study a class of quantum operators which we call quantum watermelon operators,  whose two-point functions are classical geometrical probabilities known as watermelon correlation functions.
We find a 1-to-1 correspondence between the leading watermelon correlators and the topological types of reconnection operators introduced in Sec.~\ref{sec:topologicaloperatorclassification}.

\begin{figure}[t]
\begin{center}
\includegraphics[width=0.45\textwidth]{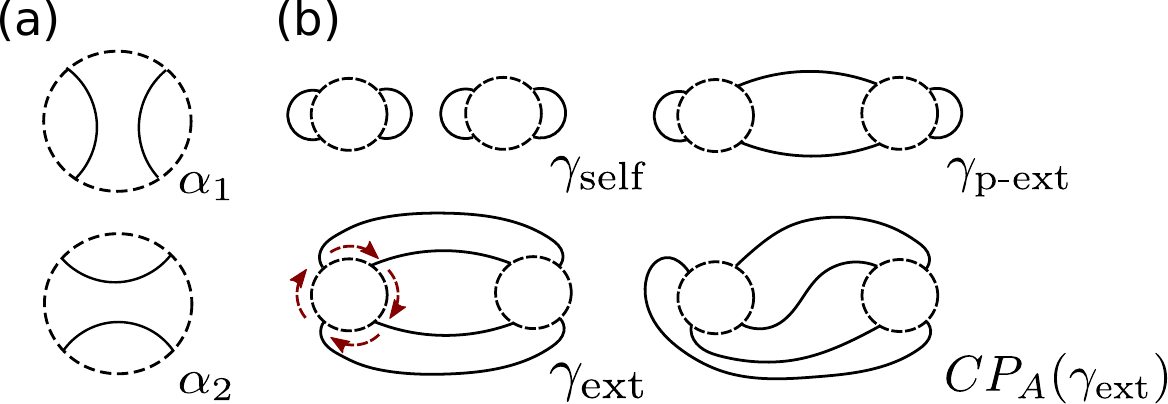}
\caption{Possible connections inside and outside discs with 4 end points. In Fig.~(a), we show the only two possible connections inside the disc with 4 end points, $\a_1$ and $\a_2$. These two connections give a 4 dimensional space of local operators with 4 end points, spanned by $|\a_i\>\<\a_j|$, $i,j=1,2$. In Fig.~(b), we show four possible connections outside the two discs. There are 24 possible connections in total, but each of them is equivalent to one of the first 3 connections in Fig.~(b) up to cyclic permutation of end points on disc $A$ (left) and on disc $B$ (right). For example, the 4th connectivity in Fig.~(b) is related to the 3rd by $CP_A$ (illustrated with red arrows), a cyclic permutation of end points on disc $A$. $CP_A$ also acts on the connections inside disc $A$, exchanging $\a_1$ and $\a_2$.}
\label{Fig: 2 loop reconnection CPA}
\end{center}
\end{figure}

In the notation of the last section, the classical watermelon correlators are the probabilities of the `fully extended' connections $\g$, for which discs $A$ and $B$ have an equal number of end points, say $2k$, and every end point on the boundary of $A$ is connected to an endpoint on $B$.
The two configurations in the bottom row of Fig.~\ref{Fig: 2 loop reconnection CPA}(b) are examples of fully extended configurations for $k=2$, whereas those in the top row are not fully extended. 

In general, for $2k$ \textit{fixed} end points on both disc $A$ and disc $B$, there are $2k$ fully extended connections (since loops cannot cross), related by the transformation generated by the cyclic permutations of end points on disc $A$ (which we denote by $CP_A$, as shown in the bottom of Fig.~\ref{Fig: 2 loop reconnection CPA}(b)). By Fourier transform, we can write the generalized watermelon correlators with spin $l = {-k+1},\dots, k$ as:
\bea 
p_{k,l}(\mathbf{r})\equiv\sum_{n=0}^{2k-1}e^{i\pi ln/k}
\,
\tilde{p}_{CP_A^n(\g_\text{ext})}(\mathbf{r}),
\eea 
where $\widetilde{p}(\mathbf{r})$ is defined in Eq.~\ref{eq:tildepdef}, and  $\g_\text{ext}$ can be any of the $2k$ fully extended connections, since repeated actions of $CP_A$ on $\g_\text{ext}$ generates all $2k$ connections. In the 2D CFT of the loop ensemble, the leading contributions to $p_{k,l}(\mathbf{r})$ at large $\mathbf{r}$ are correlators of defect operators with spin $l$ (mod $2k$).\footnote{For readers familiar with the Coulomb gas language, the spin-$l$ operator has magnetic charge $\pm k$ (emitting/absorbing $2k$ strands), and electric charge charge $l/k$ (meaning it gains a phase $e^{\pm i\pi l/k}$ when there is an extra loop surrounding it).}
In the scaling limit, they decay as 
\bea
p_{k,l}(z) \propto \frac{1}{z^{x_{2k,l} -l/2}\ \bar{z}^{x_{2k,l} + l/2}},\ l\neq k
\label{Eq:generalizedWatermelonSpin}
\eea 
where $z = x +i y$ is the complex coordinate of the point $\mathbf{r}$ and
\bea \label{eq:watermelonscalingdim}
x_{2k,l} = \frac{gk^2}{2} + \frac{l^2}{2gk^2} - \frac{(1-g)^2}{2g}.
\label{Eq:x2kl}
\eea 
The case $l=k$ is special because the spin-$k$ and spin-$(-k)$ correlators have the same scaling dimension. In that case,
\bea
p_{k,k}(z) \propto \text{Re} \, \frac{1}{z^{x_{2k,k} -k/2}\ \bar{z}^{x_{2k,k} + k/2}}.
\label{Eq:generalizedWatermelonkk}
\eea

With this understanding, we now want to look for a quantum operator $W_{k,l}$ in disc $A$ whose correlator is just the spin-$l$ watermelon correlator.
We restrict to operators in the space of states with precisely $2k$ endpoints on the boundary of the disc, and we take the positions of these endpoints to be fixed.

We define an operator $W_{k,l}$ on disc $A$ to be a quantum $2k$-leg watermelon operator if and only if
\begin{align}\notag
&\<W_{k,l}\cdot O_B\>_\g = 0 & & \text{for all $O_B$, and all $\g$ with} \\
& && \text{a self contact on disc $A$.}
\label{Eq:watermelonOperator1}
\end{align}
Further we impose
\be\label{Eq:watermelonOperator2}
CP_{A}(W_{k,l}) = e^{i\pi l/k}W_{k,l}.
\ee
$CP_A$ acts on operators by rotating the labels $\a_A$. For example, in Fig.~\ref{Fig: 2 loop reconnection CPA}(a), $CP_A(\a_1) = \a_2$.

The first condition, Eq.~\ref{Eq:watermelonOperator1}, ensures that only fully extended connections show up in correlators involving $W_{k,l}$ (see Eq.~\ref{Eq:2ptfunctionfinalexpression}); 
the second, Eq.~\ref{Eq:watermelonOperator2}, ensures that $W_{k,l}$ carries the desired spin (mod $2k$). 
(Each operator satisfying the first condition can be decomposed into watermelon operators in different spin channels.)
Since $\<O_AO_B\>_\g$ is just a polynomial of $d$ and $\bar{d}$ with total rank at most $2k$, the search for watermelon operators becomes a combinatorial task.

We prove the general results on watermelon operators, i.e. the number of them and the relation with topological types, in Appendix~\ref{Appendix:topologicaloperatorclassification}. Here we illustrate the idea with a few examples.

First consider $k=2$ (Fig.~\ref{Fig: 2 loop reconnection CPA}). With fixed end points, there are only 2 different connections inside disc $A$, $\a_1$ and $\a_2$. We can form 4 linearly independent quantum operators, and they belong to two topological types: those that reconnect 2 loops ($R_2$) and those that cannot reconnect loops ($D$). We organize them by the phase acquired under $CP_A$ (either 0 or $\pi$):
\bea
D_0 &=& |\ \raisebox{-0.3ex}{\rotatebox{90}{$\asymp$}}\ \>\<\ \ \raisebox{-0.3ex}{\rotatebox{90}{$\asymp$}}\ | + |\asymp\ \>\<\ \asymp|, \\
D_\pi &=& |\ \raisebox{-0.3ex}{\rotatebox{90}{$\asymp$}}\ \>\<\ \ \raisebox{-0.3ex}{\rotatebox{90}{$\asymp$}}\ | - |\asymp\ \>\<\ \asymp|,\\
R_{2,0} &=& |\ \raisebox{-0.3ex}{\rotatebox{90}{$\asymp$}}\ \>\<\ \asymp| + |\asymp\ \>\<\ \raisebox{-0.3ex}{\rotatebox{90}{$\asymp$}}\ |, \\
R_{2,\pi} &=& |\ \raisebox{-0.3ex}{\rotatebox{90}{$\asymp$}}\ \>\<\ \asymp| - |\asymp\ \>\<\ \raisebox{-0.3ex}{\rotatebox{90}{$\asymp$}}\ |.
\eea
By the general argument in Sec.~\ref{sec:topologicaloperatorclassification}, $D$ is an offspring type of $R_2$, and we expect to form a watermelon operator by subtracting a portion of $D$ from $R_2$. 

Outside the two discs, there are many different ways to connect the 8 end points. Fortunately, they all fall into three equivalent classes under $CP_A$ or $CP_B$ [represented by $\g_\text{self}$, $\g_\text{p-ext}$, and $\g_\text{ext}$ in Fig.~\ref{Fig: 2 loop reconnection CPA}(b)]. 
For operators which acquire a definite phase under $CP_A$ and $CP_B$, it is enough to compute $\<O_AO_B\>_\g$ for one $\g$ in each equivalence class. For example, for the operators from the above list that are invariant under $CP_A$/$CP_B$, i.e. those with $l=0$, we have according to Eq.~\ref{Eq:OAOBgammadefinition}:
\bea
\<D_0D_0\>_{\g_\text{self}} &=& (|d|^2 + |d|^4)^2,\\ 
\<R_{2,0}R_{2,0}\>_{\g_\text{self}} &=& |d|^4(d+\bar{d})^2,\\ \<D_0R_{2,0}\>_{\g_\text{self}}&=& (|d|^4+|d|^6)(d+\bar{d}),
\eea
and $\<O_AO_B\>_{\g_\text{p-ext}} = |d|^{-2}\<O_AO_B\>_{\g_\text{self}}$ for all $O_A,O_B \in \{D_0, R_{2,0}\}$.
We can then solve a linear equation to find the desired operator $W_{2,0}$  that satisfies Eq.~\ref{Eq:watermelonOperator1} (and similarly for $W_{2,2}$). By direct calculation, we find
\ba
\label{Eq:W20definition}
W_{2,0} & =  R_{2,0} - \frac{d+\bar{d}}{|d|^2+1} D_0, \\
W_{2,2} & = R_{2,\pi} + \frac{d-\bar{d}}{|d|^2-1} D_\pi.
\end{align}
Their correlators are:
\begin{align} \notag
\<W_{2,0}(\mathbf{r})W_{2,0}(0)\> &= 2\frac{|d|^6 - |d|^2(d^2+\bar{d}^2-1)}{|d|^2+1}p_{2,0}(\mathbf{r}), \\  \notag
\<W_{2,2}(\mathbf{r})W_{2,2}(0)\> &= 2\frac{|d|^6 - |d|^2(d^2+\bar{d}^2-1)}{|d|^2-1}
p_{2,2}(\mathbf{r}), \\
\<W_{2,0}(\mathbf{r})W_{2,2}(0)\> &= 0.
\end{align}

\begin{figure}[t]
\begin{center}
\includegraphics[width=0.4\textwidth]{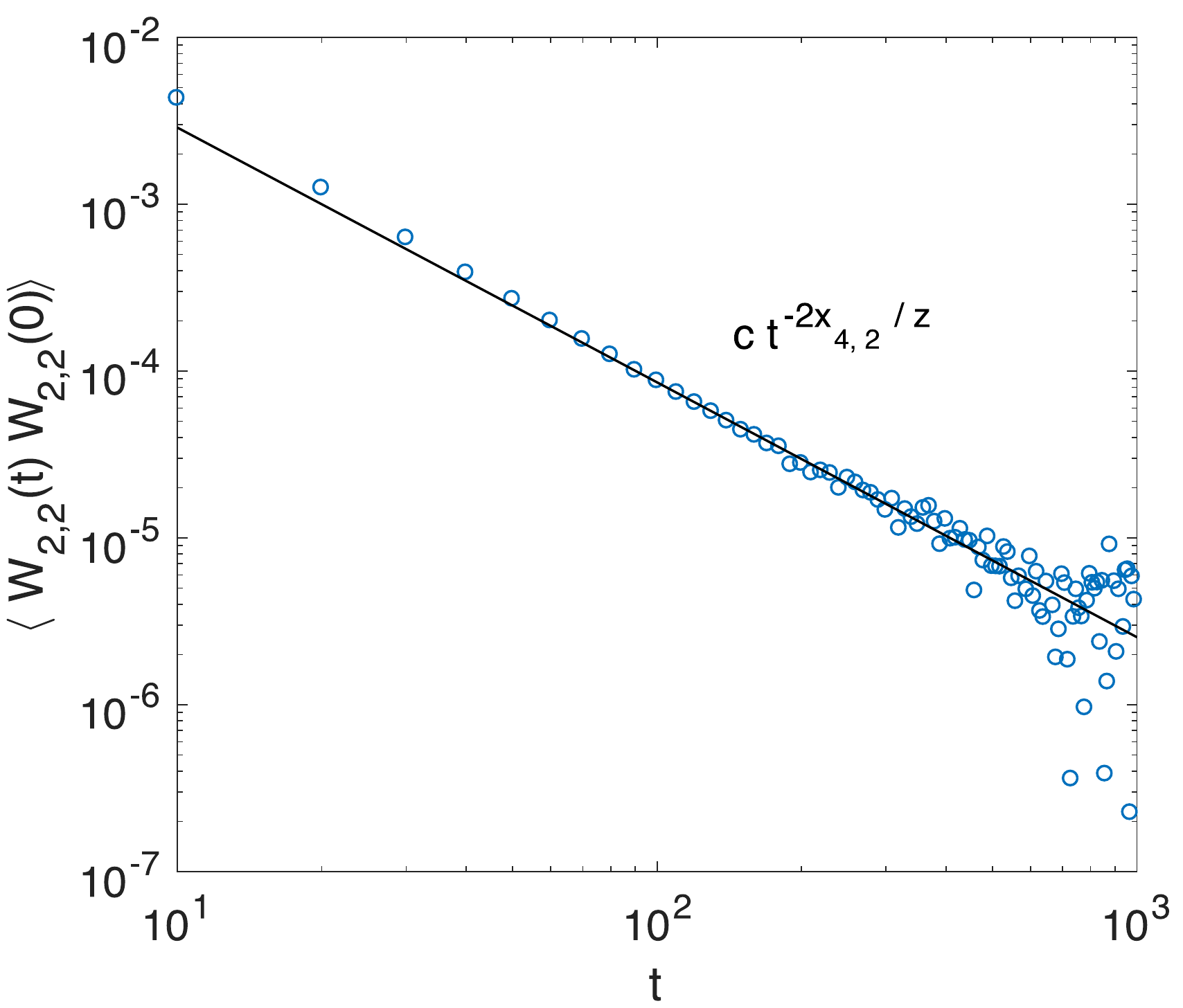}
\caption{2-point function of $W_{2,2}$. Results fit well to the black line,  which has the slope calculated from the theoretical scaling dimension $x_{4,2}$ and our numerical result ${z=3}$.}
\label{Fig:W22}
\end{center}
\end{figure}

The scaling dimensions and spins of $W_{2,0}$ and $W_{2,2}$ follow from
Eqs.~\ref{Eq:generalizedWatermelonSpin}-\ref{Eq:generalizedWatermelonkk}.
While $W_{2,0}$ corresponds to a single leading scaling operator with dimension $x_{4,0}$, which we call $\tilde{W}_{2,0}$, $W_{2,2}$ must contain 2 degenerate scaling operators with scaling dimension $x_{4,2}$ (according to Eq.~\ref{Eq:generalizedWatermelonkk}), having spins $\pm 2$, which we call $\tilde{W}_{2,+2}$ and $\tilde{W}_{2,-2}$.  (The general results in the following subsection show that we can always get a watermelon operator $\tilde{W}_{k,-k}$ by multiplying $\tilde{W}_{k,+k}$ with a diagonal operator nearby.) $\tilde{W}_{2,0}$ is the leading scaling operator with topological type $R_2$, i.e. the leading two-loop reconnection operator. $\tilde{W}_{2,\pm 2}$  is the leading scaling operator with spin $\pm2$ and topological type $R_2$.

For numerical results on their time-dependent correlation functions, see Fig.~\ref{Fig: PhysicalcorrelatorX4} and Fig.~\ref{Fig:W22}. Of course, we can multiply $\tilde{W}_{2,0}$, $\tilde{W}_{2,+2}$ and $\tilde{W}_{2,-2}$ by diagonal operators to get operators of the same topological type, which have different spins and larger scaling dimensions.

\begin{figure}[t]
\begin{center}
\includegraphics[width=0.45\textwidth]{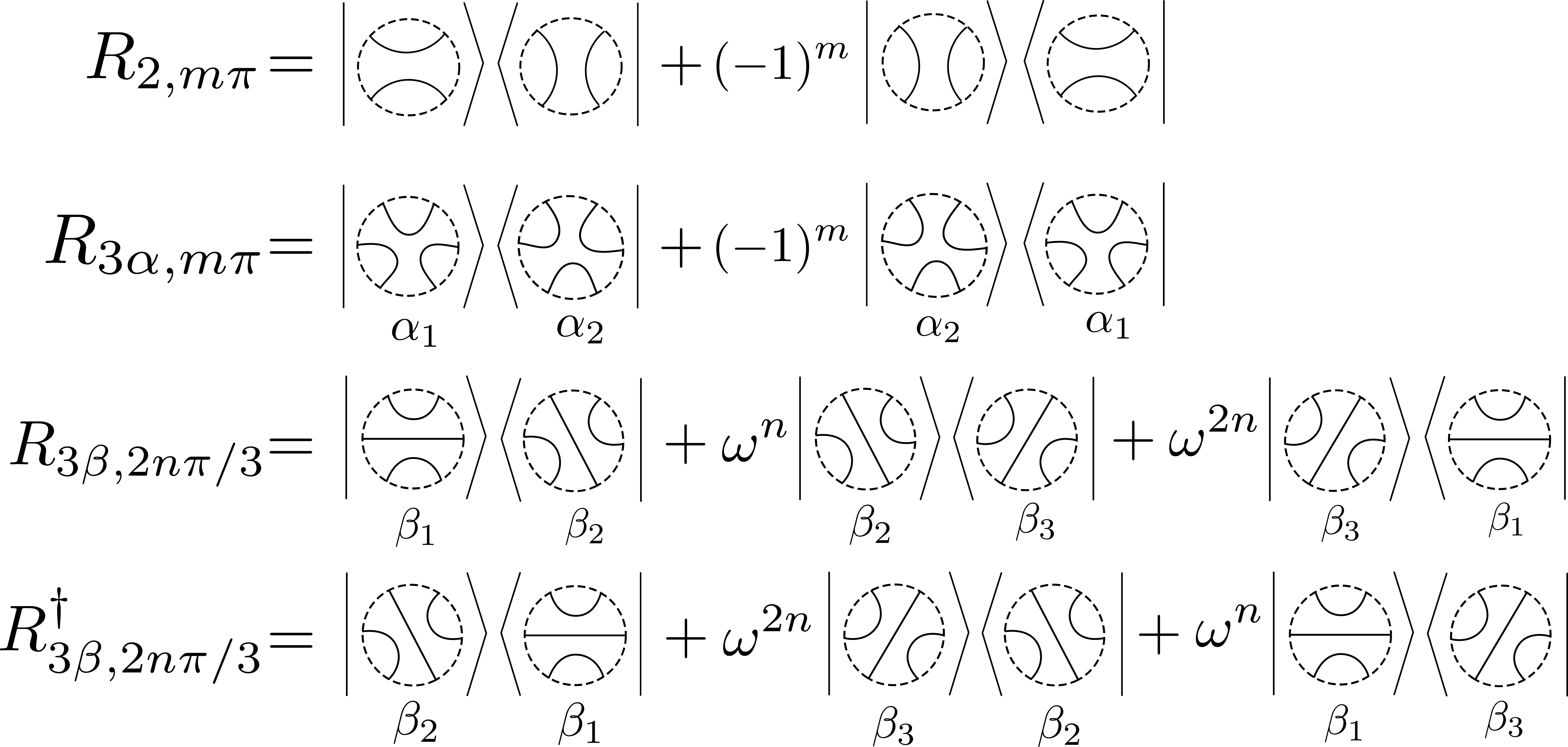}
\caption{2-loop and 3-loop reconnection operators, organized by their topological type and spin. $\a_1,\a_2$ and $\b_1,\b_2,\b_3$ labels different 3-strand configurations with fixed end points, ${m=0,1}$, ${n=0,1,2}$, ${\omega\equiv e^{2\pi i/3}}$. }
\label{Fig: 2 loop 3 loop reconnection}
\end{center}
\end{figure}

Solving for the watermelon operators for $k=2$ is not too hard, but for larger $k$, the number of linear constraints a watermelon operator must obey is apparently larger than the dimension of the corresponding operator space. A priori, there may not be any watermelon operator for a given $k$ and $l$. However, the constraints given by each $\g$ and $O_B$ are not linearly independent. In Appendix~\ref{Appendix:topologicaloperatorclassification}, we prove that the number of watermelon operators for a given $k$ equals the number of $k$-loop reconnection operators as defined in Sec.~\ref{sec:topologicaloperatorclassification} (acting on a disc with $2k$ fixed endpoints and nontrivially reconnecting them).\footnote{{Recall that a $k$-loop reconnection may involve irremovable spectators among the $k$ strands: these are strands that do not get reconnected but which cannot be pushed out of the disc because they are blocked by strands being reconnected. See Fig.~\ref{Fig: label spectator}(b) for an example.}}
In particular, the number of leading spin-0 quantum watermelon operators is just the number of distinct topological types for $k$-loop reconnection operators!

\begin{figure}[t]
\begin{center}
\includegraphics[width=0.35\textwidth]{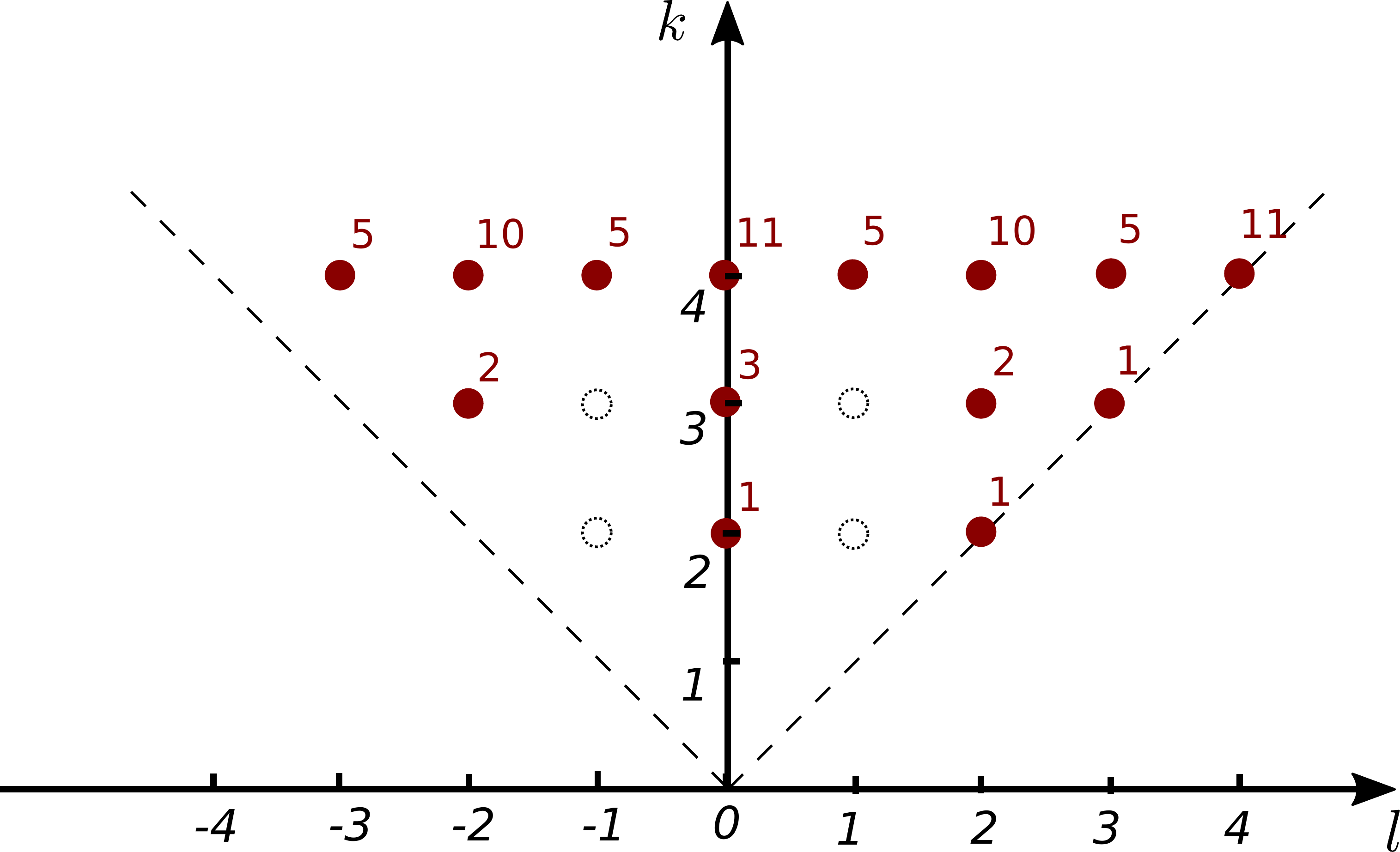}
\caption{$k$--$l$ lattice for quantum watermelon operators. Each red dot represents an allowed scaling dimension $x_{2k,l}$ (Eq.~\ref{eq:watermelonscalingdim}) for local quantum operators, with a number indicating the topological operator degeneracy at this spin and scaling dimension. Degenerate operators effect topologically different reconnection moves. Each dashed dot represents a classically allowed watermelon correlator that is \textit{not} realized by any quantum local operator.
This figure counts only the leading quantum watermelon operators, other scaling operators of the same topological type but with higher scaling dimensions (Sec.~\ref{subSec:operatorcontent}) are not included. For the leading operators, the index $l$ gives the spin under spatial rotations (operators with the same $l$ may transform differently under parity/time reversal, as discussed in the text).}
\label{Fig: m spin lattice}
\end{center}
\end{figure}

For example, there are three different topological types for 3-loop reconnection operators: $R_{3\a}$, $R_{3\b}$, and $R_{3\b}^{\dagger}$. 
These are shown in the third line of Fig.~\ref{Fig: topological type partial order}.
We can organize them by their spin $l$ (mod 6), giving 8 operators. Three of them are rotation-symmetric, as shown in Fig.~\ref{Fig: 2 loop 3 loop reconnection}.\footnote{Even though there are 3 rotation-symmetric Hermitian 3-loop reconnection operators, namely $R_{3\a,0}$, $R_{3\b,0} + R_{3\b,0}^{\dagger}$, and $i(R_{3\b,0} - R_{3\b,0}^{\dagger})$, the last one is odd under both reflection and time reversal.}
The general counting implies that we can make a watermelon operator of the corresponding spin from each of the 8 operators by subtracting operators of their offspring types. The 2-point functions involving these 8 operators can be diagonalized. The same procedure works for all $k$. This analysis confirms the general physical argument for topological types in Sec.~\ref{sec:topologicaloperatorclassification}. We summarize the results for $k\le 4$ in Fig.~\ref{Fig: m spin lattice}. For a proof of the general result see Appendix~\ref{Appendix:topologicaloperatorclassification}.

\subsection{Classification of local operators}
\label{subSec:operatorcontent}

In this subsection, we use the understanding of watermelon operators and the 2D CFT description of the loop ensemble to give a classification of quantum local operators up to hidden operators.

Eq.~\ref{Eq:operatorEquivalence} and the results of the previous subsection imply that
any lattice operator can be written, up to hidden operators, as a sum of terms, each of which is a product of a diagonal operator and, next to it in space, a $k$-loop reconnection operator.\footnote{To see this, use Eq.~\ref{Eq:operatorEquivalence} to reduce the operator to matrix elements involving reference states. It is sufficient to consider one matrix element at a time:
\be
\mathcal{O}_A = \ket{\alpha}\bra{\alpha'}
\ee
The reference configurations $\alpha$ and $\alpha'$ have the same endpoints at the boundary of the disc. Some of the endpoints are connected to `removable strands' as defined in Sec.~\ref{sec:topologicaloperatorclassification} (which are not reconnected and which could be pushed out of the disc without being blocked by the strands that are reconnected). The operator above is the product of a diagonal operator acting on these strands together with a reconnection operator acting on the other strands.}
Using the result of the previous subsection and Appendix~\ref{Appendix:topologicaloperatorclassification}, we can then write the $k$-loop reconnection operator as the sum of a $k$-loop watermelon operator satisfying Eq.~\ref{Eq:watermelonOperator1} and an operator that reconnects less than $k$ loops. Repeating this process, we can reduce every operator, up to hidden operators, to a sum of products of watermelon operators and diagonal operators.

Therefore we have a basis for (non-hidden) operators which is of the form $W_{k,q,l}\times \operatorname{Diag}$, where $W_{k,q,l}$ is a watermelon operator satisfying Eq.~\ref{Eq:watermelonOperator1}-\ref{Eq:watermelonOperator2}, and $q$ labels its topological type. Recall that there are multiple topological types for a given $k$. By definition, the type is unchanged by multiplying the diagonal operator nearby.

The leading scaling operator in the product ${W_{k,q,l}\times \operatorname{Diag}}$ is the same as that contributing to the lattice watermelon operator $W_{k,q,l}$ itself. (In general, we call the continuum scaling operator $\tilde W_{k,q,l}$.) This result is verified by directly calculating the 2-point functions of $W_{k,q,l}\times \operatorname{Diag}$ in disc $A$ and another operator in disc $B$. By Eq.~\ref{Eq:2ptfunctionfinalexpression} and the definition of the watermelon operator, it is not hard to see that the nonzero contributions to this correlator come from only the classical probabilities of configurations where the strands in the operator $W_{k,q,l}$ connect to irremovable strands of the operator in disc $B$, which requires disc $A$ and $B$ to be connected by at least $2k$ strings. The leading behavior of these probabilities is just the watermelon correlator in Eq.~\ref{Eq:generalizedWatermelonSpin} and Eq.~\ref{Eq:generalizedWatermelonkk} with the same $l$.

The leading scaling operator $\tilde{W}_{k,q,l}$ has spin $l$ (mod $2k$), known from the explicit form of the 2-point function. 
We note however that subleading operators in the product $W_{k,q,l}\times \operatorname{Diag}$ with a given $l$ may have spin differing from $l$. This is true even for the watermelon operator $W_{k,q,l}$ itself, because $l$ describes the phase acquired by the reconnection operator when the patterns of connectivity are cyclically permuted, and this permutation is not in general equivalent to a spatial rotation of the operator.

In order to determine the scaling operator content of a general lattice operator, now
consider the classical mapping for equal-time correlators involving an operator of the form $W_{k,q,l}\times \operatorname{Diag}$. 
It is convenient to imagine attaching arrows to some of the strands as in the  Coulomb gas approach to the loop models \cite{jacobsen2009conformal}. 

We may define operators on the lattice  which emit $2k$ outgoing strands from a disc, with fixed endpoints.
Let $D_{k,l}$ be such a lattice operator, and  $D_{-k,l}$ the corresponding operator which absorbs $2k$ incoming strands.
The generalized watermelon correlator $p_{k,l}(\mathbf{r})$ described above is the 2-point function of $D_{k,l}$ and $D_{-k,l}$, with outgoing and incoming arrows correspondingly.
This assignment captures the defining property of quantum watermelon operator, that the correlator is zero whenever there is a self-contact among its own strands (in the classical language, a conflict of orientations). 

The quantum two-point function for lattice operators $W_{k,q,l}$ and $W_{k,q',l}$ 
is given by
\bea
\<W_{k,q,l}(\mathbf{r})W_{k,q',l}(0)\> =
f(q,q')\<D_{k,l}(\mathbf{r})D_{-k,l}(0)\> 
\label{Eq:watermelonvsdefect}
\eea
where $f(q,q')$ is a function of the topological types.
For more general quantum operators in our basis, which can be written as $W_{k,q,l}\times 
\text{Diag}$, the classical mapping gives the two point function of the product of a defect operator and a local classical operator (a function of the local loop configuration).

After applying the OPE, the latter is a combination of the leading defect scaling operator for a given $k$, $l$, and subleading defect operators which have the same label $k$ and $l$.\footnote{In the Coulomb gas language, we can define $k$ as number of outgoing arrows minus the number of incoming arrows, and define $l$ as the product of the electric charge and magnetic charge mod $2k$. Both are preserved by multiplying local operators.}
Thus, we can expend the operator $W_{k,q,l}\times 
\text{Diag}$ as a sum of the leading operator $\tilde{W}_{k,q,l}$ and subleading operators, each of which has a 2-point function corresponding to that of a subleading classical defect operator. We can label these subleading quantum scaling operators in the product $W_{k,q,l}\times 
\text{Diag}$ by the corresponding subleading defect operators, and the topological type $q$ (which is preserved by multiplying with the diagonal operator). Thus we have the expansion
\bea 
W_{k,q,l}\times \operatorname{Diag} =  c_0 \tilde{W}_{k,q,l} + \sum_{s}c_s\tilde{W}_{k,q,l,s},
\eea
where $s$ labels different subleading operators. The three labels $k,l,s$ uniquely specify a defect scaling operator in the Coulomb gas; quantum scaling operators with the same $k,l,s$ but different topological types map to the same defect operator.

Since operators of the form $W_{k,q,l}\times 
\text{Diag}$ form a basis of non-hidden operators, we can now express each quantum operator $O$ (up to hidden operators) as a sum of diagonal scaling operators, which we call $\tilde{D}_{s}$, and the operators $\tilde{W}_{k,q,l,s}$,
\bea 
O\sim \sum_{s} o_{s} \tilde{D}_{s} + \sum_{k,q,s'}\sum_{l = -k+1}^{k} o_{k,q,l,s'}\tilde{W}_{k,q,l,s'},
\eea 
where $o_{s}$ and $o_{k,q,l,s'}$ are complex coefficients.

In conclusion, up to hidden operators, each quantum scaling operator is labeled by its topological type and the corresponding classical local/defect scaling operator.

\begin{figure}[t]
\begin{center}
\includegraphics[width=0.3\textwidth]{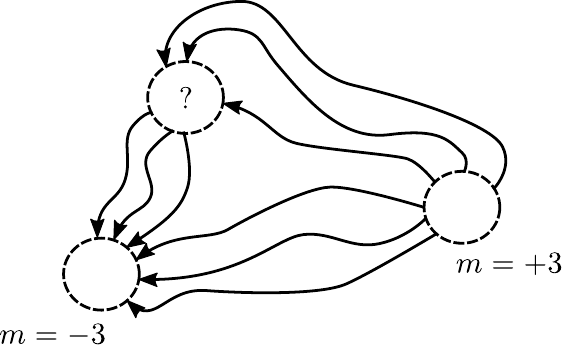}
\caption{Frustration in mapping watermelon operators to classical defect operators. When there are 2 watermelon operators, we always assign outgoing arrows for one operator, and incoming arrows for the other, but the trick fails when there are 3 or more reconnection operators.}
\label{Fig: 3 pt frustration}
\end{center}
\end{figure}

Last, we note that despite the close relation between the quantum operator spectrum and the classical defect operator spectrum in the Coulomb gas, and the correspondence used above between the quantum OPE involving only \textit{one} reconnection operator and the classical Coulomb gas OPE, the quantum OPE involving \textit{two} reconnection operators is dramatically different from the OPE of classical defects. 
For example, in the Coulomb gas language the label $k$ on the defect operator is the number of strands emitted/absorbed, so is addictive in the OPE. However, in the quantum OPE, multiplying two 2-loop reconnection operators can give the 5-loop reconnection operator with a spectator strand ($2+2\rightarrow5$), shown in Fig.~\ref{Fig: label spectator} (bottom right). We can also see this difference clearly in the calculation of 3-point functions involving 3 reconnection operators (Fig.~\ref{Fig: 3 pt frustration}).

\subsection{Topological types of local operators, revisited}
\label{subsec:1ptfunction}

In the previous subsections we introduced watermelon operators whose 2-point functions are  classical watermelon correlators, and we prove in Appendix~\ref{Appendix:topologicaloperatorclassification} that they are in one-to-one correspondence with topological types of local operators. 
In this subsection, we give an alternative understanding of the topological types using 1-point functions.
This also gives the leading scaling dimensions of scaling operators of each topological type.

\begin{figure}[t]
\begin{center}
\includegraphics[width=0.3\textwidth]{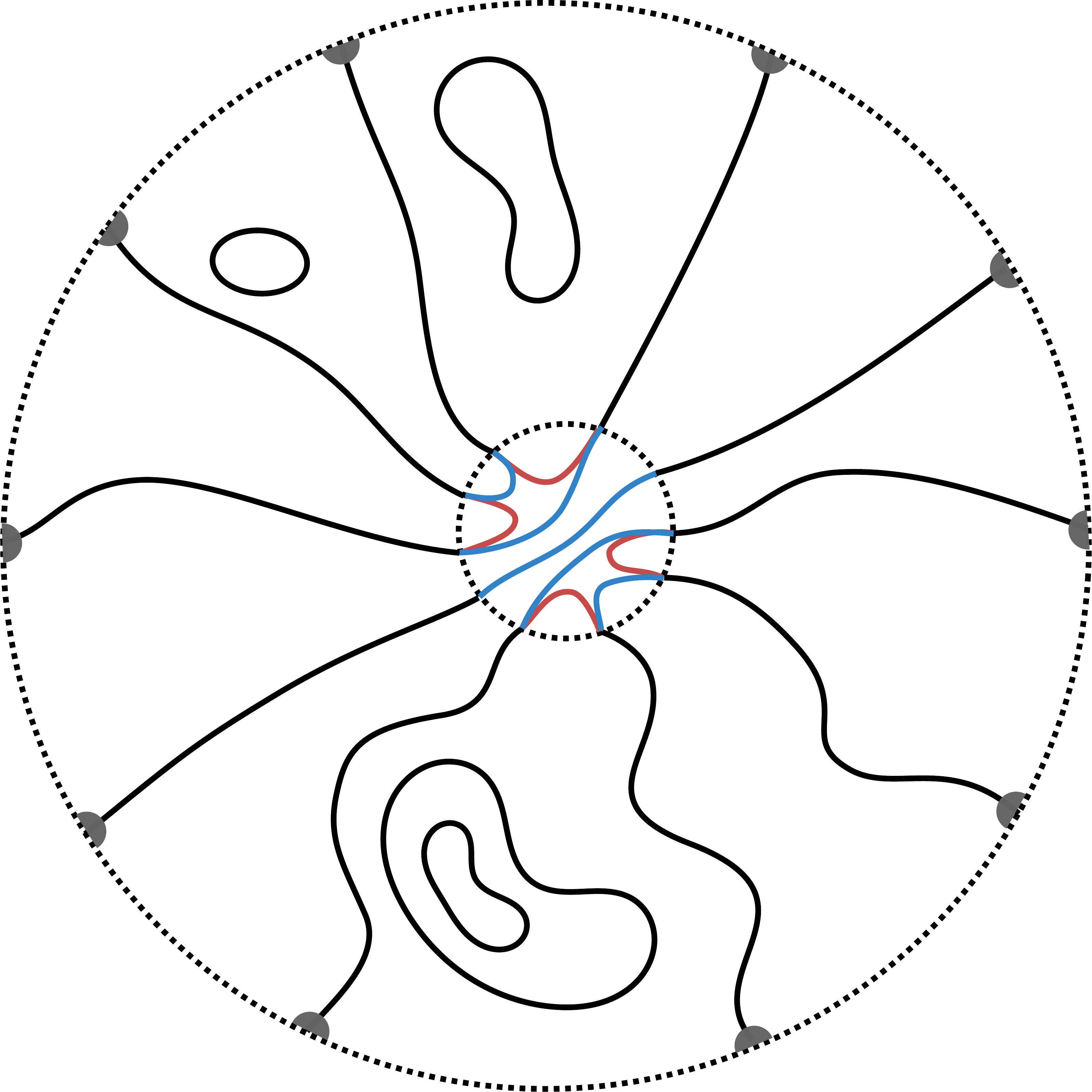}
\caption{1-point function on disc with boundary condition specified in text. In the small circle, we show an operator with the 5-loop reconnection type $(\a,\a')$. The blue lines represent the connection in $\a$, and the red lines show the change of connection.}
\label{Fig:1ptfunction}
\end{center}
\end{figure}

Consider ground states on a disc with radius $R$. We choose the boundary condition such that there are $2k$ fixed, equally spaced points on the boundary each emitting one strand (paired inside the disc into $k$ strands), and such that loops cannot touch the boundary elsewhere.
Since the Hamiltonian cannot reconnect loops, the connection of the $k$ strands, which we call $\a$, is a conserved quantity. Configurations with connectivity $\a$, which we label by $C_{\a}$, span an invariant space of the Hamiltonian, and there is a unique normalised ground state
\bea 
|\a\> \propto \sum_{C_{\a}}d^{|C_{\a}|}|C_\a\>
\eea
in this subspace.

We now insert an operator $O$ with $k$-loop reconnection type $(\a,\a')$ 
(defined in Sec.~\ref{sec:topologicaloperatorclassification}: see Fig.~\ref{Fig:1ptfunction} for example), 
acting on a much smaller disc, of order one size, at the center of the large disc. This operator has a nonzero off-diagonal matrix element $\<\a|O|\a'\>$ between different ground states. 
This matrix element is small since it comes from matrix elements $O_{C_\a,C'_{\a'}}$ involving configurations $C_\a$, $C'_{\a'}$ where all $k$ strands connecting to the boundary pass through the small disc $O$ acts on. This event has probability $p\sim 1/R^{x_{2k}}$. Thus
\bea \label{eq:1ptfn}
\<\a|O|\a'\> \sim 1/R^{x_{2k}}.
\eea 
After coarse-graining, operator $O$ also effects reconnections of lower types. Equivalently, it can connect other pairs of ground states, by acting on a subset of the $k$ strands connecting to the boundary, together with small loops. 
This effect gives other matrix elements of $O$ between ground states that decay as a smaller power. Subtracting operators of these lower topological types from $O$, we may obtain a scaling operator for $k$-loop reconnection, which we see from Eq.~\ref{eq:1ptfn} should have scaling dimension $x_{2k}$.

\section{Related critical models}

\subsection{Loop model with Jones-Wenzl projector}
\label{Sec: Jones Wenzl}

As discussed in Sec.~\ref{subsec: lowlying operator}, the loop model we focus on has many relevant perturbations: the 2-loop and 3-loop reconnection operators, and the topologically trivial operators known from the classical model. Starting from a generic spin system, which does not have the dynamical constraint, we would need to tune all of the relevant perturbations to zero to achieve this multicritical point. However, at $d=\pm\sqrt{2}$, we can add a specific 3-loop reconnection operator back while keeping the model critical. 
This is the `Jones-Wenzl projector' \cite{freedman2004class,freedman2005line,freedman2003magnetic,fendley2007quantum,fendley2008topological}.
This operator performs RG-relevant reconnection moves, leading to a new RG fixed point, but it preserves the form of the ground-state wavefunction.

Since the ground state is preserved, our results for the power-law form of equal-time correlators still hold. This implies that the new fixed point remains gapless, and that the scaling dimensions of non-hidden operators remain the same. However, a priori we know little about the \textit{dynamics} at the new fixed point. Intuitively we might expect $z_\text{new} <3$:  the 3-loop reconnection gives new ways to connect loop configurations so may speed the dynamics up. We will give a \textit{lower} bound on $z$ below.

\begin{figure}[t]
\begin{center}
\includegraphics[width=0.45\textwidth]{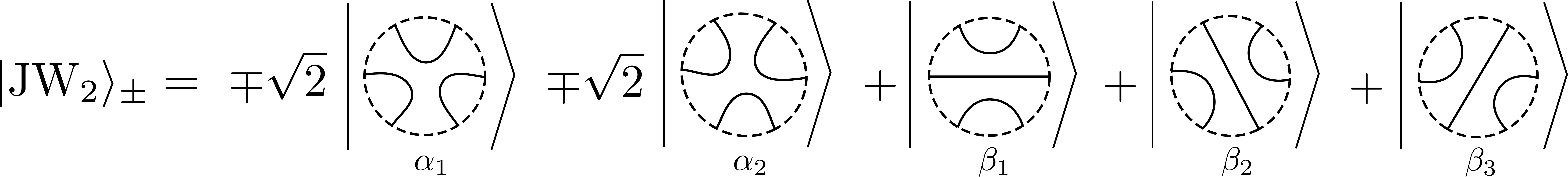}
\caption{The \textit{Jones-Wenzl state} for $d=\pm\sqrt{2}$.}
\label{Fig: JW 3 loop}
\end{center}
\end{figure}

The special property of $d=\pm\sqrt{2}$ can be traced back to the new hidden state $|\text{JW}_2\>_{\pm}$, a state on a patch that contains configurations with different connections but is orthogonal to the local reduced density matrix: see Fig.~\ref{Fig: JW 3 loop}. 

Now we can add a term $|\text{JW}_2\>_{+}\<\text{JW}_2|_{+}$ to the Hamiltonian~\cite{freedman2005line,freedman2008lieb} at $d=\sqrt{2}$. At least on the sphere,
\bea 
\rho|\text{JW}_2\>_{+} = 0
\eea 
where $\rho$ is the ground-state reduced density matrix on the disc. Thus the ground state remains a ground state of the new Hamiltonian, and the new Hamiltonian is still frustration-free. 
However, we know that at the critical point without reconnection, $|\text{JW}_2\>_{+}\<\text{JW}_2|_{+}$ has scaling dimension $x_6 = 9/2$ at $d=\sqrt{2}$ and RG eigenvalue ${y_6\simeq 1/2}$, i.e. is relevant. So we expect a new fixed point with different dynamics. 

Despite the fact that equal-time correlators still map to classical correlators,
the new Hamiltonian gives an example of a frustration-free Hamiltonian without a Markovian correspondence for its dynamics. (See Appendix~\ref{Appendix: Quantum Markovian correspondence} for explanations.)
 
Since $|\text{JW}_2\>_{+}$ is a `forbidden state' at $d=\sqrt{2}$,
operators of the form 
\bea 
|\text{JW}_2\>_{+}\<C| + |C'\>\<\text{JW}_2|_{+}
\eea 
are all hidden operators. For example, in addition to the Jones-Wenzl projector itself, the other 3-loop reconnection watermelon operator symmetric under parity is also a hidden operator at $d=\sqrt{2}$. (See Appendix~\ref{Appendix:topologicaloperatorclassification} for details.)
 
Usually, we cannot say much about the new critical point by studying the unstable critical point that flows to it. However, in this case, we can constrain the new dynamical exponent using the old scaling exponents. The argument is in fact more general, and applies to any critical frustration-free model that is perturbed by a relevant operator that leaves the ground state unchanged.
 
Consider the expectation value of a local Jones-Wenzl projector in the first excited state, which we call $|\text{ex}\>$, of the original Hamiltonian (\textit{without} reconnection) on the sphere. This matrix element must be small, since the first excited state is locally similar to the ground state, in which the  expectation value vanishes identically. Furthermore, by standard coarse-graining, we expect this matrix element to scale as $L^{-x_6}$. Thus
\bea 
\<\text{ex}|\delta H|\text{ex}\>\propto L^{2-x_6},
\eea 
where $\delta H$ is the sum of Jones-Wenzl projectors at different positions. Since $|\text{ex}\>$ is orthogonal to the new ground state, which is just the old ground state, its energy must be greater than or equal to the energy of the first excited state of the \textit{new} Hamiltonian. This argument gives a bound on the new dynamical exponent:
 \bea 
 z_\text{new} \ge \min\{z, \, x_6 -2 \} = x_6 - 2
 \eea 
where we used the relevance of the perturbing operator. Therefore
\be
z_\text{new} \geq \frac{5}{2}.
\ee

The new critical point is likely to have fewer relevant perturbations than the original one, i.e. to be more stable. In particular the JW projector itself is no longer a relevant perturbation.

On the torus, the low-lying excited states with different winding numbers are lifted by the Jones-Wenzl projector, leaving only 9 states having exactly zero energy under the frustration-free Hamiltonian~\cite{freedman2004class,freedman2005line}. For the same reason discussed in Sec.~\ref{sec:gsd}, for a generic Hamiltonian in the universality class, these 9 states still have energies well below the continuous excitation spectrum. In this sense, we can call them \textit{degenerate ground states}.

This degeneracy depends on the identification of the local physical degrees of freedom. If we interpret the loops as Ising domain walls for $\tau^z$ spins living on the hexagons, instead of chains of $\sigma^z$ down-spins living on the edges, the ground state degeneracy is reduced to 3. See Appendix~\ref{Appendix: GSD} for a self-contained discussion of the ground state degeneracy.

We have discussed the new critical point at $d=\sqrt{2}$. Similarly, we can add $|\text{JW}_2\>_{-}\<\text{JW}_2|_{-}$ to the Hamiltonian at $d=-\sqrt{2}$ to get another critical point. The two critical theories are related by the nonlocal unitary transformation $U$ with matrix elements $U_{C,C'} = (-1)^{|C|}\delta_{C,C'}$; therefore they have the same excitation spectrum. Moreover, even though the transformation $U$ is nonlocal, it maps local operators to local operators. Thus the operator spectrum and the OPE coefficients of the two theories are also exactly the same. However, as discussed in Ref.~\cite{freedman2004class}, when the theories are placed on the torus, the degenerate ground states have different topological properties. The topological difference should be clarified further in the future, but it seems to suggest the two critical theories cannot be adiabatically connected.

\subsection{Dilute critical points}
\label{subsec:dilutecriticalpoint}

The classical loop ensemble may be generalized to an ensemble with a weight $x$ per unit length of loop:
\be
Z = \sum_C |d|^{2|C|} x^\text{length}.
\ee
For a given $0<|d|^2<2$, there is a critical value $x_c$ \cite{jacobsen2009conformal} (Appendix~\ref{appendix:dilutecriticalpoint}).
For $x>x_c$ the loop model is in the `dense' phase. 
This critical phase includes the value $x=1$ that we have so far restricted to, and the critical exponents take the values we have discussed throughout this phase. The precise value of $x$ in this phase does not matter --- changing $x$ corresponds to an irrelevant perturbation of the dense fixed point.
If $x$ is smaller than $x_c$ then the loop model is not scale invariant (there are only short loops). However if $x$ is tuned exactly to $x_c$ there is a new universality class with different exponents. This is known as the `dilute' critical point. By a simple modification of the flip operator in our quantum Hamiltonian, we can generalize the quantum loop model to the dilute case too.

As we discuss in Appendix~\ref{appendix:dilutecriticalpoint}, the dilute critical points are in fact continuously connected to the dense ones via the extremal point $|d|^2 = 2$.
This observation together with the superuniversality of $z$ established in Sec.~\ref{sec:superuniversality} allows us to strengthen the analytical bound on the dynamical exponent (Sec.~\ref{sec:zbound}) to $z\geq 2.6\dot 6$ as mentioned in Sec.~\ref{Sec: introduction}.

\subsection{Towards other topologically constrained models}
\label{sec:otherensembles}

We have introduced the idea of a dynamical topological constraint that is preserved under the renormalization group.
This idea is more general than the class of quantum loop models studied in this paper, and it would be interesting to study further examples.
Here we mention some miscellaneous models which would be partially tractable thanks to a classical correspondence. 
It would of course also be interesting to look for topologically constrained universality classes that do not have a classical correspondence. 
(There are also applications of the results of this paper to one-dimensional quantum systems, which we will discuss separately \cite{1dforthcoming}.)

First we could `decorate' the loops with additional degrees of freedom. The most trivial change of this type is to give the loops a `colour' degree of freedom, taking $q>1$ values, that is uniform along the length of the loop. 
We can define the dynamics such that, in the Markov language, the colour degrees of freedom are simply carried along for the ride by the loops, being assigned randomly at each birth event and having no feedback on the dynamics of the loops' geometry. 
With this choice, correlations that do not involve the colour variable are unchanged from the $q=1$ case, so this is in a sense a completely trivial modification. 
Therefore it is surprising that this modification introduces a new local operator with a  scaling dimension that was not present before (at least not among non-hidden operators). 
This operator measures the local colour of a link. 
It has scaling dimension $x_{2,0}$, since the probability of two links having the same colour is simply related to the probability of their being on the same loop.\footnote{A similar modification is to decorate the loops with arrows on the links, consistently oriented along the length of the loop. This gives other new operators \cite{cardy1994mean}. It also allows a height field to be defined, but this is not obviously useful in constructing a continuum description.}

The loops we have discussed are unable to pass through each other. This could be relaxed, while retaining the no-reconnection constraint, in models of intersecting loops \cite{nahum2013loop,jacobsen2003dense,martins1998intersecting,lawler2004brownian}.
We can also consider loops in three dimensions (where scale-invariant classical ensembles can be achieved by tuning a parameter \cite{nahum2013phase}), and models of membranes in three dimensions at appropriate critical points. 
In all these cases the basic logic of Sec.~\ref{sec:topologicaloperatorclassification} allows a topological classification of operators.

\section{Outlook}

Our aim has been to characterize the scaling structure of strongly-interacting multicritical points for fluctuating loops that arise in lattice spin systems. 
This led to concepts that may be useful in other settings.
These include the idea of the topological operator classification, and the method for constructing the operator spectrum for models with a topological constraint, as well as general properties of frustration-free models.
Our explicit results for correlation functions and exponents show that the loop models obey scaling forms dictated by scale-invariance,
but with an unusual spectrum of scaling operators, and an unusually large dynamical exponent $z$. 

Having understood a lot about these gapless models, we come back to the question that we started with: are there useful field theories for them?
A continuum Lagrangian assigns amplitudes to spacetime configurations of fluctuating fields. 
In the loop basis, a spacetime configuration of the present model is made up of closed worldsurfaces in  2+1D. 
Unconstrained surfaces can be mapped to field theory by regarding them as domain walls or level surfaces for an appropriate field, or as worldsurfaces of flux lines in a gauge theory.  
However, the loop models do not allow reconnection (in spacetime, a reconnection event is a saddle point  tangent to the spatial plane) and we have pointed out in Sec.~\ref{sec:topologicaloperatorclassification} that this constraint is also a property of the RG fixed point.
Further, the constraint cannot simply be imposed `softly', since 2-loop reconnection is RG-relevant, and will become important in the IR if given a nonzero amplitude in the UV. 
We are not aware of a means of incorporating this constraint in a field theory.
Could the constraint mean that the loop model realizes RG fixed points have no useful continuum Lagrangian? 

We cannot answer this question.\footnote{A nonrelativistic gauge theory was originally proposed for the loop models \cite{freedman2005line}, but at that time it was believed that the dynamical exponent was $z=2$, which is now ruled out.} However we note that any such Lagrangian would have to reproduce an operator spectrum with topological quantum numbers (Sec.~\ref{sec:topologicaloperatorclassification}) that seem very different to the symmetry quantum numbers that we usually have in Lagrangian field theory.\footnote{It is interesting to contrast with the `hedgehog-free' $\mathrm{O}(3)$ model \cite{kamal1993new, motrunich2004emergent,sreejith2019emergent}, which is  an $\mathrm{O}(3)$ nonlinear sigma model  in 2+1D with a constraint forbidding pointlike topological defects (related to Dirac monopoles in $\mathrm{U}(1)$ gauge theory \cite{murthy1990action}). That  constraint  reduces to the conservation of an integer skyrmion number associated with each spatial plane, and is equivalent to a global symmetry, namely conservation of a $\mathrm{U}(1)$ charge. The topological constraint in the loop model is of a different kind, in that it is not equivalent to conservation of an invariant associated with a spatial configuration.}

Fortunately there are also more concrete questions to pursue. 

First, questions specific to the loop models. Is $z$ exactly equal to 3, and if so, why is it an integer? What is the full spectrum of `hidden' operators? What is the nature of the operator product expansion? The structure of the low-lying excited states is also almost completely open. Can we write down a more complete set of variational states, giving for example an understanding of the low-lying density of states?

The Markov process used here to study the dynamics of the quantum loop models is also interesting 
independently of the quantum mapping,
as an alternative way of thinking about the 2D classical loop models,  whose spatial correlation functions have been characterized in great detail. 
A great deal is also known about  algebraic structures underlying classical 2D loop models (see e.g.
Ref.~\cite{gainutdinov2013logarithmic}): it would be interesting to understand how these relate to the dynamical models.

In the quantum context, the loop models with Jones-Wenzl projectors at $d=\sqrt{2}$ and $d = -\sqrt{2}$ suggest that pairs of gapless states can have the same excitation spectrum, operator spectrum and OPE, yet have subtle topological differences, a possibility which was recently discussed in the context of 1+1D CFT \cite{ji2019topological}.

Second, it would be interesting to study other models with a topological operator classification (some examples are given in  Sec.~\ref{sec:otherensembles}).
What is the simplest such model for loops in 2D? What is the most stable in the RG sense?
Are there examples that can be shown to be critical but which are not frustration-free?
Are there interesting gapless states adjacent to deconfined phases of gauge theories in 3+1D, where the flux lines have topologically constrained dynamics?
(In a separate paper we will discuss a one-dimensional analogue of a topologically constrained model \cite{1dforthcoming}.)

Finally, there are RG questions to understand better.  
We have described some general features of frustration-free models which will be interesting to examine in simple examples, for example the quantum Lifshitz theory that describes the Rokhsar-Kivelson dimer model: we hope to return to this elsewhere.

We have argued, nonrigorously, that the dynamical exponent $z$ is superuniversal for lines of RG fixed points under mild assumptions. 
However there are examples of models 
(some of them including quenched disorder \cite{thomson2017quantum}) 
with critical lines along which $z$ varies.
It would be useful to understand why these models do not obey the assumptions in Sec.~\ref{sec:superuniversality}, and to be able to state these assumptions more precisely. 
In particular, Ref.~\cite{isakov2011dynamics} found an example of a quantum Rohksar-Kivelson-like model with a critical line along which $z$ was not constant, taking the value ${z=2}$ at one point on the critical line and ${z\simeq 2.17}$ at another point.
One speculative possibility is that in this example the RG flows are more complicated, such that moving along the critical line in the UV model does not correspond to moving continuously along a single RG fixed line in the IR.

\section{Acknowledgements}

We are grateful to Paul Fendley, Jesper Jacobsen, Kirill Shtengel, 
and Xiao-Gang Wen for useful discussions and correspondence. ZD acknowledges support by the Tushar Shah and Sara Zion Graduate Fellowship  and by DOE office of Basic Sciences grant number DE-FG02-03ER46076.
AN was supported by the Gordon and Betty Moore Foundation under the EPiQS initiative
(grant No. GBMF4303), 
by the EPSRC under Grant No. EP/N028678/1 and by a Royal Society University Research Fellowship.

\appendix

\section{More on the quantum-classical operator correspondence}\label{Appendix: O3 patches}

In Sec.~\ref{sec:ffscalingoperators}, we claimed that we can implement the change of the ground state
\bea
|\Psi\>\rightarrow e^{\lambda\mathcal{O}_\mathbf{r}}|\Psi\>,
\eea
by a local change of the frustration-free Hamiltonian
\bea
\mathcal{H}\rightarrow \mathcal{H} -\lambda\mathcal{O}_{3\mathbf{r}},
\eea
for any operator $\mathcal{O}_\mathbf{r}$ with a finite support, say $B_\mathbf{r}$.
The idea is to merge the original patches $D_{\mathbf{r}'}$ that overlap with $B_\mathbf{r}$ into a larger patch $D^{(2)}_\mathbf{r}$
\bea
D^{(2)}_\mathbf{r}\supset D_{\mathbf{r}'},\ \forall \mathbf{r}' \text{ s.t. } B_{\mathbf{r}'}\cap D_\mathbf{r}\neq\emptyset
\eea 
Then we can rewrite each local projector in $D_{\mathbf{r}'}$ (Eq.~\ref{Eq: local projector}) as (here $\sigma$ runs over all states in $D^{(2)}_\mathbf{r}\setminus D_{\mathbf {r}'}$)
\bea 
\mathcal{P}_{i, {\mathbf{r}'}} &=&
\ket{i}_{D_{\mathbf {r}'}} \bra{i}_{D_{\mathbf {r}'}}
\otimes
\mathbb{1}_{D^{(2)}_\mathbf{r}- D_{\mathbf {r}'}}
\otimes
\mathbb{1}_{\overline{D^{(2)}_\mathbf{r}}}\\
&=& \sum_{\s}
\ket{i\sigma} \bra{i\s}
\otimes
\mathbb{1}_{\overline{D^{(2)}_\mathbf{r}}}\\
&\equiv& \sum_{\s}\mathcal{P}_{i\sigma, {\mathbf r}}
\eea
and perform the invertible transformation
\bea
\mathcal{P}_{i\sigma, {\mathbf r}}\rightarrow
e^{-\lambda O_\mathbf{r}}\mathcal{P}_{i\sigma, {\mathbf r}}e^{-\lambda O_\mathbf{r}}.
\eea
Each new $\mathcal{P}_{i\sigma, {\mathbf r}}$ is still proportional to a projector, and at the same time annihilates the new ground state.
\bea 
\mathcal{P}_{i\sigma,\mathbf{r}}e^{\lambda \mathcal{O}}|\Psi\> = 0,
\eea
where $|\Psi\>$ is a ground state of the original Hamiltonian. Thus the frustration-free condition is preserved.

\section{Analytical bound on dynamical exponent} \label{Appendix: analytical bound}

In this appendix we prove the analytical bound on the dynamical exponent stated in Sec.\ref{sec:zbound}, $z\ge 4-d_f$, for the loop model without reconnection. Our approach extends a previous calculation showing $z\ge 2$ \cite{freedman2005line,freedman2008lieb}. 

Recall that the low-energy excitations we found correspond to the motion of large loops, and the variational ansatz for them is
\bea
|n\> \equiv \frac{1}{\sqrt{Z}} \sum_C d^{|C|}e^{2\pi inp(A_C)}|C\>,
\label{Eq: appendix, variational states}
\eea
where $n$ is an integer labeling a series of tentative excited states ($n=0$ gives the ground state), $A_C$ is the area of the largest loop in configuration $C$ and $p(A_C)$ is the cumulative probability distribution of $A_C$. We choose the total space to be a sphere, where the ground state is unique.\footnote{We define the area of a loop to be the area of the smaller part of the sphere bounded by the loop.}

We assume that when $L$ is large enough, $p(A_C)$ becomes a smooth scaling function of $A_C/L^2$ determined by the IR fixed point. More precisely, we assume
\bea
p(A+1) - p(A) \le c/L^2,\nonumber\\ c\sim O(1), \forall\  0\le A\le L^2.
\eea

The variational states constructed in Eq. \ref{Eq: appendix, variational states} are orthonormal in the limit $L\rightarrow\infty$:
\bea
\< n' | n\> &=& \frac1Z \sum_C |d|^{2|C|}e^{2\pi i (n-n') p(A_C)}\nonumber\\
&=& \int_0^1 e^{2\pi i (n-n') p(A)} dp(A)\nonumber\\
&=& \delta_{n,n'}.
\label{Eq: Appendix, orthonormal variational states}
\eea
From the first line to the second line, we changed the summation to an integral and used the definition of $p(A)$: $dp(A)$ is the probability of configurations whose largest loop has area in the range $dA$. As a consequence of Eq.~\ref{Eq: Appendix, orthonormal variational states}, variational states with $n>0$ are orthogonal to the ground state; their energy expectation is greater than or equal to the first excited state energy. 

On the other hand, the energy expectations of the variational states are
\ba
E_n & = \< n | H | n\> - \< 0 | H | 0\>\\
& = \frac1Z \sum_{C,C',p}d^{|C|}\bar{d}^{|C'|} [e^{2\pi i n (p(A_C)-p(A_C'))} -1]\<C'|H_p|C\>\notag
\end{align}
where $p$ runs over all plaquettes on the lattice, $H_p$ is the local term acting on plaquette $p$. We only need to keep the real part of the summation. Since $H_p$ is local, $A_C$ and $A_{C'}$ differ by no more than 1; according to our assumption, $p(A_C)-p(A_C')\le c/L^2$, hence $\text{Re}[1-e^{2\pi i n (p(A_C)-p(A_C')}] \le c'n^2/L^4$. The matrix element that moves a loop across plaquette $p$ is either 0 or 1. For a given configuration $C$, only those $H_p$ next to the largest loop can change its area; therefore, the summation over $C'$ and $p$ gives at most a factor $l_C$, the length of the largest loop in $C$. Putting this together, we have
\bea
E_n \le \frac{c'n^2}{L^4} \<l_C\> = O(L^{d_f - 4})
\label{Eq: appendix, bound on En}
\eea
\bea
\Longrightarrow z\ge 4-d_f,
\label{Eq: appendix, bound on z}
\eea
where $d_f = 1 + \frac{\pi}{2\text{arccos}(-|d|^2/2)}$ is the \textit{fractal dimension} of loops.

There is a gap between Eq. \ref{Eq: appendix, bound on En} and Eq. \ref{Eq: appendix, bound on z}: in order to bound the dynamical exponent, we must confirm that the variational states belong to the continuous excitation spectrum, not a branch of states below this spectrum (like the sub-spectrum tower on the torus). To fill the gap, note that we can perform the construction independently in any subsystem --- instead of tuning the phase of the largest loop, we can divide the system into several subsystems and tune the phase separately for the largest loops in these subsystems; the number of excited states we can construct is exponential in the volume of the system, given a fixed energy density. This observation concludes our proof.

Intuitively \cite{freedman2008lieb}, we have approximately reduced the dynamics of loops to the dynamics of a fictitious particle hopping on the abstract axis $A_C$; the effective hopping strength is $L^{d_f}$, effective system size $L^2$, hence the excitation energy $L^{d_f-4}$. 

From this point of view, reconnection corresponds to nonlocal hopping of $A_C$; it completely changes the dynamics of loops and reduces the dynamical exponent $z$ (probably to zero, see Sec.~\ref{Appendix: numerical methods} \cite{troyer2008local}).

\section{Quantum-Markovian correspondence} \label{Appendix: Quantum Markovian correspondence}

In this appendix, we review a previously-found correspondence between a large class of frustration-free Hamiltonians and classical Markovian dynamics~\cite{henley1997relaxation,henley2004classical,castelnovo2005quantum,velenich2010string} in the most general setting we found, in the hope of stimulating separate applications.
We also clarify the distinction between a frustration-free Hamiltonian and a Hamiltonian with a classical correspondence,
and we point out a simplification for computing correlators of off-diagonal operators in models with a classical correspondence. This appendix is self-contained: readers interested only in this correspondence may skip the main text.

The quantum-Markovian correspondence applies to the gapless loop model without reconnection as well as the toric code/double semion model. It helps us simulate the gapless loop model on a large lattice of $500\times500$ plaquettes. 

A classical Markov process satisfying detailed balance is described by a master equation
\bea
\frac{dp_{\a}}{d\tau}=\sum_{\b\neq\a}W_{\a\b}p_{\b}(\tau) - W_{\b\a}p_{\a}(\tau),
\eea
where $\a, \b$ label classical states, $p_{\a}(\tau)$ is a probability distribution evolving with time $\tau$, and $W_{\a\b}$ is the transition amplitude from state $\b$ to state $\a$. This equation is often written as 
\bea
\frac{dp}{d\tau}=Wp
\eea
by defining ${W_{\a\a}\equiv -\sum_{\b\neq\a}W_{\b\a}}$. This matrix $W$ has three properties: (1) Positivity,  ${W_{\a\b}\ge0}$ for all $\a\neq\b$; (2) Classical probability conservation, 
${(1,1,\cdots,1)W=0}$; and
(3) Detailed balance,  
${W_{\a\b}p_{\b}^{0}=W_{\b\a}p_{\a}^{0}}$ (no sums on $\alpha$ or $\beta$),
where $p_{\a}^0$ is the probability of state $\a$ in equilibrium. We shall denote the equilibrium distribution $(p_1^{0},\cdots,p_N^0)$ by $\<\tilde{0}|$ and $(1,1,\cdots,1)$ by $\<\tilde{1}|$.

On the other hand, a  local frustration-free quantum Hamiltonian is given  as a sum over local terms by $H=\sum_x P_x$ satisfying two properties: (1) Hermiticity,
$H_{\a\b}=H_{\b\a}^*$ and
(2) Frustration-freeness, meaning that each $P_x$ is projector and these projectors have a common ground state. These projectors need not commute with each other.

In order to make the quantum-classical correspondence, we \textit{fix a basis} on the quantum side. Each basis vector $|\a\>$ corresponds to a classical configuration $\a$.
Refs.~\cite{henley1997relaxation,henley2004classical,castelnovo2005quantum} point out that for Hamiltonians decomposable into blocks of $2\times2$ projectors, a classical correspondence always exists, and detailed balance in the corresponding classical dynamics is guaranteed.
(A $2\times 2$ projector is a projector involving only 2 classical configurations on a local patch. For example the Hamiltonian in Eq.~\ref{Eq:2by2projector} is decomposable into $2\times 2$ projectors whereas the Jones-Wenzl projector is not.)
Ref.~\cite{velenich2010string} points out frustrations in constructing the classical mapping for non-abelian topological models. We find that the $2\times2$ block structure is not essential: the classical mapping exists as long as each local projector obeys a ground state uniqueness property.
Regarded as an operator on an appropriate local patch, the projector is a  matrix in the classical basis with a nonzero block for the configurations the projector acts on, and zero elements elsewhere: this  nonzero block must have a unique ground state (see the precise definition below). A generic $2\times2$ projector has one ground state and one excited state, and the condition is satisfied. For generic projectors in the form of Eq.~\ref{Eq: local projector}, we may need to group multiple projectors into one in order to satisfy the condition. This observation motivates the following definition.

\begin{definition}\label{defn:P}
A projector $P_x$ acting on a local patch $x$ is called a \textit{local classical projector} if and only if there exists a set $\mathcal{A}_{x}$ of classical configurations on $x$, and a normalized state ${|\psi_{x}\>\equiv\sum_{\a\in\mathcal{A}_{x}}\psi_{x,\a}|\a\>}$ on the patch, such that 
\bea
\label{eq: P_x}
P_{x} &=& \left(\sum_{\a\in\mathcal{A}_{x}}|\a\>\<\a|\right) - |\psi_{x}\>\<\psi_{x}|\\ 
&\equiv&  P_{\mathcal{A}_x} - |\psi_{x}\>\<\psi_{x}|.
\eea
\end{definition}

In other words, $P_{x}$ has a unique ground state within the Hilbert space spanned by the configurations it acts nontrivially on. There is a large class of Hamiltonians with the  local classical projector structure, for example, those string-net Hamiltonians realizing abelian topological order \cite{levin2005string} and the dimer model at the Rokhsar-Kivelson point \cite{rokhsar1988superconductivity}. From Eqs. \ref{eq:flippable} and \ref{Eq: H gapless loop}, the ideal Hamiltonian for the loop model without reconnection is also a sum of local classical projectors.\footnote{The frustration-free Hamiltonian may also include projectors that just forbid a single local classical configuration, for example one that forbids dangling string ends in the loop model. In this case, we simply remove these classical configurations from the Hilbert space.}
We are now ready to state the main theorem of this appendix.

\begin{theorem}
\label{Thm: quantum-classical mapping}
Let $H$ be a frustration-free Hamiltonian that is a sum of local classical projectors. Let $\ket{\mathrm{GS}}$ be a ground state of $H$, and  let $\mathcal{A}$ be the set of many-body configurations in the classical basis where the ground state wavefunction amplitude ${\langle \alpha | \mathrm{GS}\rangle}$ is nonzero. Then
\begin{enumerate}
\item $\text{span}(\mathcal{A})$ is an invariant subspace of $H$: $H\text{span}(\mathcal{A})\subseteq \text{span}(\mathcal{A})$.
\item $H|_{\text{span}(\mathcal{A})}$ maps onto the transition matrix of a local classical Markov process, satisfying detailed balance, under a similarity transformation. 
\item The quantum ground state maps onto the classical equilibrium distribution.
\end{enumerate}
\end{theorem}

\begin{remark}
When $H$ has multiple ground states, Thm.~\ref{Thm: quantum-classical mapping} applies to each of them.
\end{remark}

\begin{proof}
We prove the first property by contradiction. Assume $\text{span}(\mathcal{A})$ is not an invariant subspace of $H$; then, there exists a local classical projector $P_x$ having a nonzero matrix element between a configuration $\a_x\g_{\text{rest}}\in\mathcal{A}$ and a configuration $\b_x\g_{\text{rest}}\notin\mathcal{A}$, where $\a$ and $\b$ label configurations on the local patch $x$, and $\g_\text{rest}$ labels configurations on the rest of the system excluding $x$. Since $P_{x}$ has a nonzero matrix element between these two states, $\a$ and $\b$ must be in $\mathcal{A}_x$ (defined in Defn.~\ref{defn:P}). By definition, $P_{x}$ has a unique ground state $|\psi_{x}\>$ built from configurations in $\mathcal{A}_x$. The many-body ground state is also a ground state of $P_x$, which means that the ground state amplitudes involving states in $\mathcal{A}_x$ must be proportional to $\psi_x$:
\bea
\frac{\psi_{\b_x\g_{\text{rest}}}}{\psi_{\a_x\g_{\text{rest}}}}=\frac{\psi_{x,\b}}{\psi_{x,\a}}.
\eea
However, ${\psi_{x,\b}\neq 0,\psi_{x,\a}\neq 0}$ by the definition in Eq.~\ref{eq: P_x}, contradicting the assumption that ${\b_x\g_{\text{rest}}\notin\mathcal{A}}$, i.e. that $\psi_{\b_x\g_{\text{rest}}}=0$.

The second and the third properties are closely related. In order to map the quantum ground state $|0\>=(\psi_1,\psi_2,\cdots)^{T}$ to the classical equilibrium distribution $|\tilde{0}\>=(|\psi_1|^2,|\psi_2|^2,\cdots)^{T}$, we use a diagonal matrix
\ba
S&=\left (\begin{array}{cccc}
\psi_1^* & & \\
 & \psi_2^* & \\
 & & \ddots
\end{array}\right),
&
|\tilde{0}\>&=S|0\>.
\end{align}
From now on, we work in the subspace $\text{span}(\mathcal{A})$ --- we assume each $\psi_{\a}$ is nonzero. 

Remarkably, under the corresponding similarity transformation, the quantum Hamiltonian maps onto (minus) a local transition matrix $\tilde{H}$: 
\bea
\tilde{H}=-SH|_{\text{span}(\mathcal{A})}S^{-1}.
\eea
In order to check the three defining properties of the transition matrix --- positivity, probability conservation and detailed balance --- for $\tilde{H}$, we write Eq.~\ref{eq: P_x} explicitly in the classical basis. For $\a,\b\in\mathcal{A}_x$ and $\a_x\g_{\text{rest}},\b_x\g_{\text{rest}}\in\mathcal{A}$, define $P_{\a,\b}\equiv\<\a|P_{x}|\b\>$, then Eq. \ref{eq: P_x} reads $P_{\a,\b}=\delta_{\a\b}-\psi_{x,\a}\psi_{x,\b}^*$; then, 

\ba
(\tilde{P}_x)_{\a_x\g_{\text{rest}},\b_x\g_{\text{rest}}}& \equiv (-SP_{x}\otimes I_{\text{rest}}S^{-1})_{\a_x\g_{\text{rest}},\b_x\g_{\text{rest}}}  \notag \\
&= -\psi_{\a_x\g_{\text{rest}}}^*(\delta_{\a\b}-\psi_{x,\a}\psi_{x,\b}^*)/\psi_{\b_x\g_{\text{rest}}}^*  \notag \\
& = -\psi_{x,\a}^*(\delta_{\a\b}-\psi_{x,\a}\psi_{x,\b}^*)/\psi_{x,\b}^* \notag  \\
\notag
&  = |\psi_{x,\a}|^2 - \d_{\a\b} \\
& \equiv \tilde{P}_{\a,\b}.
\end{align}
Indeed, the off-diagonal matrix elements of $\tilde{P}$ are nonnegative and 
\ba
\tilde{P}_{\a,\b}|\psi_{\b_x\g_{\text{rest}}}|^2 & = |\psi_{x,\a}|^2|\psi_{\b_x\g_{\text{rest}}}|^2\nonumber\\ 
 & = |\psi_{x,\b}|^2|\psi_{\a_x\g_{\text{rest}}}|^2 \\
 & = \tilde{P}_{\b,\a}|\psi_{\a_x\g_{\text{rest}}}|^2,\nonumber
\end{align}
for all $\a_x\g_{\text{rest}}$, $\b_x\g_{\text{rest}}\in\mathcal{A}$, $\a\neq\b$, which is detailed balance, and
\bea
\sum_{\a}\tilde{P}_{\a,\b}&=& \sum_{\a} |\psi_{x,\a}|^2 - \d_{\a\b}\nonumber\\
&=& |\<\psi_x|\psi_x\>|^2-1 = 0
\eea
for all $\beta$, which is probability conservation. Moreover, the matrix elements of $\tilde{P}$ depend only on the local configuration $\a$, not on the full global configuration $\a\g_{\text{rest}}$;
therefore, the transition matrix is local.  As a consequence, $\tilde{H}=\sum_{x}\tilde{P}_x$ is a legitimate local transition matrix.
\end{proof}

\begin{remark}
What lies at the center of this correspondence is a sense of \textit{locality} in the many-body ground state: for any local patch $x$ and any pair of configurations $\a$ and $\b$ on $x$ connected by a local projector, the ratio of ground-state amplitudes $\psi_{\a_x\g_{\text{rest}}}/\psi_{\b_x\g_{\text{rest}}}$, if $\psi_{\b_x\g_{\text{rest}}}\neq 0$, depends only on the local data $\a$ and $\b$, not on the configuration outside $x$. However, this locality of the ratio does not imply the locality of the wavefunction itself --- the ground-state wavefunction is not necessarily a product of local functions. This is best seen in the loop model with nontrivial weight $\psi(C)=d^{|C|}$, as the total number of loops is not a sum of locally measurable quantities.
Even for the gapped system with $d=-1$ (with reconnection in the Hamiltonian), the wavefunction is nonlocal in this sense.
\end{remark}

This surprising quantum-classical correspondence allows us to simulate the quantum dynamics at a low cost. Intuitively, the classical Markov process can be viewed as a Monte Carlo sampling of the square of quantum wavefunction. 
When the ground state wavefunction is known exactly, this Monte Carlo process is not surprising. However, what is nontrivial is that the time dimension in the Monte Carlo process exactly corresponds to imaginary time in the quantum system, allowing the direct measurement of imaginary time correlators in the classical simulation. This is best seen in the following corollary (previously discussed in Ref.~\cite{henley1997relaxation,henley2004classical,castelnovo2005quantum}).

\begin{corollary}
\label{Cor: diagonal temporal correlation function}
For diagonal operators $O_{1}, O_{2}, \cdots, O_{n}$ on $\text{span}(\mathcal{A})$\footnote{By diagonal operators, we mean diagonal operators in the classical basis. This basis plays a special role in the quantum-classical correspondence, and we stick to it in this paper.}, 
the imaginary-time quantum correlator equals the real-time classical correlator:
\bea
\label{eq: diagonal temporal correlation}
&\<0| & O_{n}e^{-H\Delta\tau_{n-1}}\cdots e^{-H\Delta\tau_{2}}O_{2}e^{-H\Delta\tau_{1}}O_{1}|0\>\nonumber\\
&=& \<\tilde{1}|O_{n}e^{\tilde{H}\Delta\tau_{n-1}}\cdots e^{\tilde{H}\Delta\tau_{2}}O_{2}e^{\tilde{H}\Delta\tau_{1}}O_{1}|\tilde{0}\>\\
&=& \sum_{\{\a_{i}\}}\prod_{1}^{n}O_{i}(\a_{i})\prod_{1}^{n-1}p(\a_{i+1},\tau_{i+1}|\a_{i},\tau_{i})p^0(\a_{1}),\ \ 
\eea
where $\tau_i=\sum_{j<i}\Delta\tau_{i}$, $O_i(\alpha)$ is the value of the classical observable in state $\alpha$, and ${p(\a_{i+1},\tau_{i+1}|\a_{i},\tau_{i})\equiv\<\a_{i+1}|e^{\tilde{H}\Delta\tau_{i}}|\a_i\>}$ is the conditional probability that the system is in state $\a_{i+1}$ at time $\tau_{i+1}$ given that it is in state $\a_i$ at time $\tau_i$.
\end{corollary}
\begin{proof}
This corollary follows from inserting $S^{-1}S$ between neighboring operators and between the first (last) operator and the initial (final) state on the LHS of Eq. \ref{eq: diagonal temporal correlation}, performing the similarity transformation, and noting that $\<0|S^{-1}=\<\tilde{1}|$, $SO_{i}S^{-1}=O_{i}$ since $O_{i}$ is diagonal, $\tilde{H}=-SH|_{\text{span}(\mathcal{A})}S^{-1}$, and $S|0\>=|\tilde{0}\>$.
\end{proof}

As a special case of the dynamical correspondence, we have the following static correspondence.
\begin{corollary}
\label{Cor: scaling dimension}
If the frustration-free Hamiltonian is gapless, the scaling dimension of (non-hidden) local diagonal operators equals the scaling dimension of the same operator in the statistical mechanics system with probability distribution $|\tilde{0}\>$.
\end{corollary}

At first glance one might think that  Corollary \ref{Cor: diagonal temporal correlation function} and \ref{Cor: scaling dimension} are special for diagonal operators --- after all, off-diagonal operators mean nothing in the classical equilibrium ensemble. However, off-diagonal quantum operators also have classical counterparts in the quantum-Markovian correspondence. 
Remarkably, if we restrict to two point functions, the off-diagonal operators may be replaced with diagonal operators on the classical side, though local operators may turn into nonlocal ones.
(This correspondence is different from the correspondence between reconnection operators and defect operators in the Coulomb gas formalism we discussed in Sec.~\ref{Sec: operators and correlation functions}, which works only for equal-time correlation functions and is special to the loop models.)

\begin{corollary}
\label{Cor: off-diagonal operator, correspondence}
Every local Hermitian operator $O_{q}$ on $\text{span}(\mathcal{A})$ corresponds to a diagonal classical operator $O_{c}$, such that the 2-point function of operators $O_{1,q}$ and $O_{2,q}$ equals the 2-point function of the diagonal classical operators $O^*_{1,c}$ and $O_{2,c}$:
\ba
\<0|O_{1,q}e^{-H\tau}O_{2,q}|0\> = \<\tilde{1}|O_{1,c}^*e^{\tilde{H}\tau}O_{2,c}|\tilde{0}\>\nonumber & \\ 
= \sum_{\a,\a'}O_{1,c}^*[\a']O_{2,c}[\a]p(\a',\tau|\a,0)p^0(\a) &
\end{align}
\end{corollary}

\begin{proof}
Every local Hermitian operator on a local patch $x$ can be expanded as
\begin{align}
\left(\sum_{\a}r_{\a}|\a\>\<\a| + \sum_{\a\b}c_{\a\b}|\a\>\<\b| + c_{\a\b}^*|\b\>\<\a|\right)_{x}\otimes I_{\text{rest}},
\end{align}
($r_{\a}\in\mathbb{R}$,  $c_{\a\b}\in\mathbb{C}$) where $\a, \b$ run over all configurations on $x$. Corollary~\ref{Cor: off-diagonal operator, correspondence} reduces to Corollary~\ref{Cor: diagonal temporal correlation function} for local diagonal operators, so we may focus on the off-diagonal elements. For any given $\a$, $\b$, we may write operators on $\text{span}(|\a\>, |\b\>)$ as $2\times 2$ matrices. Consider the local operator ${O_{q}=(c_{\a\b}|\a\>\<\b| + c_{\a\b}^*|\b\>\<\a|)_{x}\otimes I_{rest}}$; it maps onto:
\bea
&\tilde{O}&_{q} = SO_{q}S^{-1}\nonumber\\
&=& \sum_{\g}
\left(\begin{array}{cc}
0 & c_{\a\b}\psi_{\a\g}^*/\psi_{\b\g}^*\\ c^*_{\a\b}\psi_{\b\g}^*/\psi_{\a\g}^* & 0
\end{array}\right)
\otimes|\g\>\<\g|,\ \ 
\eea
where $\g$ runs over all configurations outside $x$. Note that $\a$ and $\b$ may not be connected by any local projector; the ratio $\psi_{\b\g}^*/\psi_{\a\g}^*$, hence $\tilde{O}_{q}$, can be nonlocal.

Now define 
\bea
O_{c}\equiv \sum_{\g}
\left(\begin{array}{cc}
c_{\a\b}\psi_{\b\g}/\psi_{\a\g} & 0\\ 0 & c^*_{\a\b}\psi_{\a\g}/\psi_{\b\g}
\end{array}\right)
\otimes|\g\>\<\g|.\nonumber
\eea
We have
\onecolumngrid
\bea
(\tilde{O}_{q} - O_c)|\tilde{0}\> = \sum_{\g}
\left(\begin{array}{cc}
-c_{\a\b}\psi_{\b\g}/\psi_{\a\g} & c_{\a\b}\psi_{\a\g}^*/\psi_{\b\g}^*\\ c^*_{\a\b}\psi_{\b\g}^*/\psi_{\a\g}^* & -c^*_{\a\b}\psi_{\a\g}/\psi_{\b\g}
\end{array}\right)
\left(\begin{array}{c}
|\psi_{\a\g}|^{2}\\ |\psi_{\b\g}|^2
\end{array}\right)
\otimes|\g\> = 0
\eea

\bea
\<\tilde{1}|(\tilde{O}_{q} - O_c^*) = \sum_{\g}
(1,1)
\left(\begin{array}{cc}
-c_{\a\b}^*\psi_{\b\g}^*/\psi_{\a\g}^* & c_{\a\b}\psi_{\a\g}^*/\psi_{\b\g}^*\\ c^*_{\a\b}\psi_{\b\g}^*/\psi_{\a\g}^* & -c_{\a\b}\psi_{\a\g}^*/\psi_{\b\g}^*
\end{array}\right)
\otimes\<\g| = 0
\eea
Thus,
\bea
\<0|O_{1,q}e^{-H\tau}O_{2,q}|0\> = \<\tilde{1}|\tilde{O}_{1,q}e^{\tilde{H}\tau}\tilde{O}_{2,q}|\tilde{0}\> 
= \<\tilde{1}|O_{1,c}^*e^{\tilde{H}\tau}O_{2,c}|\tilde{0}\>.
\eea
\twocolumngrid
By linearity of the mapping from $O_q$ to $O_c$, the equation above holds for arbitrary local Hermitian operators. 
The key feature of the present class of Hamiltonians which makes the above mapping useful is that the the ratios $\psi_{\beta\gamma}/\psi_{\alpha\gamma}$ are known analytically, so the operator $O_c$ can be explicitly constructed.
\end{proof}

Corollaries \ref{Cor: diagonal temporal correlation function}, \ref{Cor: scaling dimension}, and \ref{Cor: off-diagonal operator, correspondence} play a key role in our numerical simulation. We use them to measure the dynamical exponent and to confirm our analytical results on the scaling dimensions of quantum operators. The mapping in Corollary \ref{Cor: off-diagonal operator, correspondence} may not always be useful analytically, as the corresponding classical operator can be highly non-local, but it is very convenient for Monte Carlo simulations, since
we can replace off-diagonal operators with diagonal ones when we compute temporal 2-point functions.

We have only discussed operators acting on $\text{span}(\mathcal{A})$. One may wonder what happens to operators that map $\text{span}(\mathcal{A})$ into $\text{span}(\mathcal{A})^{\perp}$. In the loop model we study, $\text{span}(\mathcal{A})$ is the subspace spanned by all closed-loop configurations, and operators that disrespect $\text{span}(\mathcal{A})$ create endpoints of strings. These operators all have exponentially decaying correlation functions, because states with end points have energy strictly above 2.

\section{Numerical methods} \label{Appendix: numerical methods}

In this appendix, we introduce the numerical method we use in simulating the loop model without reconnection. We use the quantum-Markovian correspondence discussed in the previous appendix to calculate (1) dynamical 2-point functions for diagonal and off-diagonal operators, (2) lifetime-area-length distribution of loops during the time evolution. The first type of measurement helps us understand the scaling structure of the quantum loop model; the second provides an intuitive understanding of the low-energy dynamics, and an independent confirmation of the dynamical exponent. 

The unusual feature of this simulation is that off-diagonal operators map to nonlocal classical observables whose values depend on whether two points are on the same loop. To calculate their correlators, we need to label every loop in a classical configuration, and keep updating not only the spins but also the loop labels during the Monte Carlo process. 
These labels also help us to measure the lifetime-area-length distribution of loops. We shall talk about how we implement the Monte Carlo update, label loops, and generate initial configurations. The numerical results are shown in figures in previous sections and discussions therein.

The quantum-Markovian correspondence reduces the numerical complexity of simulating the quantum dynamics from exponential in the system size to power-law in the system size. The reason is that we do not need to keep the exponentially large number of wavefunction amplitudes. Instead, we only keep one classical state during the update and we engineer the Markov process such that the probability each state appears is proportional to the square of the time-dependent wavefunction amplitude.

A particular realization of the classical Markov process is a series of spin-changing events. O($N$) such local events happen in the system in an O(1) time interval, where $N$ is the number of plaquettes. In our discrete-time simulation, we fix the number of update attempts in each unit time interval to be exactly $N$. This change is unimportant for large systems, since relative fluctuations in the number of updates per unit time are anyway small. Below, we describe the update process for $|d|^2>1$; the process for $|d|^2<1$ is very similar.  

For each update attempt, we choose a random plaquette on the honeycomb lattice, and update the spin configuration on the links surrounding this plaquette according to the following rule: (1) If there is one loop passing through the plaquette, we flip all spins around the plaquette (move the loop across the plaquette). (2) If there is no loop passing through the plaquette, we flip all spins around the plaquette (create a small loop). (3) If there is a small loop on the plaquette, we flip all spins around the plaquette (annihilate the small loop) with probability $1/|d|^2$. (4) If there are multiple loops passing through the plaquette, we keep the current spin configuration (no reconnection). With this normalization of the rates, the Markov process is related by the Markovian correspondence to the quantum Hamiltonian in  
Eqs.~\ref{eq:flippable},~\ref{Eq: H gapless loop}, with the modified ratio of couplings $K_1/K_2= |d|^2/(1+|d|^2)$ mentioned in that section (we do not expect this choice to affect the universal properties).

In addition, we label the loops in the initial configuration, and update the labels during the dynamics, in the following way. For the initial configuration, we label the loops from the top left to the bottom right by iteration. In the $n$th step, to label the $n$th loop, we find the first unlabeled down spin in the configuration, give it the label $n$, assign an arbitrary direction to the loop, and follow the loop in this direction in order to attach the label $n$ to all the down spins along the loop. At the end of the process, each edge occupied by a down spin has an additional integer label, $n$, specifying which loop it belows to, and a direction, which will be used in calculating the area. During the update, whenever we create a loop we give its edges a new integer label different from all existing loops and an arbitrary direction; when we annihilate a loop we delete the labels around the plaquette. When we move a loop, we label the new edges by the same integer on the existing loop, assign the direction consistent with the direction of the loop, and delete labels on the edges no longer on the loop. This update process is local and takes only $O(1)$ computational time. (If we have reconnection, we have to relabel all spins along the loop. Even through the reconnection is local, the relabeling takes $O(l)$ time, where $l$ is the length of the loop.)

When calculating the lifetime-area-length distribution, we first compute the lengths and areas of loops in the initial condition, and then update the lengths and areas during the Monte Carlo process. We calculate the area by the discrete version of the integral $\int ydx$ along the loop: the result can be either positive or negative, with the sign determined by the orientation assigned to the loop. With this definition, computing the change of area in an update to the loop is a local process which takes $O(1)$ computational time. When a loop is annihilated, we take the current time as its lifetime; we only  consider loops present in the initial state. At the end, we make the scatter plot with the lifetime of each loop and the maximum area (or average area, initial area, maximum length, etc.) of the loop during the Monte Carlo process (see Fig.~\ref{Fig: area length life}(b)).

When calculating the correlation functions of reconnection operators, for example $W_{2,0}$ in Eq.~\ref{Eq:W20definition}, we first map the operator to a nonlocal classical operator $\tilde{O}$ via the general rule in Corollary~\ref{Cor: off-diagonal operator, correspondence} above, record the value of this classical operator at each spacetime point we are interested in, and then calculate the correlator in the usual way. 
Up to a normalization constant, the reconnection operator $W_{2,0}$ (defined for simplicity on a single hexagonal plaquette)
is mapped to the diagonal operator on a plaquette taking the values: 
$|d|^2$ if there are two strands visiting the plaquette, which belong to the same loop; 
$-1$ if the two strands belong to different loops;
and 0 if there are less than or more than two strands visiting the plaquette. 
We check whether they belong to the same loop by comparing their labels, which takes only O(1) time.

At $d=1$, the 2-point function of $W_{2,0}$ vanishes for any spacetime separation, and its classical counterpart $\tilde{O}$ vanishes identically. However, $\tilde{O}/(d-1)$ is well-defined in the limit $d\rightarrow 1$. We calculate the 2-point function of this operator to see whether the critical exponent of $W_{2,0}$ behaves as we predicted for $d$ close to 1.

\begin{figure}[t]
\begin{center}
\includegraphics[width=3in]{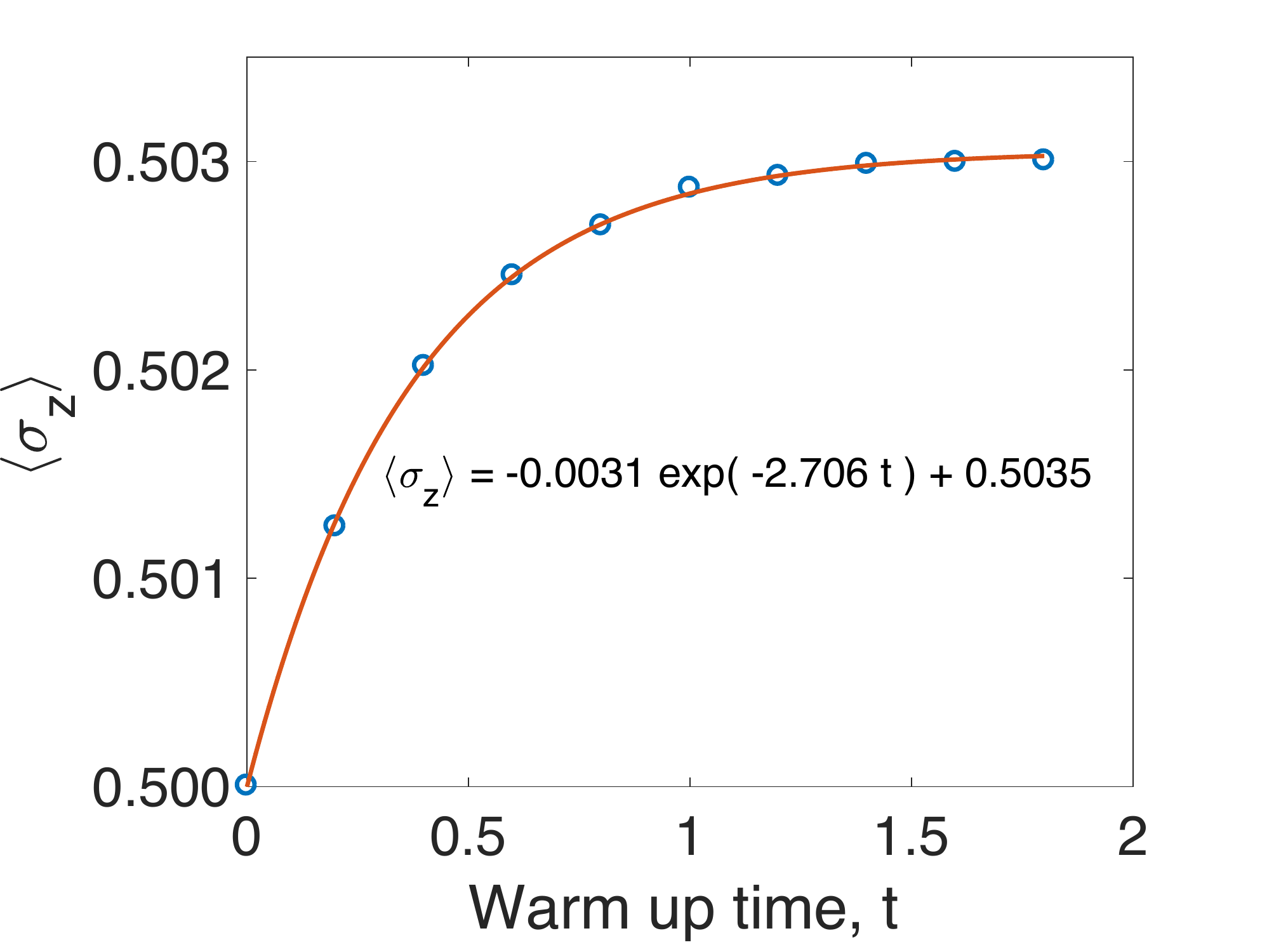}
\caption{Equilibration of the spatially averaged $\s^z$ during the warm-up (${d=\sqrt{2}}$). Warm-up time 1 means $L^2$ random plaquette update attempts (these can include loop reconnections, as described in the text). Results here are for system size $L = 500$.}
\label{Fig: Warm up, magnetization}
\end{center}
\end{figure}

Last, we describe the way we generate the initial loop configurations for arbitrary $|d|^2$. The simplest case is $|d|^2=1$. We first generate a random $\tau^z$ on each plaquette  (spin up or spin down with equal probability) and then calculate $\sigma^z$ as the product of the two neighboring $\tau^z$. This process produces the uniform probability distribution of all closed-loop configurations. We may use either open boundary conditions, periodic boundary conditions or anti-periodic boundary conditions for $\tau^z$. For $|d|^2\neq 1$, we start with the uniform distribution and run a Monte Carlo process to get the desired probability distribution: $p(C)\propto |d|^{2|C|}$. One choice is to use the same Monte Carlo process corresponding to the quantum Hamiltonian. However, its dynamical exponent is $z=3.00(6)$, so the warm up time\footnote{{I.e. the time $\tau$ of the physical Markov process; the computational time has an extra factor of $L^2$ since one unit of physical time requires $N$ updates.}}
is about $O(L^3)$. Instead, we use a nonlocal warm up process, which includes reconnection of loops, to greatly reduce the warm up time. The subtlety is that the nontrivial loop weight $|d|^2$ makes the acceptance probability for a reconnection event nonlocal, since it 
depends on whether the update increases or reduces the number of loops. Again, we determine this by comparing the labels of the strands involved in the reconnection. After each reconnection event, we must relabel loops. But relabeling takes at most $O(L^{d_f})$ computational time (and much less on average), which is much smaller than $L^3$. 
Numerical results in Fig.~\ref{Fig: Warm up, magnetization}
for $d=\sqrt{2}$ show that after including this nonlocal reconnection process, the `physical' warm up time on the torus with $500\times 500$ plaquettes is smaller than 1 (which corresponds to of order $500\times 500$ updates).

As an aside, it is interesting to ask about the scaling of this warm-up time with $L$, i.e. the dynamical exponent of the \textit{nonlocal} dynamics that includes reconnection.
At $|d|^2=1$ this is clearly $z_\text{reconnection}=0$, because for that value of the loop weight the update becomes equivalent to flipping uncorrelated $\tau^z$ spins, which relax after an order 1 time (again, 1 unit of time corresponds to $O(L^2)$ updates). 
We have not performed a scaling study, but the small warm-up times for other $d$ suggest that $z_\text{reconnection}$ may be equal to zero for all $d$ which means the dynamics is invariant under rescaling only the spatial coordinates.
Surprisingly, the nonlocal quantum Hamiltonian that maps to this kind of nonlocal dynamics has been studied numerically in Ref.~\cite{troyer2008local}:  a gap was reported, which implies a relaxation time of order 1 in the Markov process, consistent with the above.
This indicates that the dynamical exponent of the nonlocal Markov process shares the property of being independent of $d$ that we have argued for in the local case.

\section{Mathematical structure of equal-time correlators and topological types}\label{Appendix:topologicaloperatorclassification}

In this Appendix, we explore the rich mathematical structure behind 2-point functions of local operators, and prove the general result on the number of watermelon operators.

Recall that Eqs.~\ref{Eq:2ptfunctionfinalexpression}-\ref{Eq:OAOBgammadefinition} reduce every equal-time 2-point function to a sum of classical probabilities. The key to a general understanding of 2-point functions, and hence the spectrum of (non-hidden) operators, is the topological part of the correlator $\<O_AO_B\>_{\g}$ defined in Eq.~\ref{Eq:OAOBgammadefinition}. 
It is a function of the topological data of the two operators and of the connectivity $\gamma$ of end points by strands in the region exterior to the discs $A$ and $B$ that the operators act in. $\<O_AO_B\>_{\g}$ determines which combination of classical probabilities shows up in a given 2-point function.
We can easily compute this quantity for any specific $O_A$, $O_B$ and $\g$. 
Naively, in order to gain a full understanding, we have a formidable task to calculate this coefficient  for every combination of $O_A$, $O_B$ and $\g$, case-by-case. However, these coefficients can all be reduced to a kind of overlap matrix $M^k$.

The matrix $M^k$, with elements $M^k_{\a, \b}$, is defined for pairings (connections) $\a$, $\b$ with $2k$ end points. The index $\a$ labels the connection inside the disc, and $\b$ labels the connection outside the disc:
\bea 
M^k_{\a,\b}\equiv d^{[\a,\b]},
\eea
where $[\a,\b]$ is the total number of loops formed by connecting the strands in $\a$ and $\b$. In the context of the Temperley-Lieb algebra, this is the inner product between states, so is well studied \cite{freedman2004class}.

\begin{figure}[b]
\begin{center}
\includegraphics[width=0.45\textwidth]{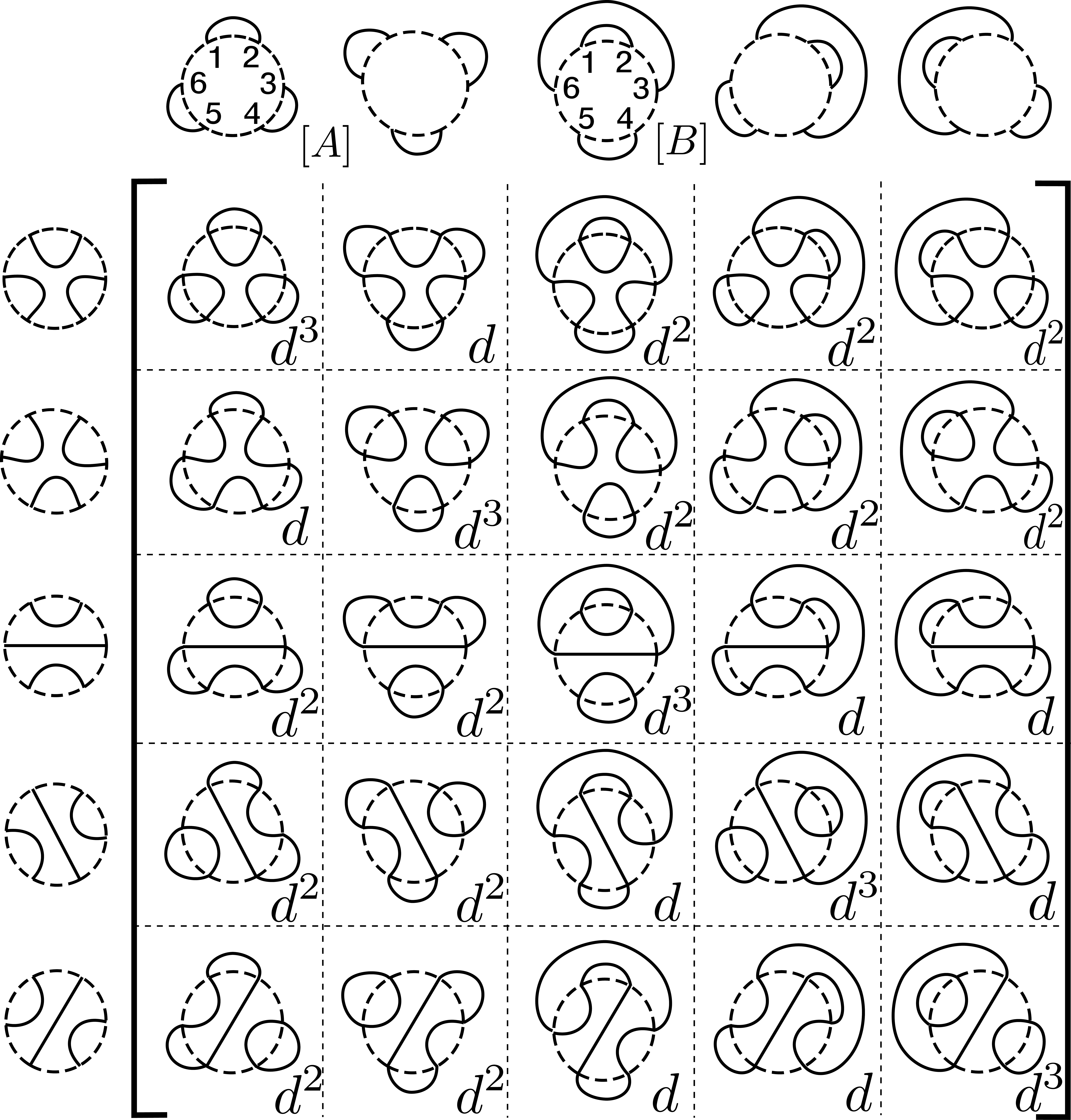}
\caption{$M^{3}_{\a,\a'}$}
\label{Fig: M matrix}
\end{center}
\end{figure}
For example, Fig.~\ref{Fig: M matrix} shows the matrix $M^3$. There are five different choices of connection inside and outside the disc. We can calculate the determinant of $M^3$ directly and see that it is generically invertible, except for $d=\pm1,\pm\sqrt{2}$.
More generally, the determinant of $M^k$ vanishes if and only if $d=\cos(p\pi/(q+1)),q\le k$; see e.g. Ref~\cite{freedman2004class} which discusses the Temperley Lieb algebra in the context of the quantum loop models. We shall discuss these special values $d=\pm\sqrt{2}$ at the end of this Appendix. Now we focus on generic $d$, when all $M^k$ are invertible.

To derive the relation between $\<O_AO_B\>_\g$ and $M^k$, we manipulate $\<O_AO_B\>_\g$ in the following way. We treat $\g$ as a map that maps connections $\a_B$ inside disc $B$ into connections outside disc $A$. This mapping can be understood as simply removing the boundary of disc $B$ and keeping only the strands that connect to end points on the boundary of disc $A$ (for an example, see Fig.~\ref{Fig:alphaGammaConnection}). We denote the action of this map as $\g\circ\a_B$. This map naturally extends to a map from operators on disc $B$ to operators acting on the linear space of connections outside disc $A$, which we label by $\b_A$:
\bea
\g(O_B)_{\b'_A,\b_A}\equiv \sum O_{\a'_B,\a_B}\bar{d}^{[\g\a'_B]}d^{[\g\a_B]},
\eea 
where $[\g\a_B]$ is the number of complete loops formed by strands in $\g$ and strands in $\a_B$ (these loops are completely outside disc $A$), and the summation runs over the preimage of $\b_A$ (or $\b'_A$):
namely the set of all $\a_B$ (respectively $\a'_B$) s.t. $\g\circ\a_B=\b_A$  (respectively $\g\circ\a'_B=\b'_A$).
On the sphere, connections inside and outside the discs are in 1-to-1 correspondence, so we can also treat the resulting operator as an operator inside disc $A$.

By definition $[\a_A\g\a_B] = [\a_A,\g\circ\a_B] + [\g\a_B]$. Thus we can rewrite $\<O_AO_B\>\g$ as
\bea
\label{Eq: fundamental equation OAOBgamma, start}\<O_A&O&_B\>\g = \sum \bar{d}^{[\a'_A\g\a'_B]}d^{[\a_A\g\a_B]} O_{\a'_A,\a_A} O_{\a'_B,\a_B}\ \ \\
&=& \sum O_{\a'_A,\a_A}\sum\bar{d}^{[\a'_A,\b'_A]} d^{[\a_A,\b_A]} \g(O_B)_{\b'_A,\b_A}\ \ \\
\label{Eq: fundamental equation OAOBgamma, end}&=& \sum M^{k_A}_{\a_A,\b_A}\bar{M}^{k_A}_{\a'_A,\b'_A}O_{\a'_A,\a_A}\g(O_B)_{\b'_A,\b_A}
\eea
where $k_A$ is half of the number of end points on the boundary of disc $A$, $\bar{M}^{k_A}$ is the complex conjugate of $M^{k_A}$, and we have used the definition: $M^{k_A}_{\a_A,\b_A} = d^{[\a_A,\b_A]}$. To simplify our notation, we define a bilinear form on the space of operators with $2k_A$ end points as
\bea
\label{Eq: definition of inner product, appendix}\<O_A,O'_A\>\equiv M^{k_A}_{\a_A,\b_A}\bar{M}^{k_A}_{\a'_A,\b'_A}O_{\a'_A,\a_A}O'_{\b'_A,\b_A}.
\eea
With this definition, Eq.~\ref{Eq: fundamental equation OAOBgamma, start}-\ref{Eq: fundamental equation OAOBgamma, end} are expressed as
\bea
\<O_AO_B\>_\g = \<O_A,\g(O_B)\>.
\eea

Note that the linear space of operators with $2k_A$ end points, $LO_{k_A}$, is related to the linear space of all connections with $2k_A$ end points, $LC_{k_A}$, as $LO_{k_A} = LC^*_{k_A}\otimes LC_{k_A}$, where $LC^*_{k_A}$ denotes the dual space of $LC_{k_A}$. The bilinear form in Eq.~\ref{Eq: definition of inner product, appendix} defines a map $W^{k_A}$ from $LO_{k_A}$ to $LO^*_{k_A}$, whose matrix elements are $W^{k_A}_{\a_A\a'_A, \b_A\b'_A} \equiv M^{k_A}_{\a_A,\b_A}\bar{M}^{k_A}_{\a'_A,\b'_A}$. In compact notation $W^{k_A}=M^{k_A}\otimes\bar{M}^{k_A}$. The eigenvalues of $W^{k_A}$ are products of the eigenvalues of $M^{k_A}$ and $\bar{M}^{k_A}$; the eigenvectors of $W^{k_A}$ (in this case, eigen-operators) are the tensor products of the eigenvectors of $M^{k_A}$ and $\bar{M}^{k_A}$ (in this case, eigen-bras and eigen-kets). In particular, we have shown that $M^{k_A}$ is invertible for generic $d$, and this statement implies that $W^{k_A}$ is invertible for generic $d$. This result is important to the study of the operator spectrum.

Now we explore the physical meaning of the formal results above. First we want to count the number of independent watermelon operators. Recall that watermelon operators are defined in Eq.~\ref{Eq:watermelonOperator1}: (ignoring the spin index for a moment) an operator is defined as a watermelon operator if and only if its correlation function with any other operator vanishes when there is a self contact among its end points. Consider an operator $O_A$ that acts on disc $A$ with $2k_A$ end points, and an arbitrary operator acts on disc $B$ with $2k_B$ end points. Using the notation developed above, $O_A$ is a watermelon operator if and only if
\bea
\<O_A\g(O_B)\> = 0,\forall O_B,\nonumber\\ \forall\text{ $\g$ with self contact}.
\eea

Since the bilinear form is invertible for generic $d$, the number of linearly independent watermelon operators on $A$, $N_k$, is the dimension of $LO_{k_A}$ minus the dimension of the union of the images of all connections $\g$ with self-contact on disc $A$ (we denote this set by $SC_A$):
\bea
N_{k} = C^2_k - \text{dim}[\cup_{\g\in SC_A}\operatorname{Im}[\g]].
\eea
where 
\bea 
C_k = \frac{1}{k+1}\binom{2k}{k}
\eea
is the number of possible connections in a disc with $2k$ end points, which is known as the $k$th Catalan number. The seemingly complicated second term also has an intuitive definition. 
Since the strands of $\g$ are unchanged by any operator in $B$, if $\g$ has a self-contact on disc $A$, the end points enclosed by this self-contact must also connect to disc $A$.
These strands all become removable strands in $\g(O_B)$ (which are defined in Sec.~\ref{sec:topologicaloperatorclassification} as strands that are not reconnected and are not blocked by any strands being reconnected), no matter what $O_B$ we choose.
It is not hard to see that $\cup_{\g\in SC_A}\operatorname{Im}[\g]$ is the space of all operators on disc $A$ with $2k$ end points that have removable strands.
Thus, $N_k$, the number of watermelon operators with $2k$ end points, equals the number of $k$-loop reconnection operators.

For each $k$-loop reconnection operator, we can make it orthogonal to the space $\cup_{\g\in SC_A}\operatorname{Im}[\g]$ by subtracting an operator in that space. This is always possible when the bilinear form $W^k$ is invertible on the subspace $\cup_{\g\in SC_A}\operatorname{Im}[\g]$.  This is indeed true for generic $d$.\footnote{When $d\rightarrow\infty$, diagonal terms in $W^k$ dominate since they have the largest power of $d$; the inner product is proportional to a Kronecker delta function, hence invertible. In general, the determinant of $W^k$ restricted to this subspace is a polynomial of $d$. Since we have already shown that this polynomial is not identically zero, its zeros must be discrete.}
Thus each $k$-loop reconnection operator becomes a unique watermelon operator after adding operators that do not reconnect $k$ loops.

We can then organize these watermelon operators by their spins, as introduced in Sec.~\ref{sec:constructreconnectionoperators}. Since the topological type is just the $k$-loop reconnection operator up to rotation (Sec.~\ref{sec:topologicaloperatorclassification}), we can form exactly one spin-0 watermelon operator for each topological type. Their 2-point functions can be diagonalized for generic $d$. Thus we arrive at the conclusion that the number of degenerate scaling operators at dimension\footnote{$x_{2k}=x_{2k,0}$, see Eq.~\ref{Eq: x2k}  and Eq.~\ref{Eq:x2kl}.} $x_{2k}$
is just the number of distinct topological types for $k$-loop reconnection operators! 

At special $d$, when $M^k$ is not invertible, some watermelon operators become hidden operators. We would like to illustrate the general idea in the case of 3-loop reconnection operators at $d=\sqrt{2}$. In this case, $M^3$ annihilates the state $|\text{JW}_2\>$, hence $|\a_A\>\<\text{JW}_2|$ and $|\text{JW}_2\>\<\a_A|$ with any $\a_A$ are all annihilated by $W^3$ (the inner product). They have zero equal-time correlator with any other operator outside $A$. We say $O_A\sim O'_A$ if they differ by a hidden operator, as descibed in the main text. Using the definition of $|\text{JW}_2\>$ in Fig.~\ref{Fig: JW 3 loop}, we have the following relation
\bea
(\sqrt{2}(|\a_1\>+|\a_2\>) + |\b_1\>+|\b_2\>+|\b_3\>)\ \<\a_i| \sim 0,\nonumber\\ 
(\sqrt{2}(|\a_1\>+|\a_2\>) + |\b_1\>+|\b_2\>+|\b_3\>)\ \<\b_j| \sim 0,\nonumber\\
|\a_i\>\ (\sqrt{2}(\<\a_1|+\<\a_2|) + \<\b_1|+\<\b_2|+\<\b_3|) \sim 0,\nonumber\\ 
|\b_j\>\ (\sqrt{2}(\<\a_1|+\<\a_2|) + \<\b_1|+\<\b_2|+\<\b_3|) \sim 0\ \ 
\eea
These relations imply
\bea
|\a_1\>\<\a_2|\sim O(R_2),\\ 
|\b_j\>\<\b_{j+1}|\sim|\b_1\>\<\b_2| + O(R_2), \\ 
|\b_{j+1}\>\<\b_{j}|\sim -|\b_{1}\>\<\b_2| + O(R_2),\ \forall j,
\eea
where $O(R_2)$ stands for any operator that can at most reconnect 2 loops. These relations imply that all 3-loop reconnection operators except ${\sum_{j=1,2,3}|\b_j\>\<\b_{j+1}| - |\b_{j+1}\>\<\b_{j}|}$ are reduced to operators that do not reconnect 3 loops (up to hidden operators), which means the watermelon operators constructed from these operators are hidden operators. But they can still have nontrivial correlators in the time direction, as discussed in Sec.~\ref{sec:ffscalingoperators}.

\section{Ground state degeneracy on the torus} \label{Appendix: GSD}

In this appendix, we explain the existence of JW projectors at special $d$, introduced in Sec.~\ref{sec:gsd} and Sec.~\ref{Sec: Jones Wenzl}. We focus on the cases $d=\pm\sqrt{2}$, and discuss the ground state degeneracy after adding the JW projectors. Results discussed here were first reported in Ref~\cite{freedman2004class} using the representation of the $SU(2)_k$ Kac-Moody algebra on loop states. We provide an alternative, self-contained explanation.

For generic $d$, there does not exist any local reconnection operator that annihilates the state ${|\text{GS}\>= \sum_C d^{|C|}|C\>}$ (for local operators here, we ignore small loops in the sense of Eq.~\ref{Eq:operatorEquivalence}). This has to do with the nonlocality of the wavefunction $\psi(C) = d^{|C|}$. For example, consider the wavefunction amplitude for 3-strand configurations inside a disc shown in Fig~\ref{Fig: M matrix}. On the sphere, there are 5 possible connections inside and outside the disc. The wavefunction amplitude for states inside the disc depends on the total number of loops, hence on the connections outside the disc.  In the matrix $M^3$ shown in Fig.~\ref{Fig: M matrix}, each column, labeled by the connection $\b$ outside the disc, is proportional to the wavefunction for the degrees of freedom inside the disc that we get from the ground state if we fix the configuration outside to be $C_{\b}$, with connection $\b$:
\bea 
\<\text{GS}|\lf |\a\>\otimes|C_{\b}\> \ri = d^{|\a C_\b|} \propto d^{[\a,\b]}\equiv M^{k}_{\a,\b},
\eea
where $\a$ labels the connection inside the disc.
As described in Appendix~\ref{Appendix:topologicaloperatorclassification}, this matrix is invertible for generic $d$, which means the 5 column vectors form a complete basis of all 3-strand states: In order to annihilate the ground state, an operator must simultaneously annihilate all states inside the disc, hence must be identically zero.

\begin{figure}[t]
\begin{center}
\includegraphics[width=0.3\textwidth]{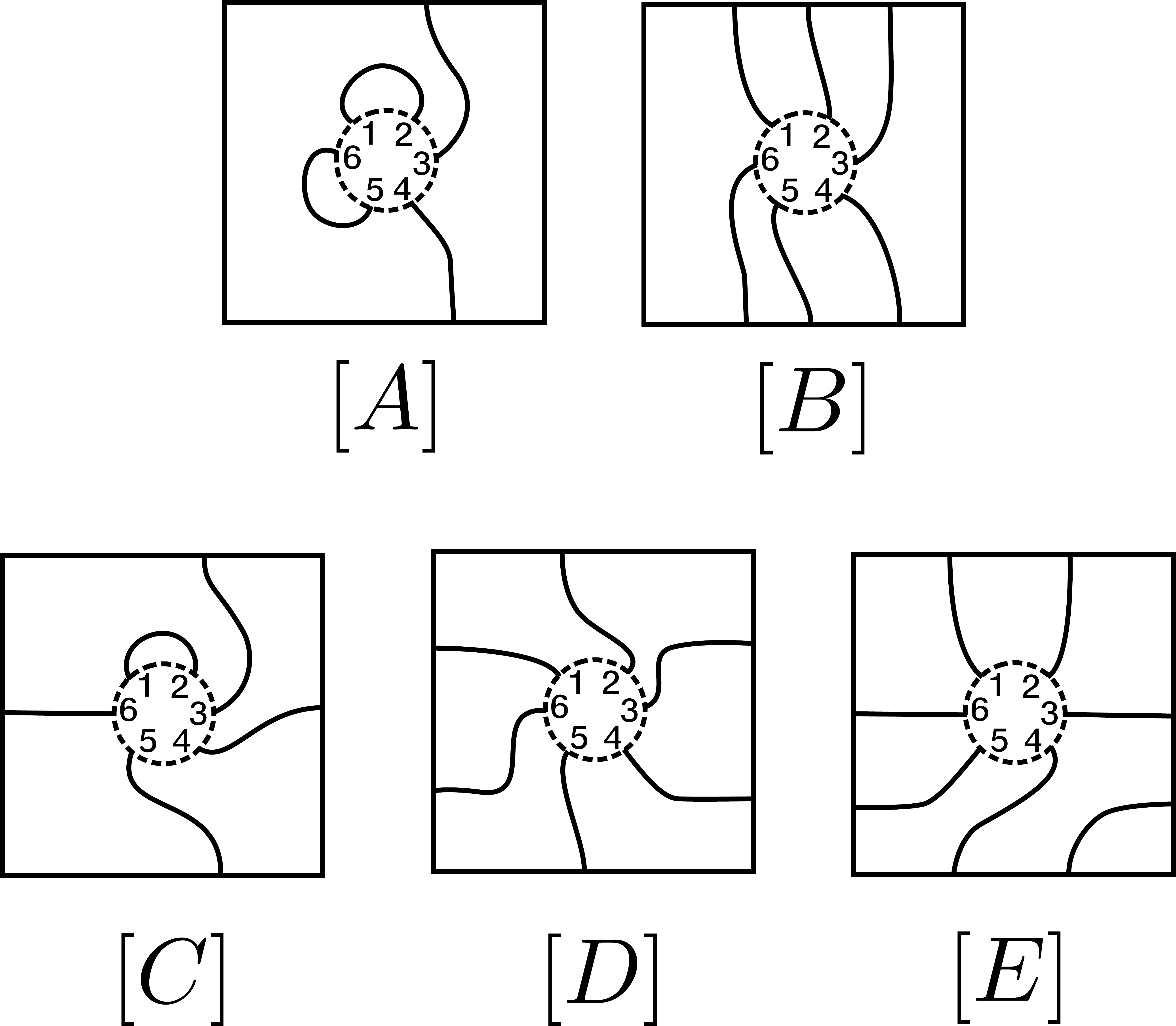}
\caption{JW on the torus}
\label{Fig:JW on torus}
\end{center}
\end{figure}

\begin{figure}[t]
\begin{center}
\includegraphics[width=0.45\textwidth]{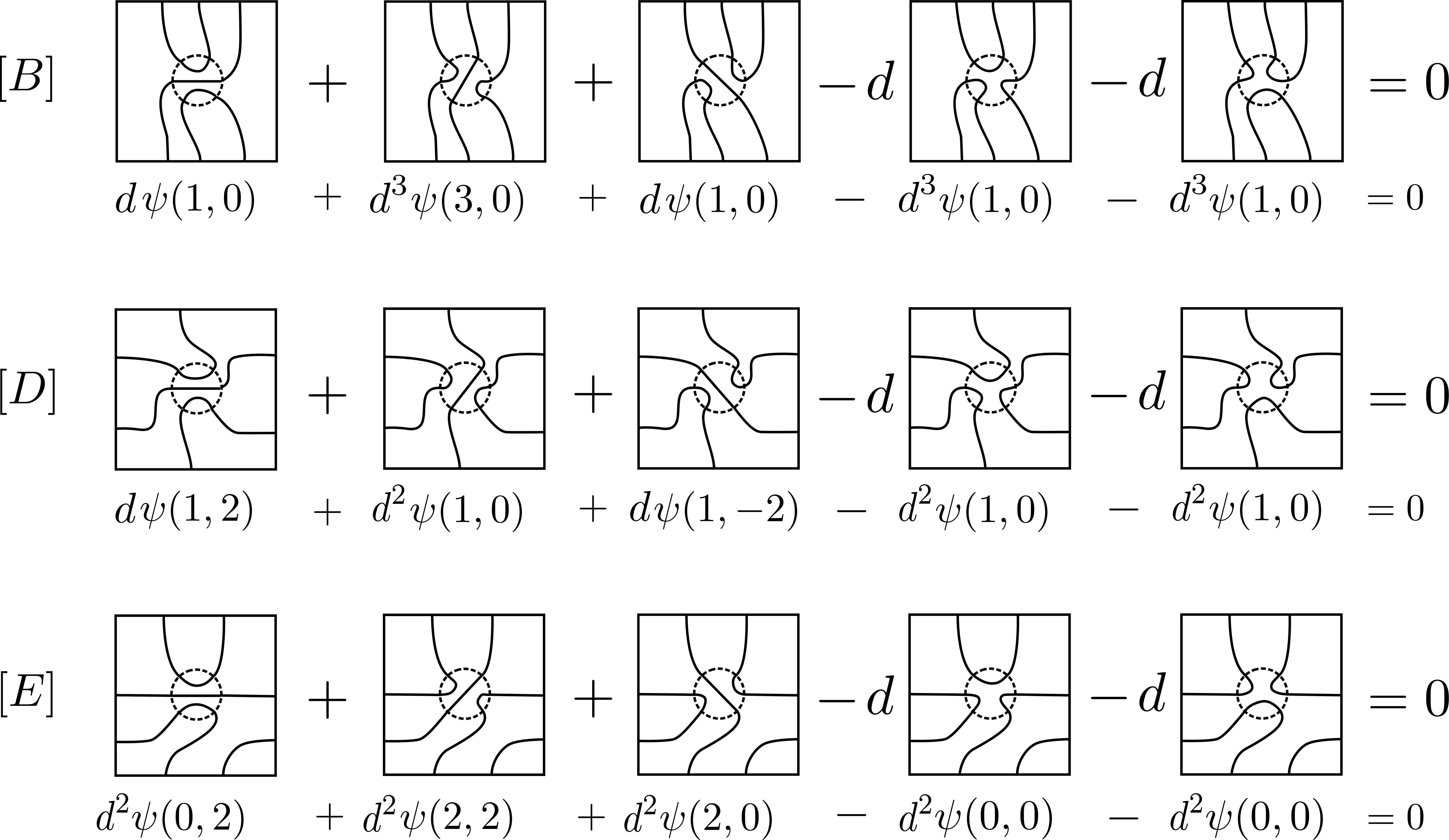}
\caption{JW constraints on the torus}
\label{Fig: JW constraints on torus}
\end{center}
\end{figure}

However, the determinant is zero for $d=0,\pm1,\pm\sqrt{2}$. For these special $d$, the 5 column vectors no longer span the total 5-dim Hilbert space, and there are \textit{hidden states} inside the disc that do not show up in the reduced density matrix of the ground state. The case $d=0$ is trivial. The cases $d=\pm1$ are explained in Sec.~\ref{sec:scalingforms}. At $d=\pm\sqrt{2}$, the hidden state is the state $|\text{JW}_2\>_{\pm}$ shown in Fig.~\ref{Fig: JW 3 loop}. It is easy to check that $|\text{JW}_2\>_{\pm}$ is orthogonal to all columns of $M^3_{\a,\a'}$. Because of this hidden state, there is a projector $|\text{JW}_2\>_{\pm}\<\text{JW}_2|_{\pm}$ that annihilates the ground state. We can add this projector to the Hamiltonian while preserving the ground state. More generally, at each $k\ge 2$, the matrix $M^{k}_{\a,\a'}$. is non-invertible if and only if $d=\cos(p\pi/(q+1)),q\le k$ (Ref~\cite{freedman2004class}). At these special weights, we have hidden states of $k$-strand configurations.

Now we focus on the case $k=3, d=\pm\sqrt{2}$, and study the ground states on the torus after adding the JW projector. On the torus, loop states can have various winding numbers. Without the JW projector, different winding sectors are degenerate. However, just like the 2-loop reconnection in the toric code, the 3-loop JW projector connects configurations in different sectors, and put additional constraints on the ground state wave function. On the torus, we can write the ground state wavefunction as $\psi(C)=\psi(n_x,n_y)d^{|C|}$, where $n_x$ and $n_y$ are the winding numbers. The second part $d^{|C|}$ is required by the loop creation/annihilation term, and is consistent with the corresponding JW projector when $d=\pm\sqrt{2}$. We want to solve the additional constraints on $\psi(n_x,n_y)$ to get the ground states, but a priori, it is not clear at all whether there are zero, finitely many, or infinitely many zero-energy states\footnote{If there is no zero-energy state on the torus, we must look for the ground states among the finite-energy states. Luckily, there are zero-energy states, and we do not need to solve the much harder problem.}. (We shall see later that the naive guess $\psi(n_x,n_y) = 1$ does not give a zero-energy state.) Unlike the cases $d=\pm1$, it is not a trivial task to write down all constraints on $\psi(n_x,n_y)$.

JW projectors connect different sectors, when the loops being reconnected wind non-trivially around the torus (See Fig~\ref{Fig:JW on torus} and Fig.~\ref{Fig: JW constraints on torus}). To find all constraints, we want to find all possible ways for the 3 loops to wind around the torus. The first thing to notice is that, there are more possible ways to connect the end points outside the disc on the torus than on the sphere. For the 3-loop JW projector, there are 6 end points on the boundary of the disc, and $15=5\times3$ ways to connect them (5 is the number of choices for the 1st point, 3 is the number of choices for the next unpaired point). Under cyclic permutation (roughly speaking, 60 degree rotation of the disc), these 15 ways fall into 5 equivalent classes, $[A],[B],[C],[D]$, and $[E]$, as shown in Fig.~\ref{Fig:JW on torus}. For example, there are 2 elements in class $[A]$: (12)(34)(56), and (16)(23)(45). 
However, only the 5 elements in class $A$ and class $B$ can be realized on the sphere without crossing. On the torus, we must consider 3 new classes, each of which has different realizations (in terms of winding numbers) on the torus.

To study all possible realizations of these classes on the torus, we need to clarify a few points about winding numbers on the torus. We first state 4 basic facts: (1) The winding number of each loop around the $x$ direction and the $y$ direction must be coprime, otherwise there must be a self-crossing. (2) If there are multiple loops that wind non-trivially, their winding numbers must be the same, otherwise there would be a crossing between them. (1) and (2) together state that the two winding numbers $n_x$ and $n_y$ uniquely determine the winding numbers of all nontrivial loops in a given configuration. The number of nontrivial loops must be the greatest common divisor (g.c.d.) of $n_x$ and $n_y$, and the winding number of each nontrivial loop must be the total winding number divided by this g.c.d.. (3) We can transform a nontrivial loop with arbitrary winding number to a nontrivial loop winding around the $y$ direction by modular transformations. Modular transformations on the torus are generated by $S$ ($(x,y)\rightarrow(y,-x)$) and $T$ ($(x,y)\rightarrow(x-y,y)$). (4) We need to use \textit{signed} winding numbers to distinguish the case $(n_x,n_y)=(1,1)$, with the case $(n_x,n_y)=(1,-1)$. We define $n_x$ and $n_y$ as follows: First, assign an arbitrary arrow to the nontrivial loops; the arrow should be the same for all nontrivial loops in the given configuration. Second, define $n_x$ as the number of strands passing through the horizontal line from below, minus the number of strands passing through the horizontal line from above. Define $n_y$ in a similar way in the other direction. Since this arrow is assigned arbitrarily, $(n_x,n_y)\simeq(-n_x,-n_y)$, but $(n_x,n_y)$ and $(n_x,-n_y)$ label distinct topological sectors.

With these preliminaries, we can write down all constraints on $\psi(n_x,n_y)$ from the JW projector. For each of the 5 classes, we enumerate cases where there are 1, 2 or 3 strands winding around the torus. In each case, use modular transformations to fix the winding number of each nontrivial strand. For each such realization, we can write down the constraint that the JW projector annihilates the the ground state
\bea 
|\text{JW}_2>_{\pm}\<\text{JW}_2|_\pm|\Psi\> = 0, \text{ for } d=\pm\sqrt{2}.
\eea 
In Fig~\ref{Fig: JW constraints on torus}, each square represents a wavefunction amplitude. It turns out that there are 3 nontrivial constraints up to modular transformations and reflection about the $x$ axis:
\bea
\label{Eq: JW constraints on torus, even, even}
\psi(2,2)+\psi(2,0)+\psi(0,2)-2\psi(0,0) = 0\\
\label{Eq: JW constraints on torus, odd, even}\psi(1,2) + \psi(1,-2) - d\psi(1,0) = 0\\
\label{Eq: JW constraints on torus, mulitplicity}\psi(0,3+n)=\psi(0,1+n),\ n\ge 0,
\eea
where $d=\pm\sqrt{2}$. The action of $S$ and $T$ and the reflection $R_x$ on wavefunction $\psi$ are defined such that
\bea
S\circ\psi(n_x,n_y) = \psi(n_y,-n_x),\\
T\circ\psi(n_x,n_y) = \psi(n_x,n_x+n_y),\\
R_x\circ\psi(n_x,n_y) = \psi(-n_x,n_y)
\eea
Modular transformations acting on the 3 nontrivial constraints generate infinitely many constraints on the infinitely many complex numbers $\psi(n_x,n_y)$. Now we want to understand and solve these constraints to find the ground states on the torus. Note that modular transformations do not change the g.c.d. of $n_x$ and $n_y$. $n_x$ and $n_y$ in Eq.~\ref{Eq: JW constraints on torus, even, even} has $\operatorname{g.c.d.}$ 2; $n_x$ and $n_y$ in Eq.~\ref{Eq: JW constraints on torus, odd, even} has g.c.d. 1. These two equations generate all constraints on $\psi(n_x,n_y)$ with $\operatorname{g.c.d.}(n_x,n_y) = 1,2$. On the other hand, after modular transformations, Eq.~\ref{Eq: JW constraints on torus, mulitplicity} simply says 
\bea
\psi(mn_x,mn_y) = \psi(n_x, n_y),\\ \forall n_x, n_y\in 2\mathbb{Z}, \forall m,\\
\text{ and }\psi(mn_x,mn_y) = \psi(n_x, n_y),\\ \forall\text{ odd }m,\forall (n_x,n_y) = 1.
\eea 
The parities of the winding numbers $n_x$ and $n_y$ are still good quantum numbers, we have four sectors: (even, even), (odd, odd), (even, odd) and (odd, even). The last 3 sectors are related by modular transformations, so we only need to consider 2 cases: (even, even) and (odd, odd). In the (even, even) sector, we need to consider only those $(n_x, n_y)$ with g.c.d. 2, and solve constraints generated by Eq.~\ref{Eq: JW constraints on torus, even, even}; in the (odd, odd) sector, we need to consider only those $(n_x, n_y)$ with g.c.d. 1, and solve constraints generated by Eq.~\ref{Eq: JW constraints on torus, odd, even}. All other amplitudes are fixed by Eq.~\ref{Eq: JW constraints on torus, mulitplicity}.

The degeneracy in the (even, even) sector has its own meaning: it is the total ground state degeneracy in the domain-wall interpretation of the loop model. This is because the spin on a plaquette always goes back to itself after traveling along the nontrivial cycles of the torus, and the number of domain walls encountered along the cycle must be even. 

Constraints generated by Eq.~\ref{Eq: JW constraints on torus, even, even} are
\bea
\psi(2m+2r, 2n+2s) + \psi(2m, 2n) + \psi(2r, 2s)\nonumber\\ = 2\psi(0,0),\ \forall
\left(\begin{array}{cc}
m & r\\
n & s
\end{array}\right)
\in SL_2(\mathbb{Z})\ 
\eea
A simple solution is ${\psi(0,0) = 1}$, ${\psi(2m,2n) = \frac{2}{3}}$, ${\forall (m,n)\neq(0,0)}$. For every other solution, we can always subtract this solution to make $\psi(0,0) = 0$. After playing with these equations, it is not hard to see that every amplitude reduces to the amplitudes $\psi(2,0), \psi(0,2)$, and $\psi(0,0)$. Thus, apart from the previous solution, we have at most 2 other solutions. By explicit construction, we found 2 other solutions
\bea
\psi(2m,2n) =
\left\{ \begin{array}{ll}
1 & (m,n) = \text{(odd,odd)}\\
\omega & (m,n)=\text{(even,odd)}\\
\omega^{-1} & (m,n) = \text{(odd,even)}
\end{array} \right.
\eea
where $\omega = e^{\pm2\pi i/3}$. These 2 solutions and the previous solution $\psi(0,0) = 1.\ \psi(2m,2n) = \frac{2}{3},\forall (m,n)\neq(0,0)$ are the only 3 linearly independent solutions. Thus the ground state degeneracy in the (even, even) sector is 3; the total ground state degeneracy in the domain-wall interpretation of the loop model is 3.

Next, we solve the ground states in the (odd, odd) sector. All constraints in this sector are generated by Eq.~\ref{Eq: JW constraints on torus, odd, even}, or equivalently by a modular transformation of Eq.~\ref{Eq: JW constraints on torus, odd, even}:
\bea
\label{Eq: JW constriant, odd, odd}\psi(1,-1) + \psi(1,3) = d\psi(1,1)
\eea
It is not hard to see that all amplitudes are fixed by $\psi(1,1)$ and $\psi(1,-1)$. Thus, there are at most 2 solutions in this sector. Indeed, we find 2 independent solutions $\psi_A$ and $\psi_B$: 
\bea
\psi_A(1,1) = 0, \psi_A(1,-1) = 1,\nonumber\\ 
\label{Eq: psi B, definition, odd, odd}\psi_B(1,1) = 1, \psi_B(1,-1) = 0.
\eea 
It is not easy to explicitly write down the expression for $\psi_A(n_x,n_y)$ and $\psi_B(n_x,n_y)$, but it is easy to write down the representation of the modular group on the 2-dim solution space. The representation of $S$ and $T^2$ is given by
\bea
\label{Eq: rep of modular group, odd, odd}\left(\begin{array}{c}
S\circ\psi_A\\
S\circ\psi_B
\end{array}\right) = 
\left(\begin{array}{cc}
0 & 1\\
1 & 0
\end{array}\right)
\left(\begin{array}{c}
\psi_A\\
\psi_B
\end{array}\right), \\
\left(\begin{array}{c}
T^2\circ\psi_A\\
T^2\circ\psi_B{}
\end{array}\right) = 
\left(\begin{array}{cc}
0 & -1\\
1 & d
\end{array}\right){}
\left(\begin{array}{c}
\psi_A\\{}
\psi_B
\end{array}\right)
\eea
We ignore the modular transformation $T$, since it does not preserve the (odd, odd) sector. Eq.~\ref{Eq: psi B, definition, odd, odd} and Eq.~\ref{Eq: rep of modular group, odd, odd} completely fixes all amplitudes $\psi(n_x,n_y)$ in the (odd, odd) sector. In fact, Eq.~\ref{Eq: psi B, definition, odd, odd} is just a choice of basis. Eq.~\ref{Eq: rep of modular group, odd, odd} is required by the choice of basis, and by the constraint in Eq.~\ref{Eq: JW constriant, odd, odd}. It is straightforward to check the consistency of Eq.~\ref{Eq: psi B, definition, odd, odd}, Eq.~\ref{Eq: rep of modular group, odd, odd}, and Eq.~\ref{Eq: JW constriant, odd, odd}. Thus, there are 2 ground states in the (odd, odd) sector, 2 ground states in the (odd, even) sector, and 2 ground states in the (even, odd) sector. The solutions in the last two sectors are the modular transformed $\psi_A$ and $\psi_B$. 

As a conclusion, we confirm the results that there are 9 ground states on the torus for the lattice Hamiltonian with JW projector. These 9 states originate from different windings of loop configurations, their degeneracy is not lifted by the JW projector, and can only be lifted by irrelevant operators within the universality class. In general, the splitting of these sub-spectrum states are smaller than the lowest excitation energy in the continuous spectrum by a power of the system size. In the case of $d=\pm\sqrt{2}$, there is a marginally irrelevant operator; therefore the splitting between these state may be only logarithmically smaller than that of excited states in the continuum spectrum.

\section{Dilute critical points}
\label{appendix:dilutecriticalpoint}

The critical value of the weight per unit length in the classical lattice model is \cite{nienhuis1982exact}
\be 
x_c = \lf 2+ (2-|d|^2)^{1/2}  \ri^{-1/2}.
\ee 
A single line of renormalization group fixed points governs both the dense and dilute universality classes, as illustrated schematically in Fig.~\ref{fig:densedilutecartoon}: the line of dense fixed points continues into the line of dilute fixed points at $|d|^2=2$. The critical exponents vary continuously along this line.

\begin{figure}[t]
\begin{center}
\includegraphics[width=0.3\textwidth]{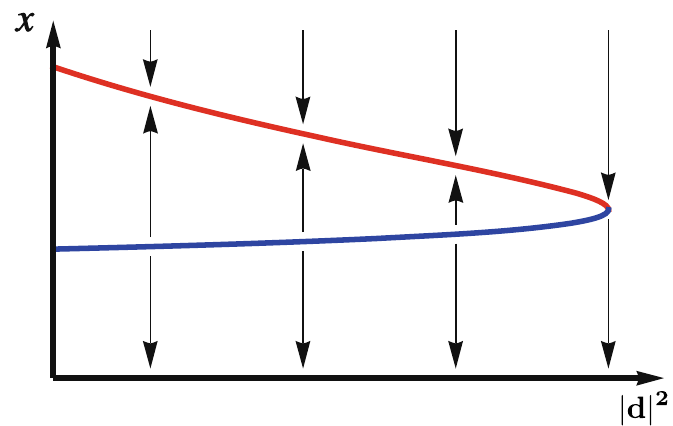}
\caption{RG fixed line governing the dense and dilute classical loop model. The dense fixed line (red) is stable as a classical ensemble, and the dilute fixed line unstable, with respect to varying $x$.}
\label{fig:densedilutecartoon}
\end{center}
\end{figure}

The quantum model may  be  generalized to nonzero $x$ by modifying the flip operator in Eq.~\ref{eq:flippable} and Eq.~\ref{Eq: H gapless loop}. 
The surface of `dense' quantum critical points then continues, at $|d|^2=2$, into another surface of `dilute' quantum critical points at $|d|^2<2$. 
The scaling dimensions in this phase follow from the results in the previous sections, if we replace the dense critical exponents with the known values in the dilute phase \cite{jacobsen2009conformal}. For example, the reconnection operators have larger scaling dimensions.

We can apply our bound on $z$, in terms of the fractal dimension (Sec.~\ref{sec:zbound}), anywhere on this enlarged critical surface.
The strongest bound is in the limit $d\rightarrow 0$ on the dilute branch, where $d_f \rightarrow 4/3$. In this limit the bound becomes $z\geq 2.6\dot 6$. 

By the superuniversality argued for in Sec.~\ref{sec:superuniversality} (assuming the surface of critical points in the lattice model corresponds in the usual way to a surface of RG fixed points) we can then apply this bound everywhere on the critical surface, including for the dense loops, as stated in the introduction.

\bibliographystyle{apsrev4-1}
\bibliography{loopref.bib}

\end{document}